\documentclass{amsart}

\usepackage{amssymb,latexsym,amsmath,bbm}
\usepackage{graphicx}
\usepackage[all]{xypic}

\newtheorem{theorem}{Theorem}[section]

\newtheorem{lemma}[theorem]{Lemma}
\newtheorem{corollary}[theorem]{Corollary}
\newtheorem{fact}[theorem]{Fact}

\newtheorem{defi}[theorem]{Definition}
\newtheorem{rema}[theorem]{Remark}
\newtheorem{exam}[theorem]{Example}
\newtheorem{obse}[theorem]{Observation}
\newtheorem{conv}[theorem]{Convention}

\newenvironment{definition}{\begin{defi}\rm}{\end{defi}}
\newenvironment{remark}{\begin{rema}\rm}{\end{rema}}
\newenvironment{example}{\begin{exam}\rm}{\end{exam}}

\newenvironment{convention}{\begin{conv}\rm}{\end{conv}}

\newtheorem{claim2}{Claim}
\newenvironment{claim}{\begin{claim2}\rm}{\end{claim2}\rm}
\newenvironment{claimfirst}{\setcounter{claim2}{0}
\begin{claim2}\rm}{\end{claim2}\rm}
\newenvironment{pfclaim}{\begin{trivlist}\item[]{\it Proof of
Claim}}{\hfill \end{trivlist}}

\newcommand{\class}[1]{\mathsf{#1}}
\newcommand{\Set}{\mathsf{Set}}
\newcommand{\Rel}{\mathsf{Rel}}

\newcommand{\T}{T}           %
\newcommand{\Tom}{T_{\om}}   %
\renewcommand{\P}{P}         %
\newcommand{\Q}{\breve{P}}         %
\newcommand{\Pom}{\P_{\om}}  %
\newcommand{\V}{V} %

\newcommand{\rl}[1]{\overline{#1}}
\newcommand{\Tb}{\rl{\T}}
\newcommand{\Pb}{\rl{\P}}

\newcommand{\nbsem}{\lambda^{\T}}
\DeclareMathOperator{\K}{K}

\newcommand{\id}{\mathit{id}}
\newcommand{\Id}{\Delta}
\DeclareMathOperator{\Graph}{\mathit{Gr}}
\newcommand{\single}{\eta}
\DeclareMathOperator{\Base}{\mathit{Base}}
\newcommand{\counit}{i} %

\newcommand{\struc}[1]{\langle #1 \rangle}

\newcommand{\wi}{\eqslantless}

\newcommand{\two}{\Omega}

\newcommand{\Twi}{\rel{\Tb {\wi}}}

\newcommand{\bbL}{\mathbb{L}}
\newcommand{\bbM}{\mathbb{M}}

\newcommand{\VT}{\V_{\T}}
\newcommand{\VP}{\V_{\Pom}}
\newcommand{\Tleq}{\rel{\Tb {\leq}}}
\newcommand{\Pleq}{\rel{\Pb {\leq}}}
\newcommand{\Tin}{\rel{\Tb {\in}}}
\newcommand{\Pin}{\rel{\Pb {\in}}}

\DeclareMathOperator{\Fr}{Fr}
\newcommand{\SupLat}{\mathop{\mathrm{SupLat}}}
\newcommand{\downarr}{{\textstyle \mathop{\downarrow}}}
\newcommand{\uparr}{{\textstyle \mathop{\uparrow}}}

\renewcommand{\lesssim}{\sqsubseteq}
\renewcommand{\phi}{\varphi} %
\newcommand{\smallvee}{{\textstyle \bigvee}}
\newcommand{\smallwedge}{{\textstyle \bigwedge}}
\newcommand{\cov}{\triangleleft}
\newcommand{\smallcup}{{\textstyle \bigcup}}
\newcommand{\smallcap}{{\textstyle \bigcap}}

\newcommand{\bw}{\bigwedge}
\newcommand{\bv}{\bigvee}
\newcommand{\bvdir}{\bigvee\nolimits^{\uparrow}}
\newcommand{\bwsmall}{\smallwedge}
\newcommand{\bvsmall}{\smallvee}

\newcommand{\nb}{\nabla}
\newcommand{\hs}{\heartsuit}

\newcommand{\SRD}{\mathit{SRD}}

\newcommand{\nada}{\varnothing}
\newcommand{\sse}{\subseteq}
\newcommand{\rst}[1]{{\upharpoonright_{#1}}}
\newcommand{\isdef}{\mathrel{::=}}

\newcommand{\Dom}{\mathsf{dom}}
\newcommand{\Ran}{\mathsf{rng}}
\newcommand{\cv}[1]{#1\breve{\hspace*{1mm}}}
\newcommand{\cor}{\mathbin{;}}
\newcommand{\cof}{\mathbin{\circ}}
\newcommand{\rel}[1]{\mathrel{#1}}

\newcommand{\al}{\alpha}
\newcommand{\be}{\beta}
\newcommand{\ga}{\gamma}
\newcommand{\de}{\delta}

\newcommand{\si}{\sigma}
\newcommand{\om}{\omega}
\newcommand{\Ga}{\Gamma}

\newcommand{\Fid}{\mathit{Id}}
\newcommand{\CF}{\mathcal{C}}
\newcommand{\cA}{\mathcal{A}}
\newcommand{\cB}{\mathcal{B}}
\newcommand{\cC}{\mathcal{C}}
\renewcommand{\cH}{\mathcal{H}}

\newcommand{\nta}{\rho}

\newcommand{\wh}[1]{\widehat{#1}}
\newcommand{\VTp}{\V_{\T'}}
\newcommand{\Tpleq}{\rel{\rl{\T'}{\leq}}}
\newcommand{\Tpin}{\rel{\rl{\T'}{\in}}}

\usepackage{natbib}
\usepackage{setspace}

\title[Generalized powerlocales via relation lifting]{Generalized
powerlocales via relation lifting}
\thanks{
The second author is grateful for financial support from the Netherlands Organization for Scientific Research
(NWO) enabling him to make two research visits to Amsterdam.
The work of the
third author was made possible by NWO VICI grant
639.073.501.}
\author{Yde Venema}
\address{Institute for Logic, Language and Computation, University of
  Amsterdam, Postbus 94242, 1090 GE Amsterdam, The Netherlands}
\email{y.venema@uva.nl}

\author{Steve Vickers}
\address{School of Computer Science, The University of Birmingham,
Birmingham, B15 2TT, United KIngdom}
\email{s.j.vickers@cs.bham.ac.uk}

\author{Jacob Vosmaer}
\address{Institute for Logic, Language and Computation, University of
  Amsterdam, Postbus 94242, 1090 GE Amsterdam, The Netherlands}
\email{contact@jacobvosmaer.nl}

\date{21 December 2010; Revised 6 February 2012}

\begin{document}
\maketitle

\date{\today}

\begin{abstract}
This paper introduces an endofunctor $\VT$ on the category of frames,
parametrized by an endofunctor $\T$ on the category $\Set$ that satisfies certain
constraints.
This generalizes Johnstone's construction of the Vietoris powerlocale, in the
sense that his construction is obtained by taking for $\T$ the finite
covariant power set functor.
Our construction of the $\T$-powerlocale $\VT \bbL$ out of a frame $\bbL$
is based on ideas from coalgebraic logic and makes explicit the connection
between the Vietoris construction and Moss's coalgebraic cover modality.

We show how to extend certain natural transformations between set functors
to natural transformations between $\T$-powerlocale functors.
Finally, we prove that the operation $\VT$ preserves some properties of frames,
such as regularity, zero-dimensionality, and the combination of
zero-dimensionality and compactness.

\end{abstract}
\section{Introduction}
\label{Section:Intro}

The aim of this paper\footnote{This paper has been accepted for publication in
  \emph{Mathematical Structures in Computer Science}.} is to show how coalgebraic modal logic can be used
to understand, study and generalize the point-free topological construction
of taking Vietoris powerlocales.

\subsection{Hyperspaces and powerlocales}
The \emph{Vietoris hyperspace construction} is a topological
construction on compact Hausdorff spaces, which was introduced by \citet{Vietoris1922} as a generalization of the
Hausdorff metric. Given a topological space $X$ one
defines a new topology $\tau_{X}$ on $\K X$, the
set of compact subsets of $X$. This new topology
$\tau_X$ has as its basis all sets of the form \[ \nabla \{
U_1, \dotsc, U_n \} := \{ F \in \K X\mid F \subseteq
\smallcup_{i=1}^n U_i \text{ and } \forall i \leq n,\, F \between
U_i\},\]
where $U_1,\dotsc,U_n \subseteq X$ is a finite collection of open sets
and $F \between U$ is notation to indicate that $F \cap U \neq \emptyset$.
Alternatively, one
can use a subbasis to generate $\tau_X$, consisting of
subbasic open sets of the shape
\begin{align*}
&\Box  U := \{F \in \K X\mid F \subseteq U\},
\intertext{and}
&\Diamond  U := \{ F \in \K X\mid F \between U\}.
\end{align*}
To generate the basic open sets $\nabla \{ U_1,\dotsc, U_n\}$ from
$\Box U$ and $\Diamond U$, one can use the following expression: \[
\nabla \{ U_1, \dotsc, U_n \} = \Box \left(\smallcup_{i=1}^n U_i
\right) \cap \smallcap_{i=1}^n \Diamond U_i.\]

In the field of \emph{point-free topology}, a considerable amount of
general topology has been recast in a way which makes it more
compatible with constructive mathematics and topos theory.  (Standard
references are \citet{Johnstone1982} and \citet{Vickers1989}).  The
main idea is to study the lattices of open sets of topological spaces,
rather than their associated sets of points.  In other words, it is an
approach to topology via algebra, where rather than categories of
topological spaces, one studies categories of \emph{locales}, or their
algebraic counterparts, \emph{frames}.  Frames are complete lattices
in which finite meets distribute over arbitrary joins, and can be seen
as the algebraic models of propositional geometric logic, a branch of
logic where finite conjunctions are studied in combination with
infinite disjunctions.  %
Substantial parts of this paper arose out of the direct application of
techniques from coalgebraic logic to frames/locales. This has led to
two consequences: firstly, most results are stated in terms of frames
rather than locales, since frames are closer to the Boolean algebras
predominantly used in coalgebraic logic. Secondly, we have given
little pause to issues of constructivity, in order to be able to
directly apply coalgebraic logic techniques. We will briefly revisit
these matters in \S \ref{s:fw}. Our bias towards frames
notwithstanding, we have favored the name  `powerlocale' over `powerframe' however.

\citet{Johnstone1982} defines a point-free, syntactic version of the
Vietoris powerlocale, using an extension of geometric logic with two
unary operators, $\Box$ and $\Diamond$.  However he soon also
introduces expressions of the shape
\[\Box (\smallvee A) \wedge \smallwedge_{b \in B} \Diamond b,\]
where $A$ and $B$ are finite sets, which should remind the reader of
the expression for $\nabla \{ U_1,\dotsc, U_n \}$ above.
Nevertheless, the description of the Vietoris powerlocale is usually
given with $\Box$ and $\Diamond$ as primitive, and not without good
reason: one may obtain the Vietoris powerlocale by first constructing
one-sided locales corresponding to the $\Box$-generators on the one
hand and the $\Diamond$-generators on the other, and then joining
these two one-sided powerlocales to obtain the Vietoris powerlocale
\citep{VT2004}.  The question remains however, if one can describe the
Vietoris powerlocale directly in terms of its basic opens,
corresponding to $\nabla\{ U_1, \dotsc, U_n\}$, rather than the
subbasic opens expressed in terms of $\Box$ and $\Diamond$.  One of
the main contributions of this paper is to show that this is indeed
possible.

\subsection{The cover modality and coalgebraic modal logic}
The reader may have noticed that the notation using $\Box$ and $\Diamond$
above is highly suggestive of modal logic.
This is no  coincidence: Johnstone's presentation of the Vietoris powerlocale
in terms of generators and relations extends the axioms of positive
(that is, negation-free) modal logic to the geometric setting.

In Boolean-based modal logic, one can define a $\nabla$-modality which is
applied to finite \emph{sets} of formulas.
This $\nabla$-modality then has the following semantics.
If $\mathfrak{M} = \langle W,R,V\rangle$ is a Kripke model and $\alpha$ is a
finite set of formulas, then for any state $w \in W$,
\begin{align*}
\mathfrak{M} ,w \Vdash \nabla \alpha \text{ iff }& \forall a \in
\alpha,\, \exists v \in R[w],\, \mathfrak{M}, v\Vdash a \text{ and
}\\
& \forall v \in R[w],\, \exists a \in \alpha,\, \mathfrak{M}, v\Vdash a.
\end{align*}
In classical modal logic, the $\nabla$-modality is equi-expressive
with the $\Box$- and $\Diamond$-modalities, using the following
translations:
\[ \nabla \alpha \equiv \Box \left( \smallvee \alpha \right) \wedge
\smallwedge_{a \in \alpha} \Diamond a,\]
and in the other direction, one can use
\[ \Box a \equiv \nabla \{ a\} \vee \nabla \emptyset,\text{ and }
\Diamond a \equiv \nabla \{ a, \top\}.\]
As a primitive modality, $\nabla$ was first introduced by
\citet{BM1996} in the study of circularity and by
\citet{JW1995} in the study of the modal $\mu$-calculus.
It was in Moss's work \citep{moss:coal99} however that the $\nabla$-modality
stepped into the spotlight as a modality suitable for generalization to the
abstraction level of \emph{coalgebras}.

The theory of Coalgebra aims to provide a general mathematical framework
for the study of state-based evolving systems.
Given a endofunctor $\T$ on the category $\Set$ of sets with functions, a
coalgebra of type $\T$, or briefly: a $\T$-coalgebra is simply a function
$\sigma \colon X \to \T X$, where $X$ is the underlying set of states of the
coalgebra, and a $\T$-coalgebra morphism between coalgebras $\sigma \colon X
\to \T X$ and $\sigma \colon X' \to \T X'$ is simply a function $f \colon X
\to X'$ such that $\T f \cof \sigma = \sigma' \cof f$.
\citet{Aczel1988} introduced $\T$-coalgebras as a means to study
transition systems.
A natural example of such transition systems is provided by the \emph{Kripke
frames} and \emph{Kripke models} used in the model theory of propositional
modal logic: the category of Kripke frames and bounded morphisms is isomorphic
to the category of $\P$-coalgebras, where $\P\colon \Set \to \Set$ is the
covariant powerset functor.
Universal coalgebra was later introduced by \citet{Rutten2000} as a
theoretical framework for modeling behavior of set-based transition systems,
parametric in their transition functor $\T \colon \Set \to \Set$.

Coalgebraic logics are designed and studied in order to reason formally about
coalgebras and their behavior; one of the main applications of this approach
is the design of specification and verification languages for coalgebras.
The most influential approach to coalgebraic logic, known as \emph{coalgebraic
modal logic} \citep{cirs:moda09},
is to try and generalize propositional modal logic from Kripke structures to
the setting of arbitrary set-based coalgebras.
Seminal in this approach was the observation of L.~Moss in the earlier
mentioned paper \citep{moss:coal99}, that the semantics of the cover modality
$\nabla$ can be described using the categorical technique of \emph{relation
lifting}.
This observation paved the way for generalizations to other functors that
admit a reasonable notion of relation lifting: Moss introduced a modality
$\nabla_{\T}$, parametric in the transition type functor $\T$, which can
be interpreted in $\T$-coalgebras via relation lifting.

While Moss's perspective was entirely semantic, his work naturally raised
the question whether good derivation systems could be developed for the
coalgebraic cover modality $\nabla_{\T}$, parametric in the coalgebra functor
$\T$.
Building on earlier work by B\'{\i}lkov\'a, Palmigiano \& Venema
\citep{PV2007,BPV2008} for the power set case, Kupke, Kurz \& Venema
\citeyearpar{KKV2008,kupk:comp10}
proved soundness and completeness of such a derivation system $\mathbf{M}_{\T}$.
The latter paper also introduces, on the category of Boolean algebras, an
associated functor $\bbM_{\T}$, which can be seen as the algebraic correspondent
of the topological Vietoris functor on the dual category of Stone spaces.

\subsection{Contribution}
In this paper we translate the coalgebraic modal derivation system
$\mathbf{M}_{\T}$ from its Boolean origins \citep{BPV2008,KKV2008} to the
setting of geometric logic.
Basically, this means we take some first steps towards developing a
\emph{geometric} coalgebraic modal logic, i.e.~a logic with finite
conjunctions, infinite disjunctions, and the coalgebraic cover modality
$\nabla_{\T}$.

The main conceptual contribution of this paper is the introduction of a
\emph{generalized powerlocale construction} $\VT$, parametric in a functor
$\T\colon  \Set \to \Set$ satisfying some categorical conditions.
Given a frame $\bbL$, we define its $\T$-powerlocale $\VT\bbL$ using a
presentation, which takes the set $\{ \nb_{\T} \al \mid \al \in \T L \}$ as
generators and the geometric version of the $\nb$-axioms as relations.

As we will see, the classical Vietoris powerlocale construction is an
instantiation of the $\T$-powerlocale, where we take $\T = \Pom$, the
covariant finite power set functor.
This reveals that the connection between the Vietoris construction and the
cover modality, which was implicit in \emph{semantic} form already in \citep{Vietoris1922}, can also be made explicit
\emph{syntactically} using coalgebraic modal logic.
Our approach shows how to describe the Vietoris constructions syntactically
using the $\nabla$-expressions as primitives, rather than as expressions
derived from $\Box$- and $\Diamond$-primitives, as it was introduced in
\citep{Johnstone1982}.

In addition, we prove some technical results concerning the $\T$-powerlocale
construction.
To start with, we discuss some \emph{functorial properties}; in particular,
we show that we are in fact dealing with a \emph{functor}
\[
\VT\colon  \Fr \to \Fr
\]
on the category of frames with algebraic frame homomorphisms.
Furthermore, we show how to extend certain natural transformations between
transition functors to natural transformations between $\T$-powerlocale
functors; this generalizes for instance the frame homomorphism from the
Vietoris locale onto the original frame.
We also give an alternative \emph{flat site presentation} of the
$\T$-powerlocale $\VT \bbL$, showing that each element of a $\T$-powerlocale
has a disjunctive normal form.
Finally, we prove some first \emph{preservation results}; in particular, we
show that the operation $\VT$ preserves some important properties of frames,
such as regularity, zero-dimensionality, and the combination of
zero-dimensionality and compactness.
\medskip

\paragraph{\textit{Overview}}
In \S \ref{Section:Prelim} we introduce preliminaries on category theory,
relation lifting, frame presentations and the classical point-free
presentation of the powerlocale.
In \S \ref{Section:VT} we introduce the $\T$-powerlocale construction $\VT$.
We then show that the $\Pom$-powerlocale is isomorphic to the classical
Vietoris powerlocale and we discuss some functorial properties of the
construction.
We conclude this section with providing the above-mentioned flat site
presentation of $\T$-powerlocales.
In \S \ref{Section:Preservation} we prove our preservation results,
and we provide a new, constructively valid proof of the preservation
of compactness for the ``classical'' Vietoris construction.
We finish in \S \ref{s:fw} with some possibilities for future work.

\paragraph{\textit{Acknowledgments}}
We would like to thank the anonymous referee for providing useful
suggestions for improving the presentation of the paper.

\section{Preliminaries}
\label{Section:Prelim}
\subsection{Basic mathematics}
\label{Subsect:BasicMath}

First we fix some mathematical notation and terminology.
Let $f\colon X \to X'$ be a function. Then the \emph{graph} of $f$ is the relation
\[
\Graph{f} \isdef \{ (x,f(x)) \in X \times X' \mid s \in X \}.
\]
Given a relation $R \sse X \times X'$, we denote the \emph{domain} and
\emph{range} of $R$ by $\Dom(R)$ and $\Ran(R)$, respectively.
Given subsets $Y \sse X$, $Y \sse X'$, the \emph{restriction} of $R$ to
$Y$ and $Y'$ is given as
\[
R\rst{Y \times Y'} \isdef R \cap (Y \times Y').
\]
The composition of two relations $R \sse X \times X'$ and $R' \sse X'
\times X''$ is denoted by $R\cor R'$, whereas the composition of two functions
$f\colon  X \to X'$ and $f'\colon  X' \to X''$ is denoted by $f'\circ f$.
Thus, we have $\Graph{(f' \circ f)} = \Graph{f}\cor\Graph{f'}$.

We will denote by $\P(X)$ and $\Pom(X)$ the \emph{power set} and \emph{finite
power set} of a given set $X$.
The \emph{diagonal on} $X$ is the relation $\Id_{X} = \{ (x,x) \mid x \in X\}$.
Given two sets $X,Y$ we say that $X$ \emph{meets} $Y$, notation: $X\between Y$,
if $X\cap Y$ is inhabited (that is, non-empty).

A \emph{pre-order} is a pair $(X,R)$ where $R$ is a reflexive and
transitive relation on $X$.
Given such a pre-order  we define the operations $\downarr_{(X,R)},
\uparr_{(X,R)}\colon
\P X \to \P X$ by $\downarr_{(X,R)}(Y) := \{ x \in X \mid x \rel{R} y \text{ for some }
y \in Y \}$ and $\uparr_{(X,R)}(Y) := \{ x \in X \mid y \rel{R} x \text{ for some }
y \in Y \}$.
If no confusion is likely, we will write $\downarr_{X}$ or $\downarr$
rather than $\downarr_{(X,R)}$.
\newcommand{\bbN}{\mathbb{N}}
\newcommand{\EKPF}{\mathit{EKPF}}
\newcommand{\B}{\mathit{B}}

\subsection{Category theory}

We will assume familiarity with the basic notions from category theory,
including those of categories, functors, natural transformations, and
(co-)monads.
As a reference text the reader may consult for instance
\citet{macl:cate98}.

We let $\Set$ denote the category with sets as objects and functions as
morphism; endofunctors on this category will simply be called \emph{set
functors}.
The most important set functor that we shall use is the covariant power set
functor $\P$, which is in fact (part) of a monad $(\P,\mu,\eta)$, with
$\eta_{X}\colon  X \to \P(X)$ denoting the singleton map $\eta_{X}\colon  x \mapsto \{
x\}$, and $\mu_{X}\colon  \P\P X \to \P X$ denoting union, $\mu_{X}(\cA) :=
\bigcup \cA$.
The contravariant power set functor will be denoted as $\Q$.

We will restrict our attention to set functors satisfying certain properties,
of which the first one is crucial.
In order to define it, we need to recall the notion of a (weak) pullback.
Given two functions $f_{0}\colon  X_{0} \to X$, $f_{1}\colon  X_{1} \to X$, a \emph{weak
pullback} is a set $P$, together with two functions $p_{i}\colon  P \to X_{i}$ such
that $f_{0}\circ p_{0} = f_{1} \circ p_{1}$, and in addition, for every
triple $(Q,q_{0},q_{1})$ also satisfying $f_{0}\circ q_{0} = f_{1}\circ
q_{1}$, there is an arrow $h\colon  Q  \to P$ such that $q_{0} = h\circ p_{0}$ and
$q_{1} = h\circ p_{0}$, in a diagram:\\
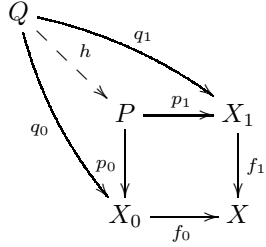
\centerline{
\xymatrix{ Q \ar@/_/[ddr]_{q_0}
\ar@/^/[drr]^{q_1} \ar@{-->}[dr]^h & & \\
& P \ar[r]^{p_1} \ar[d]_{p_0} & X_1 \ar[d]^{f_1}\\
& X_0 \ar[r]_{f_0} & X }}
For $(P,p_{0},p_{1})$ to be a \emph{pullback}, we require in addition the
arrow $h$ to be unique.

A functor $\T$ \emph{preserves weak pullbacks} if it transforms every weak
pullback $(P,p_{0},p_{1})$ for $f_{0}$ and $f_{1}$ into a weak pullback
$(\T P,\T p_{0}, \T p_{1})$ for $\T f_{0}$ and $\T f_{1}$.
An equivalent characterization is to require $\T$ to \emph{weakly preserve
pullbacks}, that is, to turn pullbacks into weak pullbacks.
In the next subsection we will see yet another, and motivating,
characterization of this property.

The second property that we will impose on our set functors is that of
standardness.
Given two sets $X$ and $X'$ such that $X \sse X'$, let $\iota_{X,X'}$ denote
the inclusion map from $X$ into $X'$.
A weak pullback-preserving set functor $\T$ is \emph{standard} if it
\emph{preserves inclusions}, that is: $\T\iota_{X,X'} = \iota_{\T X,\T X'}$
for every inclusion map $\iota_{X,X'}$.

\begin{remark}
Unfortunately the definition of standardness is not uniform throughout the
literature.
Our definition of standardness is taken from \citet{moss:coal99},
while for instance \citet{adam:auto90} have an
additional condition involving so-called distinguished points.
Fortunately, the two definitions are equivalent in case the functor preserves
weak pullbacks, see \citet[Lemma A.2.12]{kupk:fini06}.
\end{remark}

The restriction to standard functors is not essential, since every set
functor is `almost standard'~\citep[Theorem~III.4.5]{adam:auto90}: given
an arbitrary set functor $\T$, we may find a standard set functor $\T'$
such that the restriction of $\T$ and $\T'$ to all non-empty sets and
non-empty functions are naturally isomorphic.

Finally, we shall require that our functors are determined by their
behavior on finite sets.
Call a standard set functor $\T$ \emph{finitary} if $\T X = \bigcup \{ \T X'
\mid X' \sse_{\om} X \}$.
Our focus on finitary functors is not so much a restriction as a convenient
way to express the fact that we are interested in the \emph{finitary
version} of an arbitrary set functor, in the sense that $\Pom$ is the
finitary version of $\P$.
Generally, we may define, for a standard functor $\T$, the functor $\Tom$
that on objects $X$ is defined by $\Tom X = \bigcup \{ \T X' \mid X' \sse
X \}$, while on arrows $f$ we simply put $\Tom f := \T f$.

Since there are many set functor which are standard, finitary and weak
pullback-preserving, the results in this paper have a wide scope.

\begin{example}\label{WpbExamples}
The identity functor $\Fid$, the finitary power set functor $\Pom$, and,
for each set $Q$, the constant functor $C_{Q}$ (given by $C_{Q}X = Q$ and
$C_{Q}f = \id_{Q}$) are standard, finitary, and preserve weak pullbacks.

For a slightly more involved example, consider the finitary \emph{multiset}
functor $M_{\om}$.
This functor takes a set $X$ to the collection $M_{\om}X$ of maps $\mu\colon  X
\to \bbN$ of finite support (that is, for which the set $\mathit{Supp}(\mu)
:= \{ x \in X \mid \mu(x) > 0 \}$ is finite), while its action on arrows is
defined as follows.
Given an arrow $f\colon  X \to X'$ and a map $\mu \in M_{\om}X$, we define
$(M_{\om}f)(\mu)\colon  X' \to \bbN$ by putting
\[
(M_{\om}f)(\mu)(x') := \sum \{ \mu(x) \mid f(x) = x' \}.
\]
With this definition, the functor is not standard, but we may `standardize'
it by representing any map $\mu\colon  X \to \bbN$ of finite support by its
`support graph' $\{ (x,\mu x) \mid \mu x > 0 \}$.
As a variant of $M_{\om}$, consider the finitary probability functor
$D_{\om}$, where $D_{\om} X = \{ \de\colon  X \to [0,1] \mid
\mathit{Supp}(\de) \text{ is finite and } \sum_{x\in X}\de(x) = 1 \}$,
while the action of $D_{\om}$ on arrows is just like that of $M_{\om}$.

Perhaps more importantly, the class of finitary, standard functors that
preserve weak pullbacks, is closed under the following operations:
composition ($\cof$) , product ($\times$), co-product ($+$), and
exponentiation with respect to some set $D$ ($(\cdot)^{D})$.
As a corollary, inductively define the following class $\EKPF_{\om}$ of
\emph{extended finitary Kripke polynomial functors}:
\[
\T \isdef \Fid \mid \Pom \mid C_{Q} \mid M_{\om} \mid D_{\om} \mid
\T_{0} \cof \T_{1} \mid \T_{0} + \T_{1} \mid \T_{0} \times \T_{1} \mid
\T^{D}.
\]
Then each extended Kripke polynomial functor falls in the scope of the
work in this paper.

As running examples in this paper we will often take the \emph{binary
tree functor} $\B = \Fid \times \Fid$, and the finitary power set functor
$\Pom$.
\end{example}

An interesting result of standard functors is that they preserve finite
intersections~\citep[Theorem~III.4.6]{adam:auto90}: $\T(X \cap Y) = \T X \cap
\T Y$.
As a consequence, if $\T$ is finitary, for any object $\xi \in \T X$ we may
define
\[
\Base^{\T}_{X}(\xi) := \bigcap \{ X' \in \Pom(X) \mid \xi \in \T X' \},
\]
and show that $\Base^{\T}_{X}(\xi)$ is the \emph{smallest} set $X'$ such that
$\xi\in\T X'$~\citep{vene:auto06}.
In fact, the base maps provide a natural transformation $\Base^\T \colon  \T \to
\Pom$; for referencing we will mention this fact explicitly in the next
section.

To facilitate the reasoning in this paper, which will involve objects of
various different types, we use a variable naming convention.
\begin{convention}
Let $X$ be a set and let $\T \colon \Set \to \Set$ be a functor. We
use the following naming convention:
\begin{center}
\begin{tabular}{rl}
\hline
Set & Elements\\
\hline
$X$& $a,b,\dotsc, x,y, \dotsc$\\
$\T X$ & $\alpha,\beta, \dotsc$\\
$\P X$ & $A, B, \dotsc$\\
$\P \T X$& $\Gamma, \Delta, \dotsc$\\
$\T \P X$ & $\Phi, \Psi,\dotsc$
\end{tabular}
\end{center}
\end{convention}

\subsection{Relation lifting}
\label{Subsect:RelLift}

In \S \ref{Section:Intro}, we mentioned that coalgebraic modal logic
using the cover modality, as introduced by Moss, crucially uses relation lifting, both for its
syntax and semantics. Relation lifting is a technique which allows one
to extend a functor $\T \colon \Set \to \Set$ defined on the category
of sets to a functor $\Tb \colon \Rel \to \Rel$ on the category of
sets and relations in a natural way.
In this subsection we will introduce some of the
basic facts and definitions about relation lifting.

Let $\T$ be a set functor.
Given two sets $X$ and $X'$, and a binary relation $R$ between $X \times X'$,
we define the \emph{lifted relation} $\Tb (R) \sse \T X \times \T X'$
as follows:
\[
\Tb (R) := \{ ((\T\pi) (\rho), (\T\pi')(\rho)) \mid \rho \in \T R \},
\]
where $\pi \colon R \to X$ and $\pi'\colon  R \to X'$ are the projection functions given
by $\pi(x,x') = x$ and $\pi'(x,x') = x'$.
In a diagram:
\\ \centerline{\xymatrix{%
X  & R \ar[l]_{\pi} \ar[r]^{\pi'} & X'
\\
TX & TR \ar[l]_{\T \pi} \ar@{>>}[d]
\ar@/_5mm/[dd]_{\langle\T\pi,\T\pi'\rangle}
\ar[r]^{\T\pi'} & \T X'
\\
& \Tb R \ar@{^{(}->}[d]
\\
& \T X \times \T X' \ar[ruu] \ar[luu]
}}
\\In other words, we apply the functor $\T$ to the relation $R$, seen as
a \emph{span} \[\xymatrix{X  & R \ar[l]_{\pi} \ar[r]^{\pi'} & X'},\] and
define $\Tb R$ is the image of $\T R$ under the product map
$\langle\T\pi,\T\pi'\rangle$ obtained from the lifted projection maps
$\T\pi$ and $\T\pi'$.

Let us first see some concrete examples.

\begin{example} \label{RelLift:Examples}
Fix a relation $R \sse X \times X'$.
For the identity and constant functors, we find, respectively:
\begin{eqnarray*}
\rl{\Fid}  R &=& R
\\ \rl{C_{Q}} R &=& \Id_{Q}.
\end{eqnarray*}

The relation lifting associated with the power set functor $\P$ can be
defined concretely as follows:
\[
\Pb R = \{ (A,A') \in \P X \times \P X' \mid
\forall a \in A\, \exists a' \in A'. aRa'
\text{ and }
\forall a' \in A'\, \exists a \in A. aRa'
\}.
\]
This relation is known under many names, of which we mention that of the
\emph{Egli-Milner} lifting of $R$.
For any standard, weak pullback preserving functor $\T$ it can be
shown~\citep{kupk:comp10}
that the lifting of $\Tom$ agrees with that of $\T$, in the sense that
$\rl{\Tom} R = \Tb R \cap (\Tom X \times \Tom X')$.
From this it follows that
\[
\text{ for all } A \in \Tom X, A' \in \Tom X':
A \rel{\rl{\Pom} R} A' \text{ iff } A \rel{\rl{P}R} A',
\]
and for this reason, we shall write $\rl{\P} R$ rather than $\rl{\Pom} R$.

Relation lifting for the finitary multiset functor is slightly more involved:
given two maps $\mu\in M_{\om}X,
\mu'\in M_{\om}X'$, we put
\begin{align*}
\mu \rel{\rl{M_{\om}}R} \mu'  \text{ iff } &\text{there is some map } \rho\colon  R\to\bbN
\text{ such that } \\
& \forall x \in X.\, \textstyle{\sum} \{ \rho(x,x') \mid x' \in X' \} = 1,
 \text{ and } \\
& \forall x' \in X'.\, \textstyle{\sum} \{ \rho(x,x') \mid x \in X \} = 1.
\end{align*}
The definition of $\rl{D_{\om}}$ is similar.

Finally, relation lifting interacts well with various operations on
functors~\citep{herm:stru98}.
In particular, we have
\begin{eqnarray*}
\rl{T_{0}\cof T_{1}}R   &=& \rl{T}_{0}(\rl{T}_{1} R)
\\ \rl{T_{0}+T_{1}}R       &=& \rl{T}_{0}R \cup \rl{T}_{1}R
\\ \rl{T_{0}\times T_{1}}R &=&
\left\{ \left((\xi_{0},\xi_{1}),(\xi'_{0},\xi'_{1})\right) \mid
(\xi_{i},\xi'_{i}) \in \rl{T}_{i}, \text{ for } i \in \{0,1\}
\right\}.
\\ \rl{T^{D}}R &=& \{ (\phi,\phi') \mid (\phi(d),\phi'(d) \in \rl{T}R
\text{ for all } d \in D \}
\end{eqnarray*}
\end{example}

\begin{remark}\label{r:well-def}
Strictly speaking, the definition of the relation lifting of a given relation
$R$ depends on the type of the relation, i.e.\ given sets $X,X',Y, Y'$ such
that $R \sse X \times X'$ and $R \sse Y \times Y'$, it matters whether we
look at $R$ as a relation from $X$ to $X'$ or as a relation from $Y$ to
$Y'$.
We have avoided this potential source of ambiguity by requiring the functor
$\T$ to be \emph{standard}, see Fact~\ref{f:rl}(6).
\end{remark}

Relation lifting has a number of properties that we will use throughout the
paper.
It can be shown that relation lifting interacts well with the operation of
taking the graph of a function $f: X \to X'$, and with most operations on
binary relations.
Most of the properties below are easy to establish --- we refer
to~\citep{kupk:comp10} for proofs.

\begin{fact}
\label{f:rl}
Let $\T$ be a set functor.
Then the relation lifting $\Tb$ satisfies the following properties, for
all functions $f\colon  X \to X'$, all relations $R,S \sse X \times X'$, and all
subsets $Y\sse X$, $Y' \sse X'$:
\begin{enumerate}
\item $\Tb$ extends $\T$:
$\Tb (\Graph{f}) = \Graph{(\T f)}$;

\item $\Tb$ preserves the diagonal:
$\Tb(\Id_{X}) = \Id_{\T X}$;

\item $\Tb$ commutes with relation converse:
$\Tb (\cv{R}) = \cv{(\Tb{R})}$;

\item $\Tb$ is monotone:
if $R \subseteq S$ then $\Tb (R) \subseteq \Tb (S)$;

\item $\Tb$ distributes over composition:
$\Tb(R \cor S) = \Tb (R) \cor \Tb (S)$, if $\T$ preserves weak pullbacks.

\item $\Tb$ commutes with restriction:
$\Tb(R\rst{Y\times Y'}) = \Tb{R}\rst{\T Y \times \T Y'}$,
if $\T$ is standard and preserves weak pullbacks.
\end{enumerate}
\end{fact}

Fact~\ref{f:rl}(5) plays a key role in our work.
In fact, distributivity of $\Tb$ over relation composition is
\emph{equivalent} to $\T$ preserving weak-pullbacks; the proof of this
equivalence goes back to \citet{trnk:rela77}.

Many proofs in this paper will be based on Fact~\ref{f:rl}, and we will not
always provide all technical details.
In the lemma below we have isolated some facts that will be used a number
of times; the proof may serve as a sample of an argument using properties
of relation lifting.

\begin{lemma}
\label{p:rl2}
Let $\T\colon \Set \to \Set$ be a standard, finitary, weak pullback-preserving functor. Let $X,Y$ be sets, let $f,g\colon X \to Y$ be two functions and let
$R \subseteq X\times X$ and $S \sse Y \times Y$ be relations.
\begin{enumerate}
\item If $(X,R)$ is a pre-order, then so is $(\T X,\Tb{R})$.
\item If $f(x) \rel{S} g(x)$ for all $x\in X$, then $\T f (\alpha)
\rel{ \Tb S} \T g (\alpha)$ for all $\alpha \in \T X$.
\item If $x \rel{R} y$ implies $f(x) \rel{S} g(y)$ for all $x,y \in X$,
then $\alpha \rel{ \Tb R} \beta$ implies $(\T f) \alpha \rel{ \Tb S}
(\T g) \beta$ for all $\alpha,\beta \in \T X$.
\end{enumerate}
\end{lemma}

\begin{proof}
For part~1, observe that $(X,R)$ is a pre-order iff $\Id_{X} \sse R$ and
$R\cor R \sse R$.
Hence, if $(X,R)$ is a pre-order, it follows from Fact~\ref{f:rl}(2,4) that
$\Id_{\T X} = \Tb\Id_{X} \sse \Tb R$, and from Fact~\ref{f:rl}(5,4) that
${\Tb R}\cor{\Tb R} = \Tb (R\cor R) \sse \Tb R$, implying that $(\T X,\Tb{R})$
is a pre-order as well.

For part~2, observe that the antecedent can be succinctly expressed as
\[
\cv{(\Graph{f})}\cor \Graph{g} \sse S.
\]
Then it follows by the properties of relation lifting that
\begin{align*}
\cv{(\Graph{\T f})}\cor \Graph{\T g}
& =
\cv{(\Tb(\Graph{f}))} \cor{\Tb(\Graph{g})}
&& \text{(Fact~\ref{f:rl}(1))}
\\ & =
\Tb(\cv{(\Graph{f})}) \cor{\Tb(\Graph{g})}
&& \text{(Fact~\ref{f:rl}(3))}
\\ & =
\Tb(\cv{(\Graph{f})}\cor \Graph{g})
&& \text{(Fact~\ref{f:rl}(5))}
\\ &\sse
\Tb S
&& \text{(Fact~\ref{f:rl}(4))}
\end{align*}
But the inclusion $\cv{(\Graph{\T f})}\cor \Graph{\T g} \sse \Tb S$ is
just another way of stating the conclusion of part~2.

For part~3, we reformulate the statement of its antecedent as
\[
\cv{(\Graph{f})}\cor R\cor \Graph{g} \sse S.
\]
On the basis of this we may reason, via a completely analogous argument
to the one just given, that
\[
\cv{(\Graph{\T f})}\cor \Tb R\cor \Graph{\T g} \sse \T S,
\]
which is equivalent way of phrasing the conclusion of part~3.
\end{proof}

Relation lifting interacts with the map $\Base^\T $ as follows \citep[see][]{kupk:comp10}:

\begin{fact}
\label{f:base}
Let $\T$ be a standard, finitary, weak pullback-preserving functor.
\begin{enumerate}
\item
$\Base^\T $ is a natural transformation $\Base^\T \colon  \T \to \Pom$.
That is, given a map $f\colon  X \to X'$
the following diagram commutes:
\\ \centerline{\xymatrix{%
\T X \ar[r]^{\Base^\T _{X}} \ar[d]_{\T f}
& \Pom X \ar[d]^{\P f}
\\
\T X' \ar[r]^{\Base^\T _{X'}} & \Pom X'
}}
\item
Given a relation $R \sse X \times X'$ and elements $\al \in \T X$,
$\be \in \T Y$, it follows from $\al \rel{\Tb R} \be$ that $\Base^\T (\al)
\rel{\Pb R} \Base^\T (\be)$.
\end{enumerate}
\end{fact}

An interesting relation to which we shall apply relation lifting is the
\emph{membership} relation $\in$.
If needed, we will denote the membership relation restricted to a given set
$X$ as the relation ${\in_{X}} \sse X \times \P X$.
Given a set $X$ and $\Phi \in \T \P X$, we define
\[
\nbsem_{X}(\Phi) = \{ \alpha \in \T X \mid \alpha \rel{\Tb {\in_{X}}} \Phi \}.
\]
Elements of $\nbsem(\Phi)$ will be called \emph{lifted members of
$\Phi$}.
Properties of  $\nbsem$ are intimately related to those of $\Tb$ \citep[see][]{kupk:comp10}:

\begin{fact}
\label{f:db}
Let $\T\colon \Set \to \Set$ be a standard, finitary, weak pullback-preserving functor.
Then the collection of maps $\nbsem_{X}$ forms a distributive law with respect to
both the co- and the contravariant power set functor.
That is, $\nbsem$ provides two natural transformations,
$\nbsem \colon\T \P \to \P\T$, and $\nbsem\colon  \T\Q \to \Q\T$.
\end{fact}

\begin{remark}
\label{r:rl}
One can strengthen Fact \ref{f:db}:
$\nbsem$ is actually a distributive law over the \emph{monad} $(\P,\mu,\eta)$, in the
sense of being also compatible with the unit $\eta$ and the multiplication
$\mu$ of $\P$, as given by the following two diagrams:
\\\centerline{
\xymatrix{
\T X \ar[rd]_{\single_{\T X}}
\ar[r]^{\T \single_X} &
\T \P X
\ar[d]^{\nbsem_X} \\
& \P \T X }
\hspace*{20mm}
\xymatrix{\T \P \P X \ar[d]^{\T \mu_X}
\ar[r]^{\nbsem_{\P X}} &  \P \T \P X
\ar[r]^{\P \nbsem_X}&
\P \P \T X \ar[d]^{\mu_X} \\
\T \P X \ar[rr]_{\nbsem_X}& & \P \T X
}
}

In the terminology of \citet{Street:Monads}, $(T,\nbsem)$ is a
\emph{monad opfunctor} from the monad $\P$ to itself,
and there is a one-one correspondence between the monad opfunctors and the
functors $T$ equipped with extensions to endofunctors on the \emph{Kleisli
category} $\class{Kl}(\P)$ associated with $\P$.
(The explicit results in \citep{Street:Monads}, using the 2-functor $\mathbf{Alg_{C}}$,
are in terms of monad functors and extensions to the category of Eilenberg-Moore algebras.
The results for monad opfunctors and the Kleisli category are dual.)
Note that the Kleisli category of the power set monad is (isomorphic to)
the category $\class{Rel}$ with sets as objects, and binary relations as
arrows.
The correspondence mentioned then links the natural transformation
$\nbsem$ to the notion of relation lifting $\Tb$.

\end{remark}

\begin{lemma}
\label{f:rlnt}
Let $\T$ be a standard, finitary, weak pullback-preserving functor.
Let $X$ be some set and let $\Phi \in \T\P X$.
\begin{enumerate}
\item If $\nada \in \Base^\T (\Phi)$ then $\nbsem(\Phi) = \nada$.
\item If $\Base^\T (\Phi)$ consists of singletons only, then
$\nbsem(\Phi)$ is a singleton.
\item If $\T$ maps finite sets to finite sets, then for all $\Phi \in
\T \Pom X$, $|\nbsem(\Phi)| < \omega$.
\end{enumerate}
\end{lemma}

\begin{proof}
For part~1, suppose that $\al$ is a lifted member of $\Phi$; then we may
derive by Fact~\ref{f:base} that $\Base^\T (\al) \Pin \Base^\T (\Phi)$.
But from this it would follow, if $\nada \in \Base^\T (\Phi)$, that $\Base^\T (\al)$
contains a member of $\nada$, which is clearly impossible.
Consequently, then $\nbsem(\Phi)$ is empty.

For part~2, observe that another way of saying that $\Base^\T (\Phi)$ consists of
singletons only, is that $\Phi\in \T S_{X}$, with $S_{X} := \{ \{ x \} \mid
x \in X \}$.
Let $\theta_{X}\colon  S_{X} \to X$ be the inverse of $\eta_{X}$, that is,
$\theta_{X}$ is the bijection mapping a singleton $\{ x \}$ to its unique
member $x$.
Clearly then, we have
$\cv{(\Graph{\theta_{X}})} = {\in}\rst{X \times S_{X}}$, from which
it follows by Fact~\ref{f:rl} that $\cv{(\Graph{\T\theta_{X}})} =
{\Tin}\rst{\T X \times \T S_{X}}$.
From this it is immediate that if $\Phi \in \T S_{X}$, then $(\T\theta_{X})
(\Phi)$ is the unique lifted member of $\Phi$.

Finally, we consider part~3.
Since $\T$ is finitary, $\Phi \in \T \Pom X$ implies that $\Phi \in \T\Pom
Y$ for some finite set $Y$, and from this it follows that $\Base^\T (\Phi)
\sse \Pom Y$.
If $\al$ is a lifted member of $\Phi$, then by Fact~\ref{f:base} we obtain
$\Base^\T (\al) \Pin \Base^\T (\Phi)$, and so in particular we find $\Base^\T (\al)
\sse \bigcup \Base^\T (\Phi) \sse Y$.
From this it follows that $\nbsem(\Phi) \sse \T Y$, and so $\nbsem(\Phi)$
must be finite by the assumption on $\T$.
\end{proof}

\subsection{Frames and their presentations}
\label{Subsect:FramePres}

A \emph{frame} is a complete lattice in which finite meets distribute over
arbitrary joins.
The signature of frames consists of arbitrary joins and finite meets,
and it will be convenient for us to include the top and bottom as well.
Thus a frame will usually be given as $\mathbb{L} =\struc{ L,
\smallvee, \wedge,0,1}$, while we will often consider join and meet as
functions $\smallvee_\mathbb{L} \colon \P L \to L$ and
$\smallwedge_\mathbb{L} \colon \Pom L \to L$.
This enables us for instance to define a frame homomorphism $f\colon  \bbL
\to \bbM$ as a map from $L$ to $M$ satisfying $f \cof {\bw} = {\bw}\cof
(\Pom f)$ and $f \cof {\bv} = {\bv}\cof (\P f)$.
By $\class{Fr}$ we denote the category of frames and frame homomorphisms.
The initial frame (the lattice of truth values) will be denoted as $\two$,
and for a given frame $\bbL$
we will let $!_{\bbL}$ denote the unique frame homomorphism from $\two$ to
$\bbL$, omitting the subscript if $\bbL$ is clear from context.

The order relation $\leq_{\bbL}$ of a frame $\bbL$ is given by $a \leq_{\bbL}
b$ if $a\land b = a$ (or, equivalently, $a \lor b = b$).
We can adjoin an implication operation to a frame $\mathbb{L}$ by defining
$a\to b := \smallvee \{c \mid a\wedge c \leq b\}$; this operation turns
$\bbL$ into a Heyting algebra.
As a special case of implication we can  consider the \emph{negation}: $\neg a
:= \smallvee \{ c \mid a \wedge c = 0\}$.
Generally, neither of these two operations is preserved by frame homomorphisms.
A subset $S$ of $\bbL$ is \emph{directed} if for every $s_{0},s_{1} \in S$
there is an element $s \in S$ such that $s_{0},s_{1} \leq s$.
The join of a directed set $S$ is often denoted as $\bvdir S$.

A \emph{frame presentation} is a tuple $\struc{ G \mid R}$ where $G$
is a set of generators and $R\subseteq \P \Pom G
\times \P \Pom G $ is a set of relations. A
presentation $\struc{ G\mid R}$ \emph{presents} a frame $\mathbb{L}$ if
there exists a function $f\colon G \to L$ which is \emph{compatible with
$R$}, i.e.~such that
\[
\text{for all } (t_1,t_2) \in R, \, \bigvee_{A\in t_1} \smallwedge
(\Pom f) A = \bigvee_{B\in t_2} \smallwedge
(\Pom f) B,
\]
and for all frames $\mathbb{M}$ and functions $g\colon G \to M$ compatible
with $R$, there is a
unique frame homomorphism $g'\colon \mathbb{L}\to \mathbb{M}$ such
that $g'f=g$.
We call $f$ the \emph{insertion of generators} (of $G$ in $\bbL$).

\begin{fact} \label{Prelim:FramesPresent}
Every frame presentation presents a frame.
\end{fact}

The details of the proof of the above fact \citep[found in][\S 4.4]{Vickers1989} tell us how to construct a
unique frame given a presentation $\struc{ G \mid R}$. Omitting these
details of the construction, we denote this unique frame by $\Fr
\struc{ G \mid R}$. We will usually write $\bigvee_{i\in I}
\smallwedge A_i = \bigvee_{j\in J} \smallwedge B_j$ instead of $(\{
A_i \mid i\in I\}, \{ B_j\mid j\in J\})$ when specifying relations. In
light of the fact that $a \leq b$ iff $a \vee b = b$, we will also
allow ourselves the liberty to specify inequalities of the shape
$\bigvee_{i\in I} \smallwedge A_i \leq \bigvee_{j\in J} \smallwedge
B_j$ as relations.
It follows from the proof of Fact \ref{Prelim:FramesPresent} that if
$f\colon G \to \Fr \struc{ G \mid R}$ is the insertion of generators, then
every element of $\Fr \struc{G\mid R}$ can be written as
$\bigvee_{i\in I} \smallwedge \Pom f A$ for some $\{A_i \mid i\in I\} \in \P
\Pom G$; in other words every element of $\Fr \struc{G \mid R}$ can be
written as an infinite disjunction of finite conjunctions of
generators.

We will now introduce flat site presentations for frames,
which have as one of their main advantages that they allow us to assume that
an arbitrary element of the frame being presented is an infinite join of
generators.
A \emph{flat site} is a triple $\struc{ X, \lesssim, \cov_0 }$,
where $\struc{ X, \lesssim }$ is a pre-order and $\cov_0
\subseteq X \times \P X$ is a binary relation such that for all $b\lesssim a \cov_0
A$, there exists $B\subseteq {\downarrow} A \cap {\downarrow} b$ such
that $b \cov_0 B$. A flat site $\struc{ X, \lesssim, \cov_0 }$ presents a frame $\mathbb{L}$ if there
exists a function $f\colon X \to L$ such that
\begin{itemize}
\item $f$ is order-preserving,
\item $1 \leq \smallvee (\P f) X$,
\item for all $a,b \in X$,
$f(a)\wedge f(b)\leq \smallvee (\P f) ({\downarrow} a \cap {\downarrow} b)$, and
\item for all $a \cov_0 A$,
$f(a) \leq \smallvee (\P f) A$
\end{itemize}
and for all frames
$\mathbb{M}$ and all $g\colon X \to M$ satisfying the above two
properties, there exists a unique frame homomorphism $g'\colon
\mathbb{L} \to \mathbb{M}$ such that $g' \cof f =g$. Specifically, for all $a
\in \bbL$, \[ g'(a) = \smallvee \{ g(x) \mid f(x) \leq a \}.\]
To put it another way, the frame presented by a flat site is
\[\begin{array}{rl}
\Fr \struc{ X, \lesssim,
\cov_0 } \simeq \Fr \struc{ X \mid& a\leq b \quad (a \lesssim b),\\
& a \leq \smallvee A \quad (a\cov_0 A),\\
& 1 = \smallvee X \\
&
a\wedge b = \smallvee \{ c \mid c\lesssim
a, c\lesssim b \}
}.
\end{array} \]
A \emph{suplattice} is a complete $\smallvee$-semilattice; accordingly, a
suplattice homomorphism is a map which preserves $\smallvee$. A
\emph{suplattice presentation} is a triple $\struc{ X, \lesssim,
\cov_0}$ where $\struc{X, \lesssim}$ is a pre-order and $\cov_0
\subseteq X \times \P X$.
A suplattice presentation $\struc{ X, \lesssim,
\cov_0}$ presents a suplattice $\bbL$ if there exists a function
$f\colon X \to L$ such that
\begin{itemize}
\item $f$ is order-preserving;
\item for all $a \cov_0 A$, $f(a) \leq \smallvee \P f (A)$;
\end{itemize}
and for all suplattices $\bbM$ and all functions $g\colon X \to M$
respecting the above two conditions, there exists a unique suplattice
homomorphism $g' \colon \bbL \to \bbM$ such that $g' \cof f =
g$. Every suplattice presentation presents a suplattice
\citep[Prop.~2.5]{JMV2008}. Now observe that every flat site can also
be seen as a suplattice presentation with an additional stability condition. Consequently, given a flat site
$\struc{ X, \lesssim, \cov_0}$, we can generate two different objects with it: a frame
$\Fr \struc{ X, \lesssim, \cov_0}$ and a suplattice
$\operatorname{SupLat} \struc{ X, \lesssim, \cov_0}$. The
Flat site Coverage Theorem \citep[Theorem 5]{Vickers2006} tells us that these two objects are in fact
order isomorphic.
\begin{fact} \label{Prelim:FlatSiteCoverage}
Let $\struc{ X, \lesssim, \cov_0}$ be a flat site. Then $\Fr \struc{
X, \lesssim, \cov_0} \simeq \SupLat \struc{ X, \lesssim, \cov_0}$.
\end{fact}
We record the following consequences of the above fact. Suppose that
$\struc{ X, \lesssim, \cov_0}$ is a flat site which presents a frame
$\bbL$ via $f\colon X \to L$. Then
\begin{itemize}
\item every element of $\bbL$ is of the shape $\smallvee \P f (A)$ for
some $A \in \P X$;
\item we can use $\struc{ X, \lesssim, \cov_0}$ both to define
suplattice homomorphisms and frame homomorphisms.
\end{itemize}

\subsection{Powerlocales via $\Box$ and $\Diamond$}
\label{Nabla:BoxDiamond}

We will now introduce the Vietoris powerlocale. In line with our generally
algebraic approach we shall define it directly as a functor on the category of
frames rather than its opposite, the category of locales. In its full
generality it originates (as the ``Vietoris construction'') in
\citep{VietLoc}, with some earlier, more restricted references in
\citep{Johnstone1982}. For locales it is a localic analogue of hyperspace (with
Vietoris topology). The points are (in bijection with) certain sublocales of
the original locale. For a full constructive description see \citep{PowerPt}.

Given a frame $\mathbb{L}$, we first define $L_{\Box}:=L$ and
$L_{\Diamond}:=L$, and then
\[ \begin{array}{rll}
\V \mathbb{L} := \Fr \langle L_\Box \oplus L_\Diamond \mid & \Box 1 = 1\\
& \Box( a\wedge b) = \Box a \wedge \Box b\\
&\Box (\smallvee^{\uparrow} A) = \smallvee^{\uparrow}_{a\in A} \Box a
& (A\in \P L \text{ directed})\\
&\Diamond (\smallvee A) = \smallvee_{a \in A} \Diamond a   & (A \in \P L)\\
& \Box a \wedge \Diamond b \leq \Diamond(a\wedge b)\\
&\Box (a\vee b) \leq \Box a \vee \Diamond b\\
&\rangle
\end{array}
\]
\begin{remark}
We are abusing notation when specifying the relations in the definition above.
Strictly speaking, we have two maps, $\Box\colon L_{\Box}\rightarrow\V
\mathbb{L}$ for the left copy of $\mathbb{L}$ and $\Diamond\colon L_{\Diamond
}\rightarrow\V\mathbb{L}$ for the right copy of $\mathbb{L}$, so that the
insertion of generators is the map $\Box\oplus\Diamond\colon L_{\Box}\oplus
L_{\Diamond}\rightarrow\V
\mathbb{L}$.
\end{remark}

\citet{VietLoc} shows that $\V
$ gives a monad on the category of locales, i.e. a comonad on the category of
frames. We shall not need the full strength of this here, but some of the
ingredients of the comonad structure are easy to check.

\begin{itemize}
\item $\V
$ is functorial. If $f\colon\mathbb{L}\rightarrow\mathbb{M}$ is a frame
homomorphism, then the function $(\Box f)\oplus(\Diamond f)\colon L_{\Box
}\oplus L_{\Diamond}\rightarrow\V M$ is compatible with the relations in the
presentation of $\V\mathbb{L}$, so that there is a frame homomorphism $\V
f\colon\V
\mathbb{L}\rightarrow\V
\mathbb{M}$ extending this map. It is also easy to show functoriality.

\item  The counit $\counit_{\mathbb{L}}\colon\V\mathbb{L}\rightarrow\mathbb{L}$
is given by $\Box a\mapsto a$ and $\Diamond a\mapsto a$. The comultiplication
$\mu_{\mathbb{L}}\colon\V\mathbb{L}\rightarrow\V\V\mathbb{L}$ is given by
$\Box a\mapsto\Box\Box a$ and $\Diamond a\mapsto\Diamond\Diamond a$.
\end{itemize}

\section{The $T$-powerlocale construction}
\label{Section:VT}
In this section we arrive at the main conceptual contribution of this
paper. Given a weak pullback-preserving, standard, finitary functor
$\T \colon \Set \to \Set$, we define its associated $\T$-powerlocale
functor $\VT \colon \class{Fr} \to \class{Fr}$ on the category of
frames, using the Carioca axioms for coalgebraic modal logic.  This
construction truly generalizes the Vietoris powerlocale construction,
because we will see that the $\Pom$-powerlocale is isomorphic to the
Vietoris powerlocale.  The other two major results in this section are
the fact that one can lift a natural transformation between transition
functors $\nta \colon \T' \to \T$ to a natural transformation
$\widehat \rho\colon \VT \to \VTp$ going in the other direction, and
the fact that $\T$-powerlocales are join-generated by their generators
of the shape $\nabla \alpha$. We will establish the latter fact via
the stronger result by showing that $\VT \bbL$ admits a flat site
presentation. The fact that $\VT \bbL$ is join-generated by its
generators is not entirely surprising, since the Carioca axioms were
designed with the desirability of conjunction-free disjunctive normal
forms in mind \citep{BPV2008}; however the precise mathematical
formulation of this property, using flat sites and suplattices, is an
improvement over what was previously known.

This section is organized as follows. In \S \ref{Subsect:IntroVT} we
introduce the $\T$-powerlocale construction on frames. In \S
\ref{Subsect:BasicObs} we make technical observations about
$\T$-powerlocales. In \S \ref{Subsect:VP}, we consider two
instantiations of the $\T$-powerlocale construction, the most notable
of which is the $\Pom$-powerlocale which is isomorphic to the
classical Vietoris powerlocale. In \S \ref{Section:VT:functorial} we
extend the $\T$-powerlocale construction to a functor $\VT$ on the
category of frames, and we show how one can lift natural
transformations between set functors $\T$, $\T'$ to natural
transformations between powerlocale functors $\VT$, $\VTp$. We
conclude this section with \S \ref{Subsect:VTFlatsite}, in which we
show that the $\T$-powerlocale construction admits a flat site
presentation, a corollary of which is that each element of $\VT \bbL$
has a disjunctive normal form.

\subsection{Introducing the $\T$-powerlocale}\label{Subsect:IntroVT}
In this subsection, we will use the Carioca axioms for coalgebraic
modal logic \citep{BPV2008} to define the $\T$-powerlocale $\VT \bbL$ of a given
frame $\bbL$ using a frame presentation, i.e.~using generators and
relations. The generators of $\VT \bbL$ will be given by the set $\T
L$; in order to specify the relations we will use \emph{relation
lifting} (\S \ref{Subsect:RelLift}) and \emph{slim redistributions},
which we will introduce below. In addition, we will provide
an alternative presentation of $\VT \bbL$, which does not use slim
redistributions. From a conceptual viewpoint, it is not immediately
obvious which presentation of $\VT \bbL$ should be taken as the
primary definition. Our choice to use slim redistributions in the primary
definition is motivated by the extant literature \citep{BPV2008,KKV2008,kupk:comp10}.

\begin{definition}
Let $\T\colon \Set \to \Set$ be a standard, finitary, weak pullback-preserving functor,
let $X$ be a set and let $\Gamma \in \Pom\T X$. The set of all
\emph{slim redistributions} of $\Gamma$ is defined as follows:
\[\SRD(\Gamma)= \left\{ \Psi \in \T \Pom \left (
\smallcup_{\gamma\in \Gamma} \Base^\T (
\gamma) \right) \mid \forall \gamma\in\Gamma, \,
\gamma \rel{\rl{\T}{\in}} \Psi \right\} \]
\end{definition}

Intuitively, $\Psi \in \T \Pom X$ is a slim redistribution of $\Ga \in
\Pom\T X$ if (i) $\Psi$ is `obtained from the material of $\Ga$', that
is: \[\Psi \in \T \Pom \left (\smallcup_{\gamma\in \Gamma} \Base^\T (\gamma)
\right),\] and (ii) every element of $\Ga$ is a lifted member of
$\Psi$, or equivalently, $\Ga \sse \nbsem(\Psi)$. We illustrate this
with the motivating example of slim redistributions, namely slim
redistribution for the finite powerset functor.

\begin{example} \label{SRD:Pom}
Recall from Example \ref{RelLift:Examples} that if $R\subseteq X\times Y$ is a relation then
$\rl{\Pom}R\subseteq\Pom X\times\Pom Y$ can be characterized as follows:%
\[
\alpha\rl{\Pom}R\beta\text{ iff }\forall x\in\alpha,\,\exists y\in
\beta,\,xRy\text{ and }\forall y\in\beta,\,\exists x\in\alpha,\,xRy\text{.}%
\]
In particular, for ${\in} \,\subseteq X\times\P X$ we get $\alpha
\rel{\rl{\Pom}\in}\Gamma$ iff $\alpha\subseteq\smallcup\Gamma$ and
$\forall\gamma\in\Gamma,\,\gamma\between\alpha$.
(Recall that $\gamma\between\alpha$ means that $\gamma\cap\alpha$ is
inhabited.)
For an order $\leq$, let us define the \emph{upper, lower
}and \emph{convex} pre-orders on finite sets:%
\begin{align*}
\alpha &  \leq_{L}\beta\text{ if } \alpha \subseteq \downarr \beta,
\text{ i.e. } \forall x\in\alpha,\,\exists y\in
\beta,\,x\leq y\\
\alpha &  \leq_{U}\beta\text{ if } \uparr \alpha \supseteq \beta,
\text{ i.e. }\forall y\in\beta,\,\exists x\in
\alpha,\,x\leq y\\
\alpha &  \leq_{C}\beta\text{ if }\alpha\leq_{L}\beta\text{ and }\alpha
\leq_{U}\beta\text{.}%
\end{align*}
Thus $\rl{\Pom}\leq$ is $\leq_{C}$.

Next, if $\alpha\in\Pom S$ then%
\[
Base(\alpha)=\smallcap\{S^{\prime}\in\Pom(S)\mid \alpha\subseteq S^{\prime
}\}= \alpha\text{.}%
\]
From this, if $\Gamma\in\Pom\Pom X$ then%
\begin{align*}
\SRD(\Gamma)  &  =\{\Psi\in\Pom\Pom\left(  \smallcup\Gamma\right)  \mid
\forall\gamma\in\Gamma,\,\left(  \gamma\subseteq\smallcup\Psi\text{ and }%
\forall\alpha\in\Psi,\,\alpha\between\gamma\right) \\
&  =\{\Psi\in\Pom\Pom\left(  X\right)  \mid\smallcup\Psi=\smallcup\Gamma\text{ and
}\forall\gamma\in\Gamma,\,\forall\alpha\in\Psi,\,\alpha\between\gamma
\}\text{.}%
\end{align*}

\end{example}

\begin{definition}
Let $\T$ be a standard, finitary, weak pullback-preserving
functor. Let $\mathbb{L}$ be a frame. We define the
\emph{$T$-powerlocale of $\mathbb{L}$}
\[ \V_\T \mathbb{L} := \Fr \struc{ \T
L \mid (\nabla 1), (\nabla 2), (\nabla 3) },\] where the relations
are the \emph{Carioca axioms} \citep{BPV2008}:
\[ \begin{array}{lll}
(\nabla 1) & \nabla \alpha \leq \nabla \beta, &(\alpha \rel{ \Tb {\leq
} }
\beta)\\
(\nabla 2) & \smallwedge_{\alpha\in \Gamma} \nabla \alpha \leq \smallvee
\{\nabla (\T \smallwedge ) \Psi \mid \Psi \in \SRD(\Gamma) \},&(\Gamma
\in \Pom\T L)\\
(\nabla 3) & \nabla (\T \smallvee) \Phi \leq \smallvee \{ \nabla \beta
\mid \beta \rel{ \Tb {\in}} \Phi \},& (\Phi \in \T \P L)
\end{array} \]
\end{definition}
\begin{remark}
To be precise, we assume that $\nabla \colon \T L \to \V_\T L$ is
the insertion of generators, so when specifying the relations we
should write e.g.~$\alpha \leq \beta$ instead of $\nabla \alpha \leq
\nabla \beta$. The way we have specified the relations above is more
consistent with \citep{BPV2008}.
\end{remark}
We will discuss the instantiation of these axioms for $\T = \Pom$ in
some more detail in \S \ref{Subsect:VP}.

We will now present a very useful equivalent definition of $\VT \bbL$.
The crucial observation behind the alternative definition of $\VT\bbL$
is the following technical lemma, which  characterizes the slim
redistributions of a given finite subset $\Ga$ of $\struc{ \T L, \Tb {\leq} }$
as the maximal lower bounds of $\Ga$. Observe that the lemma also
holds in case $\Gamma = \emptyset$.

\begin{lemma}\label{Flatsite:LowerBounds}
Let $\T\colon \Set \to \Set$ be a standard, finitary, weak
pullback-preserving functor, let $\bbL$ be a meet-semilattice (e.g., a frame) and let $\Gamma \in \Pom \T L$.
Then for any $\al\in \T L$, the following are equivalent:
\begin{enumerate}
\item[(a)]
$\al \in \T L$ is a lower bound of $\Gamma$, that is, $\al \rel{ \Tb {\leq}}
\ga$ for all $\ga\in\Ga$;
\item[(b)]
$\al \rel{ \Tb {\leq}} (\T\bwsmall)\Phi$ for some $\Phi\in\SRD(\Ga)$.
\end{enumerate}
In particular, if $\Phi\in\SRD(\Ga)$ then $(\T\bwsmall)\Phi \rel{ \Tb {\leq}}
\ga$ for all $\ga\in\Ga$.
\end{lemma}

\begin{proof} Recall that
\[ \SRD(\Gamma) := \left\{ \Psi \in \T \P \left(\smallcup_{\gamma \in \Gamma}
\Base^\T (\gamma) \right) \mid \Gamma \subseteq \nbsem(\Psi) \right\}.\]

For the implication from (b) to (a), observe that for any $a\in L$ and $A
\in \Pom L$,  we have that $a \in A$ implies that $\smallwedge A \leq a$.
By Fact \ref{f:rl} it follows that for all $\ga \in \T L$ and $\Psi \in
\T\Pom L$, if $\ga \rel{ \Tb{\in}} \Psi$ then $\T \smallwedge (\Psi) \rel{\Tb
{\leq}} \ga$.
Now suppose that $\Psi$ is a slim redistribution of $\Ga$.
Then $\Gamma \subseteq \nbsem(\Psi)$, and so $(\T\bw)\Psi$ is a
$\Tb{\leq}$-lower bound of $\Ga$.
From this the implication (b) $\Rightarrow$ (a) is immediate.

For the opposite implication, take $\alpha \in \T L$ such that $\forall
\gamma \in \Gamma$, $\alpha \rel{ \Tb {\leq}} \gamma$.
Then by Fact~\ref{f:base}, we obtain $\Base^\T (\alpha) \rel{\Pb {\leq}}
\Base^\T (\gamma)$ for all $\gamma \in \Gamma$.
Abbreviate $C:= \smallcup_{\gamma \in \Gamma} \Base^\T (\gamma)$, and define $f\colon
\Base^\T (\alpha) \to \P C$ as follows:
\[
f\colon a \mapsto \uparr_L a \cap C,
\]
that is: $f(a) = \{ c \in C \mid a \leq c \}$.
Then $\T f$ is a function
\[
\T f\colon \T \Base^\T (\alpha) \to \T \P C.
\]
We claim that $\Psi:= \T f (\alpha)$ is an element of $\SRD (\Gamma) $ and
that $\alpha \rel{\Tb {\leq}} \T \smallwedge (\Psi)$. For the first claim,
since $\Psi \in \T\P C$,
all we need to show is that $\Gamma \subseteq \nbsem (\Psi)$,
i.e.~that for all $\gamma \in \Gamma$, $\gamma \rel{\Tb {\in}}
\Psi$. So suppose that $\gamma \in \Gamma$; then by assumption,
$\alpha \rel{\Tb {\leq}} \gamma$, so $\Base^\T  (\alpha) \rel{\Pb {\leq}}
\Base^\T (\gamma)$. It follows from the definition of $f$ that for all
$b\in \Base^\T (\gamma)$, and all $a \in \Base^\T (\alpha)$, if $a\leq b$ then
$b \in f(a)$. It follows
by Fact \ref{f:rl} that \[ \forall\delta \in \T \Base^\T (\alpha),\,
\forall \beta \in \T \Base^\T (\gamma),\, \delta \rel{\Tb {\leq}} \beta
\Rightarrow \beta \rel{\Tb {\in}} \T f (\delta).\] So in particular, since
$\alpha \in \T \Base^\T (\alpha)$, $\gamma \in \T \Base^\T (\gamma)$ and
$\alpha \rel{\Tb {\leq}} \gamma$, we see that $\gamma \rel{\Tb {\in}}
\T f (\alpha)=\Psi$. Since $\gamma\in \Gamma$ was arbitrary, it
follows that $\Gamma \subseteq \nbsem (\Psi)$. Consequently,
$\Psi\in \SRD(\Gamma)$, as we wanted to show.

For the second
claim, i.e.~that $\alpha \rel{\Tb {\leq}} \T \smallwedge (\Psi)$, it
suffices to observe that $a \leq \smallwedge f(a)$ for all $a\in \Base^\T  (\alpha)$,
so by Fact \ref{f:rl}, \[ \forall\delta \in \T \Base^\T (\alpha), \,
\delta \rel{\Tb {\leq}} \T \smallwedge \cof \T f (\delta).\] Since
$\alpha \in \T \Base^\T (\alpha)$ and
$\Psi = \T f (\alpha)$, we get that $\alpha \rel{\Tb {\leq}} \T
\smallwedge \cof \T f (\alpha) = \T \smallwedge (\Psi)$.
\end{proof}

\begin{corollary}\label{VT:Nabla2prime}
Let $\T\colon \Set \to \Set$ be a standard, finitary, weak
pullback-preserving functor and let $\bbL$ be a frame. Then \[ \V_\T \mathbb{L} \simeq \Fr \struc{ \T
L \mid (\nabla 1), (\nabla 2'), (\nabla 3) },\] where the relations
are as follows:
\[ \begin{array}{lll}
(\nabla 1) & \nabla \alpha \leq \nabla \beta, &(\alpha \rel{ \Tb {\leq
} }
\beta)\\
(\nabla 2') & \smallwedge_{\gamma\in \Gamma} \nabla \gamma \leq \smallvee
\{\nabla \alpha \mid \forall \gamma \in \Gamma,\, \alpha \Tleq \gamma \},&(\Gamma
\in \Pom\T L)\\
(\nabla 3) & \nabla (\T \smallvee) \Phi \leq \smallvee \{ \nabla \beta
\mid \beta \rel{ \Tb {\in}} \Phi \},& (\Phi \in \T \P L)
\end{array} \]
\end{corollary}
\begin{proof}
Observe that the only difference between $\Fr \struc{ \T
L \mid (\nabla 1), (\nabla 2'), (\nabla 3) }$ and the original
definition of $\VT \bbL$ is that we replaced $(\nabla 2)$,
\[ \begin{array}{lll}
(\nabla 2) & \smallwedge_{\alpha\in \Gamma} \nabla \alpha \leq \smallvee
\{\nabla (\T \smallwedge ) \Psi \mid \Psi \in \SRD(\Gamma) \},&(\Gamma
\in \Pom\T L)
\end{array} \]
with
$(\nabla 2')$. The equivalence of these two relations is an immediate
corollary of Lemma \ref{Flatsite:LowerBounds}: take any $\Gamma \in \T
\Pom L$, then
\begin{align*}
&\smallvee \{\nabla \T \smallwedge  (\Psi) \mid \Psi \in \SRD(\Gamma) \}\\
&= \smallvee \{ \nabla \alpha \mid \exists  \Psi \in \SRD(\Gamma),\,
\alpha \Tleq \nabla \T \smallwedge  (\Psi) \}&&\text{by order theory and
$(\nabla 1)$,}\\
&=  \smallvee
\{\nabla \alpha \mid \forall \gamma \in \Gamma,\, \alpha \Tleq \gamma
\} && \text{by Lemma \ref{Flatsite:LowerBounds}.}
\end{align*}
It follows that $\V_\T \mathbb{L} \simeq \Fr \struc{ \T
L \mid (\nabla 1), (\nabla 2'), (\nabla 3) }$.
\end{proof}

\begin{remark}
We will see later that both axioms $(\nabla 2)$ and $(\nabla 2')$ are
equally useful. It seems that $(\nabla 2')$ has not been studied
before in the literature on coalgebraic modal logic via the $\nabla$-modality \citep{PV2007,BPV2008,kiss:comp09,kupk:comp10}.
\end{remark}

\subsection{Basic properties of the $\T$-powerlocale}
\label{Subsect:BasicObs}

In this subsection we make some technical observations about slim
redistributions and about the structure of the $\T$-powerlocale.
We start with two facts on slim redistributions.

\begin{lemma}
\label{f:srd1}
Let $\T\colon \Set \to \Set$ be a standard, finitary, weak
pullback-preserving functor. Then
$\SRD(\nada) = \T\{\nada\}$.
\end{lemma}

\begin{proof}
If $\Phi$ is a slim redistribution of the empty set, then by definition
$\Phi \in \T\Pom(\nada) = \T \{ \nada \}$.
Conversely, any $\Phi \in \T \{ \nada \}$ satisfies the condition that
$\nada \sse \nbsem(\Phi)$, and so $\Phi \in \SRD(\nada)$.
\end{proof}

The following Lemma plays an essential role when defining $\VT$ on
frame homomorphisms, rather than just on frames. It is of crucial use
when showing that if $f\colon \bbL \to \bbM$ is a frame homomorphism,
then $\VT f\colon \VT \bbL \to \VT \bbM$ preserves conjunctions, as we
will see in \S \ref{Section:VT:functorial}.
\begin{lemma}
\label{l:bl1}
Let $\T\colon \Set \to \Set$ be a standard, finitary, weak pullback-preserving functor,
let $X, Y$ be sets and let $f\colon X \to Y$ be a
function; let $\Gamma \in \Pom \T X$. Then the restriction of $\T\Pom f
\colon \T\Pom X \to \T\Pom Y$ to $\SRD(\Gamma)$ is a surjection onto
$\SRD(\Pom\T f \Gamma)$.
\end{lemma}

\begin{proof}
Let $X,Y,f$ and $\Ga$ be as in the statement of the Lemma, and abbreviate
$\Ga' := (\Pom\T f) \Ga$, $C:=\bigcup_{\ga\in\Ga}\Base^\T (\ga)$ and
$C':= \bigcup_{\ga'\in\Ga'}\Base^\T (\ga')$.
Then an easy calculation shows that
\begin{align*}
C' &= \bigcup_{\ga\in\Ga} \Base^\T (\T f)(\ga)
& \text{(definition of $\Ga'$)}
\\ &= \bigcup_{\ga\in\Ga} (\P f)\Base^\T (\ga)
& \text{($\Base^\T $ is natural transformation)}
\\ &= (\P f)(C)
& \text{(elementary set theory)}
\end{align*}

We will first show that $\T\Pom f$ maps slim redistributions of $\Ga$ to
slim redistributions of $\Ga'$.
For that purpose, take an arbitrary element $\Phi \in \SRD(\Ga)$, and
write $\Phi' := (\T\Pom f)\Phi$.
We claim that $\Phi' \in \SRD(\Ga')$, and first show that
\begin{equation}
\label{eq:srd1a}
\Phi' \in \T\Pom C',
\end{equation}
or equivalently, that $\Base^\T \Phi' \sse \Pom C'$.
To prove this inclusion, take an arbitrary set $A' \in \Base^\T (\Phi')$.
Since by Fact~\ref{f:base}, $\Base^\T (\Phi') = (\Pom\Pom f)(\Base^\T (\Phi)$,
this means that $A'$ must be of the form $(\Pom f)(A)$ for some $A \in
\Base^\T (\Phi)$.
In particular, $A'$ must be a subset of $(\Pom f)(\bigcup\Base^\T (\Phi))$.
Also, because $\Phi$ is a slim redistribution of $\Ga$, by definition we
have $\Base^\T (\Phi) \sse \Pom C$, and so $\bigcup \Base^\T (\Phi) \sse \bigcup C$.
From this it follows that $A' \sse (\P f)(\bigcup \Base^\T (\Phi)) \sse
(\P f)(\bigcup C) = C'$, as required.

Second, we claim that
\begin{equation}
\label{eq:srd1b}
\Ga' \sse \nbsem(\Phi').
\end{equation}
To prove this, take an arbitrary element of $\Ga'$, say, $(\T f)\ga$ for
some $\ga \in \Ga$.
We have $\ga
\rel{\Tb{\in}} \Phi$ by the assumption that $\Phi\in\SRD(\Ga)$.
But then, since $a \in A$ implies $fa \in (\Pom f)A$ for any $a \in C$ and
$A \sse C$, it follows by Lemma~\ref{p:rl2} that $\ga' = (\T f)\ga
\rel{\Tb {\in}} (\T\Pom f)(\Phi) = \Phi'$.
This means that $\ga'$ is a lifted member of $\Phi'$, as required.

Clearly, the claims \eqref{eq:srd1a} and \eqref{eq:srd1b} above suffice to
prove that $\Phi' \in \SRD(\Ga')$, which means that indeed, $\T\Pom f$ maps
slim redistributions of $\Ga$ to slim redistributions of $\Ga'$.

Thus it is left to prove that every slim redistribution of $\Ga'$ is of
the form $(\T\Pom f)\Phi$ for some slim redistribution $\Phi$ of $\Ga$.
Take an arbitrary $\Phi' \in \SRD(\Ga')$, and recall that $\Q$ denotes the
contravariant power set functor.
Restrict $f$ to the map $f^{-}\colon  C \to C'$, which means that $\Q f^{-}\colon
\Pom C' \to \Pom C$.
It follows that $\T\Q f^{-}\colon  \T\Pom C' \to \T\Pom C$, so that we may
define $\Phi := (\T\Q f^{-}) \Phi'$, and obtain $\Phi \in \T\Pom C$.
Hence, in order to prove that
\begin{equation}
\label{eq:srd2}
\Phi \in \SRD(\Ga),
\end{equation}
it suffices to show that $\Ga\sse\nbsem(\Phi)$.
But this is an immediate consequence of the fact that $\nbsem$ is a
distributive law of $\T$ over $\Q$ (Fact~\ref{f:db}), since for an
arbitrary $\ga \in \Ga$ we may reason as follows.
From $\ga \in \Ga$ it follows by definition of $\Ga'$ that $(\T f^{-})(\ga)
= (\T f)(\ga)$ belongs to $\Ga'$.
Since $\Ga' \sse\nbsem_{Y}(\Phi)$ by assumption, by definition of $\Q$ we
find that $\ga \in (\Q\T f)\nbsem_{Y}(\Psi)$.
But by $\nbsem\colon  \T\Q \to \Q\T$ we know that $(\Q\T f)\nbsem_{Y}(\Psi) =
\nbsem_{X}(\T\Q f)(\Psi) = \nbsem_{X}(\Phi)$.
Thus we find $\ga \in \nbsem(\Phi)$, as required.

Finally, observe that $f^{-}\colon  C \to C'$ is surjective, so that it follows
by properties of the co- and contravariant power set functors that
$\Pom f^{-} \cof \Q f^{-} = \id_{\Pom C'}$.
From this it is immediate by functoriality of $\T$ that
\[
\Phi' = (\T\Pom f^{-} \cof \T\Q f^{-}) \Phi' = (\T\Pom f^{-}) \Phi
= (\T\Pom f) \Phi.
\]
This finishes the proof of the Lemma.
\end{proof}

In the following lemma we gather some basic observations on the
frame structure of the $\T$-powerlocale.
These facts generalize results from~\citep{kupk:comp10} to our geometrical
setting.

\begin{lemma}
\label{p:basicV}
Let $\T$ be a standard, finitary, weak pullback-preserving functor and let
$\bbL$ be a frame.
\begin{enumerate}
\item  If $\al \in \T L$ is such that $0_{\bbL} \in \Base^\T (\al)$, then
$\nb\al = 0_{\VT\bbL}$.
\item   If $A \sse L$ is such that $a \land b = 0_{\bbL}$ for all $a \neq b$
in $A$, then $\nb\al \land \nb\be = 0_{\VT\bbL}$ for all $\al\neq\be$ in
$\T A$.
\item  If there is no relation $R$ such that $\al \rel{\Tb R} \be$, then
$\nb\al\land\nb\be = 0_{\VT\bbL}$.
\item  $1_{\VT\bbL} = \bv \{ \nb\ga \mid \ga \in \T\{ 1_{\bbL} \} \}$.
\item  For any $A \sse L$ such that $1_{\bbL} = \bv A$, we have $1_{\VT\bbL}
= \bv \{ \nb \al \mid \al \in \T A \}$.
\end{enumerate}
\end{lemma}

\begin{proof}
For part~1, let $\al \in \T L$ be such that $0_{\bbL} \in \Base^\T (\al)$.
Consider the map $f\colon  L \to \P L$ given by
\[
f(a) :=
\left\{\begin{array}{ll}
\emptyset & \text{if } a = 0_{\bbL},
\\ \{ a\} & \text{if } a > 0_{\bbL}.
\end{array}\right.
\]
Then $\id_{L} = \bv \circ f$, so that $\id_{\T L} = (\T\bv)
\circ (\T f)$ by functoriality of $\T$.
In particular, we obtain that $\al = (\T\bv)(\T f)(\al)$, so that we may
calculate
\begin{align*}
\nb\al
&=  \bv \Big\{ \nb\be \mid \be \rel{\Tb{\in}} (\T f)(\al) \Big\}
&& \text{(axiom $\nb2$)}
\\&\leq \bv \Big\{ \nb\be \mid
\Base^\T (\be) \rel{\Pb{\in}} \Base^\T ((\T f)(\al)) \Big\}
&& \text{(Fact~\ref{f:base}(2))}
\\&= \bv \emptyset
&& (\dag)
\\&= 0_{\VT\bbL}
\end{align*}

In order to justify the remaining step (\dag) in this calculation, observe
that it follows from the naturality of $\Base^\T $ (Fact~\ref{f:base}(1)) that
\[
\Base^\T ((\T f)(\al)) = (\P f)(\Base^\T (\al)),
\]
and so by the assumption that $0_{\bbL} \in \Base^\T (\al)$ we obtain $\emptyset \in
\Base^\T ((\T f)(\al))$.
Now suppose for contradiction that there is some $B \sse L$ such that
$B \rel{\Pb{\in}} \Base^\T ((\T f)(\al))$.
Then by definition of $\Pb$ there is a $b \in B$ such that $b \in \emptyset$,
which provides the desired contradiction.
This proves (\dag), and finishes the proof of part~1.
\medskip

For part~2, let $A \sse L$ be such that $a \land b = 0_{\bbL}$ for all
$a \neq b$ in $A$, and take two distinct elements $\al,\be \in \T A$.
In order to prove that $\nb\al \land \nb\be = 0_{\VT\bbL}$, it suffices
by axiom ($\nb2$) to show that
\begin{equation}
\label{eq:bo1}
\nb (\T\bwsmall)(\Phi) = 0_{\VT\bbL}, \mbox{ for all } \Phi \in
\SRD\{\al,\be\}.
\end{equation}
Take an arbitrary slim redistribution $\Phi$ of $\{ \al, \be\}$, then by
Fact~\ref{f:rlnt}, $\Base^\T (\Phi)$ contains a set $A_{0} \sse_{\om} A$ of size
$> 1$.
Define the map $d\colon  \Base^\T (\Phi) \to \Pom(A) \cup \left\{\left\{ 1_{\bbL}
\right\}\right\}$ by putting:
\[
d(B) := \left\{\begin{array}{ll}
\emptyset & \mbox{ if } |B| > 1,
\\ B         & \mbox{ if } |B| = 1,
\\ \{ 1_{\bbL} \} & \mbox{ if } |B| = 0.
\end{array}\right.
\]
It is straightforward to verify from the assumptions on $A$ and the
definition of $d$, that $\bw B \leq \bv d(B)$, for each $B \in \Base^\T (\Phi)$.
Hence it follows by Fact~\ref{f:rl} that $(\T\bw)(\Phi) \rel{\Tb{\leq}}
(\T\bv)(\T d)(\Phi)$, so that by axiom ($\nb1$) we may conclude that
\begin{equation}
\label{eq:bo2}
\nb(\T\bwsmall)(\Phi) \leq \nb(\T\bvsmall)(\T d)(\Phi)
\end{equation}

Finally, it follows from the naturality of $\Base^\T $ (Fact~\ref{f:base}(1))
that $\Base^\T (\T d)(\Phi) = (\P d)(\Base^\T (\Phi))$.
Consequently, for the set $A_{0} \in \Base^\T (\Phi)$ satisfying $|A_{0}| > 1$,
we find $\emptyset = d(A_{0}) \in \Base^\T (\T d)(\Phi)$, and then
$0_{\bbL} = \bv\nada \in (\P \bv) \Base^\T (\T d)(\Phi) = \Base^\T (\T\bv)(\T d)(\Phi)$.
Thus by part~(1) of this lemma it follows that
\begin{equation}
\label{eq:bo3}
\nb(\T\bvsmall)(\T d)(\Phi) = 0_{\VT\bbL}.
\end{equation}
This finishes the proof of part~2, since \eqref{eq:bo1} is immediate on
the basis of~\eqref{eq:bo2} and~\eqref{eq:bo3}.
\medskip

In order to prove part~3, suppose that $\al,\be \in \T L$ are not linked
by any lifted relation.
Consider the (unique) map
\[
f\colon  L \to \{ 1 \},
\]
and define $\al' := (\T f)\al$, $\be' := (\T f)(\be)$.
Suppose for contradiction that $\al'=\be'$.
Then we would find $\al \rel{\Tb (\cv{(\Graph{f})} \cor \Graph{f})} \be$,
contradicting the assumption on $\al$ and $\be$.
It follows that $\al'$ and $\be$ are distinct, and so by part (2) of this
lemma (with $A = \{ 1_{\bbL} \}$), we may infer that $\nb\al' \land
\nb\be' = 0_{\VT\bbL}$.
This means that we are done, since it follows from $\Graph{f} \sse {\leq}$
and the definitions of $\al',\be'$, that $\al \rel{\Tb{\leq}} \al'$ and
$\be \rel{\Tb{\leq}}\be'$, and from this we obtain by ($\nb$1) that
\[
\nb\al \land \nb\be \leq
\nb\al' \land \nb\be'\leq
0_{\VT\bbL}.
\]

For part~4, we reason as follows:
\begin{align*}
1_{\VT\bbL} =&
\bv \{ \nb (\T\bwsmall)(\Phi) \mid \Phi \in \SRD(\emptyset) \},
&& \text{(axiom ($\nb$2) with $A =\emptyset$)}
\\ =& \bv \{ \nb (\T\bwsmall)(\Phi) \mid \Phi \in \T\{\emptyset\} \}
&& \text{(Fact~\ref{f:srd1})}
\\ =& \bv \{ \nb\ga \mid \ga \in \T\{ 1_{\bbL}\} \}
&& \text{(\ddag)}
\end{align*}
where the last step (\ddag) is justified by the observation that, since
the map $\bw\colon  \Pom L \to L$ restricts to a bijection $\bw\colon  \{ \emptyset\}
\to \{ 1_{\bbL}\}$, its lifting restricts to a bijection $\T\bw\colon
\T\{\emptyset\} \to \T \{1_{\bbL}\}$.
\medskip

Finally, we turn to the proof of part~5.
Let $A \sse L$ be such that $1_{\bbL} = \bv A$, and consider an arbitrary
element $\Phi \in \T \{ A \}$.
We claim that
\begin{equation}
\label{eq:bo3a}
\nbsem(\Phi) \sse \T A.
\end{equation}
To see this, take an arbitrary lifted element $\al$ of $\Phi$.
It follows from $\al \rel{\Tb{\in}} \Phi$ that $\Base^\T (\al) \rel{\Pb{\in}}
\Base^\T (\Phi)$.
In particular, each $a \in \Base^\T (\al)$ must belong to some $B \in \Base^\T (\Phi)
\sse \{ A \}$.
In other words, $\Base^\T (\al) \sse A$, which is equivalent to saying that
$\al \in \T A$.
This proves \eqref{eq:bo3a}.

By \eqref{eq:bo3a} and axiom ($\nb3$) we obtain
\begin{equation}
\label{eq:bo4}
\nb(\T\bvsmall)(\Phi) \leq \bv \{ \nb\al \mid \al \in \T A \}.
\end{equation}
Now we reason as follows:
\begin{align*}
1_{\VT\bbL}
=   & \bv \{ \nb\al \mid \al \in \T \{ 1_{\bbL} \} \}
&& \text{(part 4)}
\\ =   & \bv \{ \nb(\T\bvsmall)(\Phi) \mid \Phi \in \T \{ A \} \}
&& \text{($\ast$)}
\\ \leq& \bv \{ \nb\al \mid \al \in \T A\},
&& \text{\eqref{eq:bo4}}
\end{align*}
To justify the second step ($\ast$), observe that if we restrict the map
$\bv\colon  \P L \to L$ to the bijection $\bv\colon  \{ A \} \to \{ 1_{\bbL}\}$, as its
lifting we obtain a bijection $\T\bv\colon  \T\{ A \} \to \T \{1_{\bbL}\}$.
\end{proof}

\subsection{Two examples of the $\T$-powerlocale construction}
\label{Subsect:VP}
\newcommand{\IdSet}{\operatorname{Id}}

In this subsection we will discuss two examples of
$T$-powerlocales. First, we discuss the somewhat trivial example of
the $\IdSet$-powerlocale. After that, we will discuss the defining
example of $\T$-powerlocales, namely the $\Pom$-powerlocale, which is
isomorphic to the classical Vietoris powerlocale.

\begin{example}
Let $\IdSet \colon \Set \to \Set$ be the identity functor on the
category of sets. Then for all frames $\bbL$, $\V_{\IdSet} \bbL \simeq
\bbL$.

First recall from Example \ref{RelLift:Examples} that for any relation
$R \subseteq X\times Y$, $\rl{\IdSet} R = R$. Moreover, if $A \in
\IdSet \Pom L= \Pom L$, then it is  straightforward to verify that
\begin{align*}
\SRD (A) &= \{ \Psi \in \Pom (\smallcup_{c \in A} \{ c\})
\mid \forall c \in A, c \in \Psi\}\\
&= \{ A\}.
\end{align*}
Consequently, the
$\nabla$-relations reduce to the following in case $\T = \IdSet$:
\[ \begin{array}{lll}
(\nabla 1) & \nabla a \leq \nabla b, &(a  \leq b)\\
(\nabla 2) & \smallwedge_{a \in A} \nabla a \leq
\nabla \smallwedge A,&(A
\in \Pom L)\\
(\nabla 3) & \nabla \smallvee A \leq \smallvee \{ \nabla b
\mid b \in A \}.& (A \in \P L)
\end{array} \]
The identity $\id_L \colon L \to L$ obviously satisfies $(\nabla 1)$,
$(\nabla 2)$ and $(\nabla 3)$. Moreover if we have a frame $\bbM$ and
a function $f\colon
L \to M$ which is compatible with $(\nabla 1)$,
$(\nabla 2)$ and $(\nabla 3)$, then it is easy to see that $f$ is in
fact a frame homomorphism $\bbL \to \bbM$. By the universal property
of frame presentations, it follows that
$\V_{\IdSet} \bbL \simeq \bbL$.
\end{example}

We now turn to the $\Pom$-powerlocale.
Recall from Example \ref{WpbExamples} that $\Pom\colon\Set\to\Set$,
the covariant finite power set functor, is indeed standard,
weak pullback-preserving and finitary.
We will now show that the $\Pom$-powerlocale is the  Vietoris
powerlocale. The equivalence of the $\nabla$ axioms and the $\Box$,
$\Diamond$ axioms on distributive lattices is already known from the
work of \citet{PV2007}; what is different here is that we
consider infinite joins rather than only finite joins.

We will use the presentation using $(\nabla 1)$, $(\nabla 2')$ and
$(\nabla 3)$ as our point of departure.
Recall that for all $\alpha,\beta \in \Pom L$,
\begin{align*}
\alpha \leq_L \beta &\text{ if } \alpha \subseteq \downarr \beta,\\
\alpha \leq_U \beta &\text{ if } \uparr \alpha \supseteq \beta,\\
\alpha \leq_C \beta &\text{ if } \alpha \leq_L \beta \text{ and }
\alpha \leq_U \beta.
\end{align*}
By Example \ref{SRD:Pom},
two of the relations presenting $\V_{\Pom}\mathbb{L}$ thus become%
\[%
\begin{array}
[c]{ll}%
(\nb2') & \bigwedge_{\gamma\in\Gamma}\nb\gamma
\leq\bigvee\left\{  \nb \alpha\mid  \forall \gamma\in \Gamma,\, \alpha
\leq_C \gamma\right\}  \\
(\nb3) & \nb\left\{  \bigvee\alpha\mid\alpha\in
\Phi\right\}  \leq\bigvee\left\{  \nb\beta\mid\beta\in\Pom\left(
\bigcup\Phi\right)  \text{ and }\forall\alpha\in\Phi,\,\alpha\between
\beta\right\}
\end{array}
\]

\begin{lemma}
We consider the presentation of $\V_{\Pom}\mathbb{L}$.

\begin{enumerate}

\item  In the presence of $(\nb1)$, the relation
$(\nb2^{\prime})$ can be replaced by%
\[%
\begin{array}
[c]{ll}%
(\nb2.0) & 1\leq\bigvee\left\{  \nb\beta\mid\beta
\in\Pom L\right\}  \\
(\nb2.2) & \nb\gamma_{1}\wedge\nb
\gamma_{2}\leq\bigvee\left\{  \nb\beta\mid\beta\leq_{C}\gamma
_{1},\beta\leq_{C}\gamma_{2}\right\}
\end{array}
\]

\item  In the presence of $(\nb1)$ and $(\nb2)$ (or
its equivalent formulations), the relation $(\nb3)$ can be
replaced by%
\[%
\begin{array}
[c]{lll}%
(\nb3.{\uparrow}) & \nb\left(  \gamma\cup
\{\bigvee\nolimits^{\uparrow}S\}\right)  \leq\bigvee\nolimits^{\uparrow
}\left\{  \nb\left(  \gamma\cup\{a\}\right)  \mid a\in S\right\}
\quad \text{(}S\text{ directed)}\\
(\nb3.0) & \nb\left(  \gamma\cup\{0\}\right)
\leq0 & \\
(\nb3.2) & \nb\left(  \gamma\cup\{a_{1}\vee
a_{2}\}\right)  \leq\nb\left(  \gamma\cup\{a_{1}\}\right)
\vee\nb\left(  \gamma\cup\{a_{2}\}\right)  \vee\nb
\left(  \gamma\cup\{a_{1},a_{2}\}\right)   &
\end{array}
\]
\end{enumerate}
\end{lemma}

\begin{proof}

(1) $(\nb2.0)$ and $(\nb2.2)$ are special cases of
$(\nb2^{\prime})$, when $\Gamma$ is empty or a doubleton. To show
that they imply $(\nb2^{\prime})$ is an induction on the number
of elements needed to enumerate the finite set $\Gamma$.%

(2) Each of the replacement relations is a special case of $(\nb
3)$ in which all except one of the elements of $\Phi$ are singletons. We now
show that they are sufficient to imply $(\nb3)$. First, we show
for any \emph{finite} $S$ that%
\[
\nb\left(  \gamma\cup\{\bigvee S\}\right)  \leq\bigvee\left\{
\nb\left(  \gamma\cup\alpha\right)  \mid\emptyset\neq\alpha
\in\Pom S\right\}  \text{.}%
\]
We use induction on the length of a finite enumeration of $S$. The base case,
$S$ empty, is $(\nb3.0)$. Now suppose $S=\{a\}\cup S^{\prime}$.
Then%
\begin{align*}
&\nb\left(  \gamma\cup\{\bigvee S\}\right)\\
&  =\nb \left(  \gamma\cup\{a\vee\bigvee S^{\prime}\}\right)  \\
&  \leq\nb\left(  \gamma\cup\{a\}\right)
\vee\nb \left(  \gamma\cup\{\bigvee S^{\prime}\}\right)
\vee\nb\left(\left(  \gamma\cup\{a\}\right)
\cup\{\bigvee S^{\prime}\}\right)
&& \text{(by }(\nb3.2)\text{)}\\
&  \leq\nb\left(  \gamma\cup\{a\}\right)  \vee\bigvee\left\{
\nb\left(  \gamma\cup\alpha^{\prime}\right)  \mid\emptyset
\neq\alpha^{\prime}\in\Pom S^{\prime}\right\}  \\
&  \qquad\vee\bigvee\left\{  \nb\left(  \gamma\cup\{a\}\cup
\alpha^{\prime}\right)  \mid\emptyset\neq\alpha^{\prime}\in\Pom S^{\prime}\right\}
&& \text{(by induction)}\\
&  =\bigvee\left\{  \nb\left(  \gamma\cup\alpha\right)
\mid\emptyset\neq\alpha\in\Pom S\right\}  \text{.}%
\end{align*}
Now we can use $(\nb3.{\uparrow})$ to relax the finiteness
condition on $S$, since for an arbitrary $S$ we have%
\begin{align*}
\nb\left(  \gamma\cup\{\bigvee S\}\right)   &  =\nb
\left(  \gamma\cup\left\{  \bigvee\nolimits^{\uparrow}\left\{  \bigvee
S_{0}\mid S_{0}\in\Pom S\right\}  \right\}  \right)  \\
&  \leq\bigvee\nolimits^{\uparrow}\left\{  \nb\left(  \gamma
\cup\left\{  \bigvee S_{0}\right\}  \right)  \mid S_{0}\in\Pom S\right\}
\text{.}%
\end{align*}

Finally, we can use induction on the length of a finite enumeration of $\Phi$
to deduce $(\nb3)$. More precisely, one shows by induction on $n$
that%
\begin{gather*}
\nb\left(  \gamma\cup\{\bigvee S_{1},\ldots,\bigvee
S_{n}\}\right) \\ \leq\bigvee\left\{  \nb\left(  \gamma\cup
\alpha\right)  \mid\emptyset\neq\alpha\in\Pom\left(  \bigcup_{i=1}^{n}%
S_{i}\right)  \text{ and }\forall i,\,\alpha\between S_{i}\right\}
\text{.} %
\end{gather*}
\end{proof}

\begin{remark}
Relation $(\nb 2.0)$ can be weakened even further, to%
\[
1\leq\nb\emptyset\vee\nb\{1\}\text{.}%
\]
For if $\beta$ is non-empty %
then $\beta\leq_{C}\{1\}$.
From $(\nb 2.2)$ we can also deduce that
$\nb\emptyset \wedge \nb\{1\} = 0$,
giving that $\nb\emptyset$ and $\nb\{1\}$
are clopen complements.
\end{remark}

\begin{lemma}
\label{VBoxLemma}In $\V\mathbb{L}$ we have, for any $S\subseteq\mathbb{L}$,%
\[
\Box\left(  \bigvee S\right)  =\bigvee\left\{  \Box\left(  \bigvee
\alpha\right)  \wedge\bigwedge_{a\in\alpha}\Diamond a\mid\alpha\in\Pom
S\right\}  \text{.}%
\]
\end{lemma}

\begin{proof}
$\geq$ is immediate. For $\leq$, first note that since $\bigvee S$ is a
directed join $\bigvee_{\alpha\in\Pom S}^{\uparrow}\bigvee\alpha$, we have
$\Box\left(  \bigvee S\right)  \leq\bigvee_{\alpha\in\Pom S}^{\uparrow
}\Box\left(  \bigvee\alpha\right)  $ and thus we reduce to the case where
$S$ is finite. We show that for every $\alpha,\beta\in\Pom S$ we have%
\[
\Box\left(  \bigvee \alpha\vee\bigvee \beta\right)  \wedge\bigwedge_{a\in \alpha}\Diamond
a\leq\text{ RHS in statement,}%
\]
after which the result follows by taking $\beta=S$ and $\alpha=\emptyset$. We use
$\Pom$-induction on $\beta$, effectively an induction on the length of an
enumeration of its elements. The base case, $\beta=\emptyset$, is trivial. For the
induction step, suppose $\beta=\beta^{\prime}\cup\{b\}$. Then%
\begin{align*}
&\Box\left(  \bigvee \alpha\vee\bigvee \beta\right)  \wedge\bigwedge_{a\in \alpha}\Diamond
a\\  &  =\Box\left(  \bigvee \alpha\vee b\vee\bigvee \beta^{\prime}\right)
\wedge\bigwedge_{a\in \alpha}\Diamond a\\
&  =\Box\left(  \bigvee \alpha\vee b\vee\bigvee \beta^{\prime}\right)
\wedge\bigwedge_{a\in \alpha}\Diamond a\wedge\left(  \Box\left(  \bigvee
\alpha\vee\bigvee \beta^{\prime}\right)  \vee\Diamond b\right) \\
&  =\left(  \Box\left(  \bigvee \alpha\vee\bigvee \beta^{\prime}\right)
\wedge\bigwedge_{a\in \alpha}\Diamond a\right)  \vee\left(  \Box\left(
\bigvee\left(  \alpha\cup\{b\}\right)  \vee\bigvee \beta^{\prime}\right)
\wedge\bigwedge_{a\in \alpha\cup\{b\}}\Diamond a\right) \\
&  \leq\text{ RHS, by induction.}%
\end{align*}
\end{proof}

\begin{theorem}
Let $\mathbb{L}$ be a frame. Then $\V\mathbb{L}\cong\V_{\Pom}\mathbb{L}$.
\end{theorem}

\begin{proof}
First we define a frame homomorphism $\phi\colon \V_{\Pom}\mathbb{L}\rightarrow
\V\mathbb{L}$ by $\phi(\nb\alpha)=\Box\left(  \smallvee
\alpha\right)  \wedge\smallwedge_{a\in\alpha}\Diamond a$. We must check that
this respects the relations. For $(\nb1)$, suppose $\alpha
\leq_{C}\beta$. From $\alpha\leq_{U}\beta$ and $\alpha\leq_{L}\beta$ we get
$\smallwedge_{a\in\alpha}\Diamond a\leq\smallwedge_{b\in\beta}\Diamond b$ and
$\smallvee\alpha\leq\smallvee\beta$, giving $\phi(\nb\alpha)\leq
\phi(\nb\beta)$.

For $(\nb2.0)$, we have $1=\Box(0\vee1)=\Box0\vee\left(
\Box1\wedge\Diamond1\right)  =\phi\left(  \nb\emptyset\right)
\vee\phi\left(  \nb\left\{  1\right\}  \right)  $.

For $(\nb2.2)$, $\phi\left(  \nb\gamma_{1}\right)
\wedge\phi\left(  \nb\gamma_{2}\right)  $ is%
\begin{align*}
&  \Box\left(  \smallvee\gamma_{1}\right)  \wedge\smallwedge_{c\in\gamma_{1}%
}\Diamond c\wedge\Box\left(  \smallvee\gamma_{2}\right)
\wedge\smallwedge_{c^{\prime}\in\gamma_{2}}\Diamond c^{\prime}\\
&  =\Box\left(  \smallvee\gamma_{1}\wedge\smallvee\gamma_{2}\right)
\wedge\smallwedge_{c\in\gamma_{1}}\Diamond c\wedge\smallwedge_{c^{\prime}\in
\gamma_{2}}\Diamond c^{\prime}\\
&  =\Box\left(  \smallvee\gamma_{1}\wedge\smallvee\gamma_{2}\right)
\wedge\smallwedge_{c\in\gamma_{1}}\Diamond\left(  c\wedge\smallvee\gamma_{1}%
\wedge\smallvee\gamma_{2}\right)  \wedge\smallwedge_{c^{\prime}\in\gamma_{2}%
}\Diamond\left(  c^{\prime}\wedge\smallvee\gamma_{1}\wedge\smallvee\gamma
_{2}\right)  \\
&  =\Box\left(  \smallvee\gamma_{1}\wedge\smallvee\gamma_{2}\right)
\wedge\smallwedge_{c\in\gamma_{1}}\Diamond\left(  c\wedge\smallvee\gamma
_{2}\right)  \wedge\smallwedge_{c^{\prime}\in\gamma_{2}}\Diamond\left(
c^{\prime}\wedge\smallvee\gamma_{1}\right)  \\
&  =\Box\left(  \smallvee_{c\in\gamma_{1}}\smallvee_{c^{\prime}\in\gamma_{2}%
}c\wedge c^{\prime}\right)  \wedge\smallwedge_{c\in\gamma_{1}}
\smallvee_{c^{\prime}\in\gamma_{2}}\Diamond\left(  c\wedge c^{\prime}\right)
\wedge\smallwedge_{c^{\prime}\in\gamma_{2}}\smallvee_{c\in\gamma_{1}}%
\Diamond\left(  c\wedge c^{\prime}\right)
\end{align*}
Redistributing the disjunctions of the $\Diamond$s, we find that each
resulting disjunct is of the form%
\[
\Box\left(  \smallvee_{c\in\gamma_{1}}\smallvee_{c^{\prime}\in\gamma_{2}}c\wedge
c^{\prime}\right)  \wedge\smallwedge_{cRc^{\prime}}\Diamond\left(  c\wedge
c^{\prime}\right)
\]
for some $R\in\Pom\left(  \gamma_{1}\times\gamma_{2}\right)  $ such that
$\gamma_{1}\,\overline{\Pom}R\,\gamma_{2}$. Note that for any such $R$ if we
define $\beta_{R}=\{c\wedge c^{\prime}\mid cRc^{\prime}\}$ then we have
$\beta_{R}\leq_{C}\gamma_{i}$ ($i=1,2$). Now by Lemma~\ref{VBoxLemma} we see%
\begin{align*}
&  \Box\left(  \smallvee_{c\in\gamma_{1}}\smallvee_{c^{\prime}\in\gamma_{2}%
}c\wedge c^{\prime}\right)  \wedge\smallwedge_{cRc^{\prime}}\Diamond\left(
c\wedge c^{\prime}\right)  \\
&  \leq\smallvee\left\{  \Box\left(  \smallvee_{cR^{\prime}c^{\prime}}c\wedge
c^{\prime}\right)  \wedge\smallwedge_{c(R\cup R^{\prime})c^{\prime}}%
\Diamond\left(  c\wedge c^{\prime}\right)  \mid R^{\prime}\in\Pom\left(
\gamma_{1}\times\gamma_{2}\right)  \right\}  \\
&  \leq\smallvee\left\{  \Box\left(  \smallvee_{cR^{\prime}c^{\prime}}c\wedge
c^{\prime}\right)  \wedge\smallwedge_{cR^{\prime}c^{\prime}}\Diamond\left(
c\wedge c^{\prime}\right)  \mid R\subseteq R^{\prime}\in\Pom\left(  \gamma
_{1}\times\gamma_{2}\right)  \right\}  \\
&  =\smallvee\left\{  \phi\left(  \nb\beta_{R^{\prime}}\right)
\mid R\subseteq R^{\prime}\in\Pom\left(  \gamma_{1}\times\gamma_{2}\right)
\right\}
\end{align*}
and the result follows.

For $(\nb3.{\uparrow})$: the LHS is%
\begin{align*}
&  \Box\left(  \smallvee\gamma\vee\smallvee^{\uparrow}S\right)
\wedge\smallwedge_{c\in\gamma}\Diamond c\wedge
\Diamond\left(  \smallvee^{\uparrow}S\right) \\
&  =\smallvee^{\uparrow}\left\{  \Box\left(  \smallvee\gamma\vee
a\right)  \mid a\in S\right\}  \wedge\smallvee^{\uparrow}\left\{
\smallwedge_{c\in\gamma}\Diamond c\wedge\Diamond a\mid a\in S\right\} \\
&  =\smallvee^{\uparrow}\left\{  \Box\left(  \smallvee\gamma\vee
a\right)  \wedge\smallwedge_{c\in\gamma}\Diamond c\wedge\Diamond a\mid a\in
S\right\}
\end{align*}
which is the RHS.

For $(\nb3.0)$: the LHS is%
\[
\Box\left(  \smallvee\gamma\vee0\right)  \wedge\smallwedge_{c\in\gamma}\Diamond
c\wedge\Diamond0=0\text{.}%
\]

For $(\nb3.2)$: the LHS is%
\begin{align*}
&  \Box\left(  \smallvee\gamma\vee a_{1}\vee a_{2}\right)
\wedge\smallwedge_{c\in\gamma}\Diamond c\wedge\Diamond\left(  a_{1}\vee a_{2}\right) \\
&  =\bigvee_{i=1}^{2}\smallvee\left\{  \Box\left(  \smallvee\beta\right)
\wedge\smallwedge_{c\in\beta\cup\gamma\cup\{a_{i}\}}\Diamond c\mid\beta\in
\Pom\left(  \gamma\cup\{a_{1},a_{2}\}\right)  \right\} \\
&  \leq\bigvee_{i=1}^{2}\smallvee\left\{  \phi\left(  \nb\left(
\beta\cup\gamma\cup\{a_{i}\}\right)  \right)  \mid\beta\in\Pom\left(
\gamma\cup\{a_{1},a_{2}\}\right)  \right\} \\
&  =\bigvee_{i=1}^{2}\smallvee\left\{  \phi\left(  \nb\beta\right)
\mid\gamma\cup\{a_{i}\}\subseteq\beta\in\Pom\left(  \gamma\cup\{a_{1}%
,a_{2}\}\right)  \right\} \\
&  =\phi\left(  \nb\left(  \gamma\cup\{a_{1}\}\right)  \right)
\vee\phi\left(  \nb\left(  \gamma\cup\{a_{2}\}\right)  \right)
\vee\phi\left(  \nb\left(  \gamma\cup\{a_{1},a_{2}\}\right)
\right)  \text{.}%
\end{align*}

Next, we define the frame homomorphism $\psi\colon \V\mathbb{L}\rightarrow\V_{\Pom
}\mathbb{L}$ by%
\begin{align*}
\psi\left(  \Box a\right)   &  =\smallvee\left\{  \nb\alpha
\mid\alpha\leq_{L}\{a\}\right\}  =\nb\emptyset\vee
\nb\{a\}\\
\psi\left(  \Diamond a\right)   &  =\smallvee\left\{  \nb\alpha
\mid\alpha\leq_{U}\{a\}\right\}  =\smallvee\left\{  \nb\left(
\beta\cup\{a\}\right)  \mid\beta\in\Pom L\right\}  \text{.}%
\end{align*}
(Observe that the expression for $\psi\left(  \Diamond a\right)$ could
be simplified even further to $\nabla\{ 1,a\}$.)
We check the relations. First, it is clear that $\psi$ respects monotonicity
of $\Box$ and $\Diamond$.

$\Box$ preserves directed joins:%
\[
\psi\left(  \Box\left(  \smallvee_{i}^{\uparrow}a_{i}\right)  \right)
=\nb\emptyset\vee\nb\{\smallvee_{i}%
^{\uparrow}a_{i}\}=\smallvee_{i}^{\uparrow}\psi\left(  \Box
a_{i}\right)  \text{.}%
\]

$\Box$ preserves top immediately from $(\nb2.0)$.

$\Box$ preserves binary meets:%
\begin{align*}
\psi\left(  \Box a_{1}\right)  \wedge\psi\left(  \Box a_{2}\right)   &
=\nb\emptyset\vee\left(  \nb\{a_{1}\}\wedge
\nb\{a_{2}\}\right)  \\
&  =\nb\emptyset\vee\smallvee\left\{  \nb\beta
\mid\beta\leq_{C} \{a_{1}\}, \,\beta \leq_C \{a_2\}\right\}  \\
&  =\nb\emptyset\vee\nb\{a_{1}\wedge a_{2}%
\}=\psi\left(  \Box\left(  a_{1}\wedge a_{2}\right)  \right)  \text{.}%
\end{align*}

$\Diamond$ preserves joins:%
\begin{align*}
\psi\left(  \Diamond\left(  \smallvee A\right)  \right)    & =\smallvee\left\{
\nb\left(  \beta\cup\{\smallvee A\}\right)  \mid\beta\in
\Pom L\right\}  \\
& =\smallvee\left\{  \nb\left(  \beta\cup\alpha\right)  \mid
\beta\in\Pom L,\emptyset\neq\alpha\in\Pom A\right\}  \\
& =\smallvee_{a\in A}\smallvee\left\{  \nb\left(  \beta
\cup\{a\}\right)  \mid\beta\in\Pom L\right\}  =\smallvee_{a\in A}%
\psi(\Diamond a)\text{.}%
\end{align*}

For the first mixed relation, and noting that $\nb\emptyset
\wedge\nb\left(  \beta\cup\{b\}\right)  \leq\nb
\emptyset\wedge\nb\{1\}=0$, we have:%
\begin{align*}
\psi\left(  \Box a\right)  \wedge\psi\left(  \Diamond b\right)   &
=\smallvee_{\beta\in\Pom L}\left(  \nb\emptyset
\vee\nb\{a\}\right)  \wedge\nb\left(  \beta
\cup\{b\}\right)  \\
&  =\smallvee_{\beta\in\Pom L}\nb\{a\}\wedge
\nb\left(  \beta\cup\{b\}\right)  \\
&  =\smallvee\left\{  \nb\gamma\mid\exists\beta,\,\gamma\leq
_{C}\{a\},\gamma\leq_{C}\beta\cup\{b\}\right\}  \\
&  \leq\smallvee_{\beta\in\Pom L}\nb\left(  \beta
\cup\{a\wedge b\}\right)  =\psi\left(  \Diamond\left(  a\wedge b\right)
\right)  \text{.}%
\end{align*}

For the second:%
\begin{align*}
\psi\left(  \Box\left(  a\vee b\right)  \right)   &  =\nb
\emptyset\vee\nb\{a\vee b\}\\
&  =\nb\emptyset\vee\nb\{a\}\vee\nb
\{b\}\vee\nb\{a,b\}\\
&  \leq\psi\left(  \Box a\right)  \vee\psi\left(  \Diamond b\right)
\end{align*}
since $\nb\emptyset\vee\nb\{a\}=\psi\left(  \Box
a\right)  $ and $\nb\{b\}\vee\nb\{a,b\}\leq
\psi\left(  \Diamond b\right)  $.

It remains to show that $\phi$ and $\psi$ are mutually inverse.%
\[
\phi\left(  \psi\left(  \Box a\right)  \right)  =\phi\left(  \nb
\emptyset\vee\nb\{a\}\right)  =\Box0\vee\left(  \Box
a\wedge\Diamond a\right)  =\Box a
\]
since $\Box0\wedge\Diamond a\leq\Diamond\left(  0\wedge a\right)  =0$.

Next, to show $\phi\left(  \psi\left(  \Diamond a\right)  \right)  =\Diamond
a$, we have
\begin{align*}
\phi\left(  \psi\left(  \Diamond a\right)  \right)    & =\smallvee_{\beta\in
\Pom L}\left(  \Box\left(  \smallvee\beta\vee a\right)  \wedge
\smallwedge_{b\in\beta}\Diamond b\wedge\Diamond a\right)  \\
& \leq\Diamond a\\
& =\Box\left(  1\vee a\right)  \wedge\Diamond1\wedge\Diamond a=\phi\left(
\nb\{1,a\}\right)  \leq\phi\left(  \psi\left(  \Diamond a\right)
\right)  \text{.}%
\end{align*}

Finally, to show $\psi\left(  \phi\left(  \nb\alpha\right)
\right)  =\nb\alpha$, we have
\begin{align*}
\psi\left(  \phi\left(  \nb\alpha\right)  \right)   &
=\psi\left(  \Box\left(  \smallvee\alpha\right)
\wedge\smallwedge_{a\in\alpha}\Diamond a\right)  \\
&  =\left(  \nb\emptyset\vee\nb\{\smallvee
\alpha\}\right)  \wedge\smallwedge_{a\in\alpha}\smallvee_{\beta_{a}\in
\Pom L}\nb\left(  \beta\cup\{a\}\right)  \text{.}%
\end{align*}
Now,%
\begin{align*}
\smallwedge_{a\in\alpha}\smallvee_{\beta_{a}\in\Pom L}\nb
\left(  \beta\cup\{a\}\right)    & =\smallvee\left\{  \nb\gamma
\mid\forall a\in\alpha,\,\exists\beta_{a}\in\Pom L,\,\gamma\leq
_{C}\beta_{a}\cup\{a\}\right\}  \\
& =\smallvee\left\{  \nb\gamma\mid\gamma\leq_{U}\alpha\right\}
\text{.}%
\end{align*}
Also%
\begin{align*}
\nb\emptyset\wedge\smallvee\left\{  \nb\gamma
\mid\gamma\leq_{U}\alpha\right\}    & =\smallvee\left\{  \nb
\delta\mid\delta\leq_{C}\emptyset,\delta\leq_{U}\alpha\right\}  \\
& =\left\{
\begin{array}
[c]{ll}%
\nb\alpha & \text{if }\alpha=\emptyset\\
0 & \text{if }\alpha\neq\emptyset
\end{array}
\right.  \\
\nb\left\{  \smallvee\alpha\right\}  \wedge\smallvee\left\{
\nb\gamma\mid\gamma\leq_{U}\alpha\right\}    & =\smallvee\left\{
\nb\delta\mid\delta\leq_{C}\left\{  \smallvee\alpha\right\}
,\delta\leq_{U}\alpha\right\}  \\
& =\nb\left(  \alpha\cup\left\{  \smallvee\alpha\right\}  \right)
\\
& =\smallvee\left\{  \nb\left(  \alpha\cup\alpha^{\prime}\right)
\mid\emptyset\neq\alpha^{\prime}\in\Pom\alpha\right\}  \\
& =\left\{
\begin{array}
[c]{ll}%
0 & \text{if }\alpha=\emptyset\\
\nb\alpha & \text{if }\alpha\neq\emptyset
\end{array}
\right.
\end{align*}
It follows that, whether $\alpha$ is empty or not, $\psi\left(  \phi\left(
\nb\alpha\right)  \right)  =\nb\alpha$.
\end{proof}

\subsection{Categorical properties of the $T$-powerlocale}
\label{Section:VT:functorial}

In this section we discuss two categorical properties of the $T$-powerlocale
construction.
First we show how to extend the frame construction $\VT$ to an endofunctor on
the category $\class{Fr}$ of frames.
As a second topic we will see how the natural transformation $\counit\colon  \VP \to
\V_{\Fid}$ (discussed in \S \ref{Nabla:BoxDiamond} as $\counit\colon  \V \to \Fid$) can be generalized to
a natural transformation
\[
\wh{\nta}\colon  \VT \to \V_{\T'},
\]
for any natural transformation $\nta\colon  \T'\to \T$ satisfying some mild
conditions (where $\T$ and $\T'$ are two finitary, weak pullback preserving
set functors).

\subsubsection{$\VT$ is a functor}
We start with introducing a natural way to transform a frame
homomorphism $f\colon \mathbb{L}\to \mathbb{M}$ into a frame
homomorphism from $\VT\bbL$ to $\VT\bbM$.
For that purpose we first prove the following technical lemma.

\begin{lemma}
\label{p:srd2}
Let $\T\colon \Set \to \Set$ be a standard, finitary, weak pullback-preserving functor,
let $\mathbb{L},\mathbb{M}$ be frames and let $f\colon \mathbb{L}\to
\mathbb{M}$ be a frame homomorphism. Then the map $\nabla \cof \T f
\colon \T L \to \V_\T M$, i.e.~$\alpha \mapsto \nabla (\T f)
(\alpha)$, is compatible with the relations $(\nabla 1)$, $(\nabla 2)$
and $(\nabla 3)$.
\end{lemma}

\begin{proof}
We abbreviate $\hs := \nb\circ \T f$, that is, for $\al \in \T L$, we
define $\hs\al := \nb(\T f)(\al)$.

In order to prove that $\hs$ is compatible with $(\nb1)$, we need to show
that
\begin{equation}
\label{eq:comp1}
\text{for all } \al,\be \in \T L\colon  \;\al \rel{\Tb {\leq_{\bbL}}} \be
\text{ implies } \hs\al \leq_{\VT\bbM} \hs\be.
\end{equation}
To see this, assume that $\al,\be \in \T L$ are such that $\al
\rel{\Tb {\leq_{\bbL}}} \be$.
From this it follows by Lemma~\ref{p:rl2} and the assumption that $f$
is a frame homomorphism, that $(\T f) (\al) \rel{\Tb {\leq_{\bbM}}} (\T f) (\be)$.
Then by $(\nb1)_{\bbM}$ we obtain that $\hs\al \leq_{\VT\bbM} \hs\be$, as
required.

Proving compatibility with $(\nb2)$ boils down to showing
\begin{equation}
\label{eq:comp2}
\text{for all } \Ga \in \Pom\T L\colon
\bw_{\al\in\Ga} \hs\al \leq
\bv \{ \hs(\T\bwsmall) (\Psi) \mid \Psi \in \SRD(\Ga) \}.
\end{equation}
For this purpose, given $\Ga \in \Pom\T L$, let $\Ga' \in \Pom\T M$ denote
the set $\Ga' := (\Pom\T f) (\Ga) = \{ (\T f) (\al) \mid \al \in \Ga \}$.
Then we may observe
\begin{align*}
\bw_{\al\in\Ga} \hs\al
&= \bv \{ \nb(\T\bwsmall) (\Psi) \mid \Psi \in \SRD(\Ga') \}
&& \text{$(\nb1)$}
\\&\leq \bv \{ \nb(\T\bwsmall)(\T\Pom f) (\Phi)  \mid \Phi \in \SRD(\Ga) \}
&& \text{(Lemma~\ref{l:bl1})}
\\&= \bv \{ \nb(\T f)(\T\bwsmall) (\Phi)  \mid \Phi \in \SRD(\Ga) \}
&& \text{(\dag)}
\\&= \bv \{ \hs(\T\bwsmall) (\Phi)  \mid \Phi \in \SRD(\Ga) \}
&& \text{(definition of $\hs$)}
\end{align*}
Here the identity marked (\dag) is easily justified by $f$ being a
homomorphism: it follows from $f \circ {\bw} = \bw \circ (\Pom f)$ and
functoriality of $\T$ that $(\T f) \circ (\T\bw) = (\T\bw)\circ(\T\Pom f)$.

Finally, for compatibility with $(\nb3)$ we need to verify that
\begin{equation}
\label{eq:comp3}
\text{for all } \Phi \in \T\P L\colon
\hs (\T\bvsmall) (\Phi) \leq
\bv \{ \hs\be \mid \be \rel{\Tb{\in}} \Phi \}.
\end{equation}
To prove this, we calculate for a given $\Phi\in\T\P L$:
\begin{align*}
\hs(\T\bvsmall) (\Phi)
&= \nb(\T f)(\T\bvsmall) (\Phi)
&& \text{(definition of $\hs$)}
\\&= \nb(\T \bvsmall)(\T\P f) (\Phi)
&& \text{($f$ a frame homomorphism)}
\\&\leq \bv \{ \nb\be \mid \be \rel{\Tb{\in}} (\T\P f) (\Phi) \}
&& \text{$(\nb3)_{\bbM}$}
\\ & = \bv \{ \nb(\T f) (\ga) \mid \ga \rel{\Tb{\in}}\Phi \}
&& \text{(\ddag)}
\\ & = \bv \{ \hs\ga \mid \ga \rel{\Tb{\in}}\Phi \}
&& \text{(definition of $\hs$)}
\end{align*}
Here the identity (\ddag) follows from the observation that for all $\be
\in \T M$ and $\Phi \in \T\P L$, we have $\be \rel{\Tb{\in}} (\T\P f) (\Phi)$
iff $\be$ is of the form $\be= (\T f) (\ga)$ for some $\ga \in \T L$.
Using Fact~\ref{f:rl}, this is easily derived from the observation that for
$b \in M$ and $A \in \P L$, we have $b \in (\P f)A$ iff $b = f(c)$ for some
$c \in A$.
\end{proof}

\noindent
Lemma~\ref{p:srd2} justifies the following definition.

\begin{definition}
Let $\T\colon \Set \to \Set$ be a standard, finitary, weak
pullback-preserving functor and let  $f\colon \mathbb{L}\to \mathbb{M}$ be a frame homomorphism.  We define
$\V_\T f \colon \V_\T \mathbb{L} \to \V_\T \mathbb{M}$ to be the unique
frame homomorphism extending \[\nabla \cof \T f \colon \T L \to \V_\T M.\]
\end{definition}

\begin{theorem}
Let $\T$ be a standard, finitary, weak pullback-preserving functor.
Then the operation $\V_\T$ defined above is an endofunctor on the
category $\class{Fr}$.
\end{theorem}

\begin{proof}
Since for an arbitrary $f\colon  \bbL \to \bbM$ we have ensured by definition
that $\VT F$ is a frame homomorphism from $\VT\bbL$ to $\VT\bbM$, it
is left to show that $\VT$ maps the identity map of a frame to the
identity map of its T-powerlocale, and distributes over function composition.
We confine our attention to the second property.

Let $f\colon  \mathbb{K} \to \bbL$ and $g\colon  \bbL \to \bbM$ be two frame
homomorphisms.
In order to show that $\VT(g\circ f) = \VT g \circ \VT f$, first recall that
$\VT (g \circ f)$ is by definition the unique frame homomorphism extending
the map $\nb_{\bbM}\circ \T(g\circ f)\colon \T K \to \VT\bbM$.
Hence, it suffices to prove that the map $\VT g \circ \VT f$, which is
obviously a frame homomorphism, extends $\nb_{\bbM}\circ \T(g\circ f)$.
But it is easy to see why this is the case: given an arbitrary element
$\al \in \T K$, a straightforward unraveling of definitions shows that
\[
(\VT g \circ \VT f) (\al) =
\VT g (\nb_{\bbL}(\T f) (\al)) =
\nb_{\bbM}(\T g)(\T f) (\al) =
\nb_{\bbM}T(g\circ f) (\al),
\]
as required.
\end{proof}

\subsubsection{Natural transformations between $\VT$ and $\VTp$}
Now that we have seen how each finitary, weak pullback preserving set functor
$\T$ induces a functor $\VT$ on the category of frames, we investigate the
relation \emph{between} two such functors $\VT,\VTp$.
In fact, we have already seen an example of this: recall that in \S \ref{Nabla:BoxDiamond} we
mentioned Johnstone's result~\citep{VietLoc} that the standard Vietoris
functor $\V$ is in fact a \emph{comonad} on the category of frames.
In our nabla-based presentation of this functor as $\V = \VP$, thinking
of the identity functor on the category $\class{Fr}$ as the Vietoris functor
$\V_{\Fid}$, we can see the \emph{counit} of this comonad as a natural
transformation
\[
\counit\colon  \VP \to V_{\Fid},
\]
given by $\counit_{\bbL}\colon  \nb A \mapsto \bw A$.
More precisely, we can show that the map $\hs\colon  \Pom L \to L$ given by
$\hs A := \bw A$ is compatible with the $\nb$-axioms, and hence can be
uniquely extended to the homomorphism $\counit_{\bbL}$; subsequently we
can show that this $\counit$ is natural in $\bbL$.
Recall that in the case of a concrete topological space $(X,\tau)$, this
counit corresponds on the dual side to the singleton map $\si_{X}\colon  s \mapsto
\{ s \}$ which provides an embedding of a compact Hausdorff topology into
its Vietoris space.

We will now see how to generalize this picture, of the natural transformation
$\counit\colon  \VP \to \V_{\Fid}$ being induced by the singleton natural
transformation $\si\colon  \Fid \to \Pom$, to a more general setting.
First consider the following definition.

\begin{definition}
\label{d:nt1}
Let $\T$ and $\T'$ be standard, finitary, weak pullback-preserving functors.
A natural transformation $\nta\colon  \T'\to\T$ is said to \emph{respect relation
lifting} if for any relation $R \sse X \times Y$ we have, for all
$\al'\in \T' X$ and $\be' \in \T' Y$:
\begin{equation}
\label{eq:ntrl}
\text{if } \al' \rel{\rl{T'} R} \be' \text{ then } \nta_{X}(\al') \rel{\rl{T}R}
\nta_{Y}(\be').
\end{equation}
We call $\nta$ \emph{base-invariant} if it commutes with $\Base$, that is,
\begin{equation}
\Base^{\T'} = \Base^{\T}\cof\nta.
\end{equation}
for any set $X$.
\end{definition}
\begin{example} We record three examples of base-invariant natural
transformations which respect relation lifting.
\begin{enumerate}
\item
The base transformation
$\Base^{\T}\colon  \T \to \Pom$;
\item The  singleton natural transformation
$\si \colon \operatorname{Id} \to \Pom$, which is in fact a special
case of (1);
\item
The diagonal map $\de$ (given by $\de_{X}\colon
x \mapsto (x,x)$); it is straightforward to check that as a natural
transformation, $\de\colon  \Fid \to \Fid \times \Fid$ also satisfies
both properties of Definition~\ref{d:nt1}.
\end{enumerate}
\end{example}

As we will see now, every base-invariant natural transformation $\nta\colon  \T'
\to \T$ that respects relation lifting, induces a natural transformation
$\wh{\nta}\colon  \VT \to \VTp$.
In particular, the natural transformation $\counit\colon  \V\to\Fid$ can be seen as
$\counit = \wh{\si}$, where $\si\colon  \Fid \to \Pom$ is the singleton transformation
discussed above.

\begin{theorem}
\label{th:nat-trans}
Let $\T$ and $\T'$ be standard, finitary, weak pullback-preserving functors,
assume that $\nta\colon  \T'\to\T$ is a base-invariant natural transformation that
respects relation lifting, and let $\mathbb{L}$ be a frame.
Then the map from $\T L$ to $\VTp L$ given by
\[
\al \mapsto \bv \{ \nb \al' \mid \al' \in \T' L, \nta(\al') \Tleq \al \}
\]
specifies a frame homomorphism
\[
\wh{\nta}_{\bbL} \colon \VT\bbL \to \VTp\bbL
\]
which is natural in $\bbL$.
\end{theorem}

\begin{proof}
We let $\hs\colon  \T L \to L$ denote the map given in the statement of the Theorem,
that is, $\hs \al := \bv \{ \nb \al' \mid \al' \in \T' L, \nta(\al') \Tleq \al
\}$.
We will first prove that this map is compatible with, respectively, ($\nb1$),
($\nb2$) and ($\nb3$), and then turn to the naturality of the induced frame
homomorphism.

\begin{claimfirst}
\label{cl:nt:1}
The map $\hs$ is compatible with ($\nb1$).
\end{claimfirst}

\begin{pfclaim}
To show that $\hs$ is compatible with ($\nb1$), take two elements $\al,\be \in
\T L$ such that $\al \Tleq \be$.
Then for any $\al'\in \T' L$ such that $\nta(\al') \Tleq \al$, by transitivity
of $\Tleq$ (Fact~\ref{f:rl}(5)), we obtain that $\nta(\al')  \Tleq \be$.
From this it is immediate that $\hs \al \leq \hs \be$, as required.
\end{pfclaim}

\begin{claim}
\label{cl:nt:2}
The map $\hs$ is compatible with ($\nb2$).
\end{claim}

\begin{pfclaim}
For compatibility with ($\nb2$), it suffices to show compatibility with
($\nb2'$).
That is, for $\Ga \in \Pom\T L$, we will verify that
\begin{equation}
\label{eq:nt2:a}
\bw \{ \hs \ga \mid \ga \in \Ga \} \leq
\bv \{ \hs \be \mid \be \Tleq \ga, \text{ for all } \ga \in \Ga \}.
\end{equation}
We start with rewriting the left hand side of \eqref{eq:nt2:a} into
\begin{align*}
\bw \{ \hs \ga \mid \ga \in \Ga \}
&= \bw \Big\{
\bvsmall \left\{ \nb \ga' \mid \nta(\ga') \Tleq \ga \right\}
\mid \ga\in\Ga \Big\}
&& \text{(definition of $\hs$)}
\\ &= \bv \Big\{
\bwsmall \left\{ \phi_{\ga} \mid \ga \in \Ga \right\}
\mid \phi \in \cC_{\Ga} \Big\}
&& \text{(frame distributivity)}
\end{align*}
where we define $\cC_{\Ga} := \{ \phi\colon  \Ga \to \T' L \mid
\nta(\phi_{\ga}) \Tleq \ga, \text{ for all } \ga \in \Ga \}$.

For any map $\phi \in \cC_{\Ga}$ we may calculate
\begin{align*}
&\bw \left\{ \phi_{\ga} \mid \ga \in \Ga \right\}\\
&= \bv \{ \nb\ga' \mid \ga' \Tpleq \phi_{\ga}, \forall \ga \in \Ga \}
&& \text{($\nb2'$)}
\\&\leq \bv \{ \nb\ga' \mid \nta(\ga') \Tleq \nta(\phi_{\ga}),
\forall \ga \in \Ga\}
&& \text{($\nta$ respects relation lifting)}
\\&\leq \bv \{ \nb\ga' \mid \nta(\ga') \Tleq \ga,  \forall \ga \in \Ga \}
&& \text{($\phi\in\cC_{\Ga}$, transitivity of $\Tleq$)}
\\&= \bv \Big\{
\bvsmall \{ \nb\ga' \mid \nta(\ga') \Tleq \be \}
\mid \be \Tleq \ga,  \forall\ga \in \Ga \Big\}
&& \text{(associativity of $\bvsmall$)}
\\&= \bv \{ \hs \be \mid \be \Tleq \ga,  \forall \ga \in \Ga \}
&& \text{(definition of $\hs$)}
\end{align*}

From the above calculations, \eqref{eq:nt2:a} is immediate.
\end{pfclaim}

\begin{claim}
\label{cl:nt:3}
The map $\hs$ is compatible with ($\nb3$).
\end{claim}

\begin{pfclaim}
We need to show, for an arbitrary but fixed set $\Phi \in \T\P L$, that
\begin{equation}
\label{eq:nt2:1}
\hs(\T \bvsmall) (\Phi) = \bv \{ \hs \al \mid \al \Tin \Phi \}.
\end{equation}
By definition, on the left hand side of \eqref{eq:nt2:1} we find
\[
\hs(\T \bvsmall) (\Phi) =
\bv \{ \nb\be' \mid \nta(\be') \Tleq (\T\bvsmall) (\Phi) \},
\]
while on the right hand side we obtain, by definition of $\hs$,
\begin{align*}
\bv \{ \hs \al \mid \al \Tin \Phi \}
&= \bv \Big\{ \bvsmall \{ \nb \al' \mid \nta(\al') \Tleq \al \}
\mid \al \Tin \Phi
\Big\}
\\ &= \bv \Big\{ \nb \al' \mid \nta(\al') \rel{\Tb{({\leq}\cor{\in})}} \Phi
\Big\}
\end{align*}
where the latter equality is by associativity of $\bv$, and the
compositionality of relation lifting (Fact~\ref{f:rl}(5)).

As a consequence, in order to establish the compatibility of $\hs$ with
($\nb3$), it suffices to show that
\begin{equation}
\label{eq:nt2:2}
\nb\be' \leq
\bv \Big\{ \nb \al' \mid \nta(\al') \rel{\Tb{({\leq}\cor{\in})}} \Phi \Big\},
\text{ for any } \be' \text{ with }
\nta(\be') \Tleq (\T\bvsmall) (\Phi).
\end{equation}
Let $\be'$ be an arbitrary element of $\T L$ such that $\nta(\be') \Tleq
(\T\bvsmall) (\Phi)$.
Our goal will be to find a set $\Phi' \in \T'\P L$ satisfying \eqref{eq:nt2:6},
\eqref{eq:nt2:7} and \eqref{eq:nt2:8} below: clearly this will satisfy to
prove \eqref{eq:nt2:2}.

By Fact~\ref{f:base} we obtain that
\[
\Base^{\T}(\nta\be') \Pleq
\Base^{\T}((\T\bvsmall)(\Phi)) = (\P\bvsmall)\Base^{\T}(\Phi),
\]
and since $\nta$
is base-invariant, we have $\Base^{T'}(\be') = \Base^{\T}(\nta\be')$.
Combining these facts we see that $\Base^{T'}(\be') \Pleq (\P\bv)\Base^{\T}
(\Phi)$.
This motivates the definition of the following map $\cH\colon  \Base^{\T'}(\be') \to
\Pom\P L$:
\[
\cH(b) := \{ B \in \Base^{\T}(\Phi) \mid b \leq \bvsmall B \}.
\]
From the definitions it is immediate that
\begin{equation}
\label{eq:nt2:3}
\text{for all } b\in\Base^{\T'}(\be'):\; b \leq
\bw \big\{ \bvsmall B \mid B \in \cH(b) \big\}.
\end{equation}
Also, given a set $\cB \in \Pom\P L$, let $\CF_{B}$ be the collection of choice
functions on $\cB$, that is:
\[
\CF_{\cB} := \{ f\colon  \cB \to L \mid f(B) \in B \text{ for all } B \in \cB \}.
\]
Then it follows by frame distributivity that
\begin{equation}
\label{eq:nt2:4}
\bw \big\{ \bvsmall B \mid B \in \cB \big\} =
\bv \big\{ \bwsmall (\P f)(\cB) \mid f \in \CF_{\cB} \big\}.
\end{equation}
Define the map $K\colon  \Pom\P L \to \P L$ by
\[
K(\cB) := \big\{\bwsmall (\P f)(\cB) \mid f \in \CF_{\cB} \big\},
\]
then it follows from \eqref{eq:nt2:3}, \eqref{eq:nt2:4} and the definitions
that
\begin{equation}
\label{eq:nt2:5}
\text{for all } b\in\Base^{\T'}(\be'):\; b \leq \bv K(\cH(b)).
\end{equation}
As a corollary, if we define
\[
\Phi' := (\T' K)(\T' \cH)(\be'),
\]
then it follows from \eqref{eq:nt2:5}, by the properties of relation lifting,
that $\be' \Tpleq (\T'\bv)(\Phi')$, so that an application of ($\nb1$)
yields
\begin{equation}
\label{eq:nt2:6}
\nb\be' \leq \nb(\T'\bvsmall)(\Phi').
\end{equation}
Also, on the basis of an application of ($\nb3$) we may conclude that
\begin{equation}
\label{eq:nt2:7}
\nb(\T'\bvsmall)(\Phi') \leq \bv \{ \nb\ga' \mid \ga' \Tpin \Phi' \}.
\end{equation}
This means that we are done with the proof of \eqref{eq:nt2:2} if we can
show that
\begin{equation}
\label{eq:nt2:8}
\text{ for any } \ga' \in \T' L, \text{ if } \ga' \Tpin \Phi'
\text{ then } \nta(\ga') \rel{\Tb{({\leq}\cor{\in})}} \Phi.
\end{equation}

For a proof of \eqref{eq:nt2:8}, let $\ga'$ be an arbitrary $\T'$-lifted
member of $\Phi'$ and recall that $\Phi' = (\T K)(\T \cH)(\be')$.
Then it follows by the assumption that $\nta$ respects relation lifting,
that
$\nta(\ga') \Tin \nta(\Phi') = (\T K)(\T \cH)(\nta(\be'))$.
Given our assumption on $\be'$, this means that the relation between
$\nta(\ga')$ and $\Phi$ can be summarized as
\begin{equation}
\label{eq:nt:9}
\nta(\ga') \Tin (\T K)(\T \cH)(\be) \text{ and } \be \Tleq (\T\bvsmall) (\Phi)
\text{ for some } \be \in \T\Base^{\T'}(\be'),
\end{equation}
where for $\be$ we may take $\nta(\be')$.

Returning to the ground level, observe that for any $c \in L$, $A \in
\Base^{\T}(\Phi)$,
we have
\begin{equation}
\label{eq:nt:10}
\text{if } c \in K\cH (b) \text{ and } b \leq \bvsmall A,
\text{ for some } b \in \Base^{\T'}(\be'),
\text{ then } c \rel{({\leq}\cor{\in})} A.
\end{equation}
To see why this is the case, assume that $c \in K\cH(b)$ and $b \leq
\bvsmall A$, for some $b \in \Base^{\T'}(\be')$.
Then by definition of $\cH$ we find $A \in \cH(b)$, while $c \in K\cH(b)$
simply means that $c = \bw \{ f(B) \mid B \in \cH(b)\}$, for some $f \in
\CF_{\cH(b)}$.
But then it is immediate that $c \leq f(A)$, while $f(A) \in A$ by
definition of $\CF_{\cH(b)}$.
Thus $f(A)$ is the required element witnessing that $c \rel{({\leq}\cor{\in})}
A$.

But by the properties of relation lifting, we may derive from \eqref{eq:nt:10}
that
\begin{align}
&\text{if } \ga \Tin (\T K)(\T \cH)(\be) \text{ and } \be \Tleq
(\T\bvsmall) (\Phi)
\text{ for some } \be \in \T\Base^{\T'}(\be'), \notag\\
&\text{ then } \ga \rel{\Tb({\leq}\cor{\in})} \Phi, \label{eq:nt:11}
\end{align}
so that it is immediate by \eqref{eq:nt:9} that
$\nta(\ga') \rel{\Tb{({\leq}\cor{\in})}} \Phi$.
This proves \eqref{eq:nt2:8}.

As mentioned already, the compatibility of $\hs$ with ($\nb3$) is immediate
by \eqref{eq:nt2:6}, \eqref{eq:nt2:7} and \eqref{eq:nt2:8}, and so this
finishes the proof of Claim~\ref{cl:nt:3}.
\end{pfclaim}

As a corollary of the Claims~\ref{cl:nt:1}--\ref{cl:nt:3}, we may uniquely
extend $\hs$ to a homomorphism $\wh{\nta}_{\bbL}\colon  \VT\bbL \to \VTp\bbL$.
Clearly then, in order to prove the theorem it suffices to prove the
following claim.

\begin{claim}
\label{cl:nt4}
The family of homomorphisms $\wh{\nta}_{\bbL}$ constitutes a natural
transformation $\wh{\nta}\colon  \VT \to \VTp$.
\end{claim}

\begin{pfclaim}
Given two frames $\bbL$ and $\bbM$ and a frame homomorphism $f\colon  \bbL \to \bbM$,
we need to show that the following diagram commutes:
\\ \centerline{\xymatrix{%
\VT\bbL \ar[r]^{\wh{\nta}_{\bbL}} \ar[d]_{\VT f}
& \VTp\bbL \ar[d]^{\VTp f}
\\
\VT\bbM \ar[r]^{\wh{\nta}_{\bbM}} & \VTp\bbM
}}
\\[2mm]
To show this, take an arbitrary element $\al \in \T L$, and consider the
following calculation:

\begin{align*}
&(\VTp f)(\wh{\nta}_{\bbL}(\nb\al))\\
&= (\VTp f)(\hs \al)
&& \text{(definition of $\wh{\nta}_{\bbL}$))}
\\&= (\VTp f)\left(\bv \big\{ \nb\be' \mid
\nta_{L}(\be') \Tleq \al \big\} \right)
&& \text{(definition of $\hs$))}
\\&= \bv \big\{ (\VT' f)(\nb\be') \mid \nta_{L}(\be') \Tleq \al \big\}
&& \text{($\VTp f$ is a frame homomorphism)}
\\&= \bv \big\{ \nb(\T' f)(\be') \mid \nta_{L}(\be') \Tleq \al \big\}
&& \text{(definition of $\VTp f$))}
\\&= \bv \big\{ \nb\de' \mid \nta_{M}(\de') \Tleq (\T f) (\al)\big\}
&& \text{(\dag)}
\\&= \hs(\T f)(\al)
&& \text{(definition of $\hs$))}
\\&=\wh{\nta}_{\bbM}(\nb(\T f)(\al))
&& \text{(definition of $\wh{\nta}_{\bbM}$))}
\\&= \wh{\nta}_{\bbM}((\VT f)(\nb \al))
&& \text{(definition of $\VT f$))}
\end{align*}

Here the crucial step, marked (\dag), is proved by establishing the two
respective inequalities, as follows.
For the inequality $\leq$, it is straightforward to show that the set of
joinands on the left hand side is included in that on the right hand side,
and this follows from
\begin{equation}
\label{eq:nt:12}
\nta_{L}(\be') \Tleq \al \text{ implies }
\nta_{M}((\T'f)(\be')) \Tleq (\T f)(\al).
\end{equation}
To prove \eqref{eq:nt:12}, suppose that $\nta_{L}(\be') \Tleq \al$; then it
follows by the fact that $f$ is a homomorphism, and hence, monotone, that
$(\T f)(\nta_{L}(\be')) \Tleq (\T f)(\al)$.
But since $\nta$ is a natural transformation, we also have $(\T
f)(\nta_{L}(\be')) = \nta_{M}(\T' f)(\be')$, and from this \eqref{eq:nt:12}
is immediate.

In order to prove the opposite inequality
\begin{equation}
\label{eq:nt2:10}
\bv \big\{ \nb\de' \mid \nta_{M}(\de') \Tleq (\T f) (\al)\big\}
\leq \bv \big\{ \nb(\T' f)(\be') \mid \nta_{L}(\be') \Tleq \al \big\},
\end{equation}
fix an arbitrary element $\de' \in \T L$ such that $\nta_{M}(\de') \Tleq
(\T f) (\al)$.

Define the map $h\colon  \Base^{\T'}(\de') \to L$ by putting
\[
h(d) := \bw\{ a \in \Base^{\T}(\al) \mid d \leq f(a) \}.
\]
Then for all $d \in \Base^{\T'}(\de')$ and all $a \in \Base^{\T}(\al)$, we
find that $d \leq fa$ implies $hd \leq a$; this can be expressed by the
relational inclusion
\[ \Graph{f}\cor{\geq}\cor\Graph{h} \;\sse\; {\geq}
\]
so that by the properties of relation lifting we may conclude that
$\Graph{(\T f)} \cor {\Tb{\geq}} \cor \Graph{(\T h)} \;\sse\; {\Tb{\geq}}$,
which is just another way of saying that, for all $\de \in \T
\Base^{\T'}(\de')$, we have
\begin{equation}
\label{eq:nt:16}
\de \Tleq (\T f)(\al) \text{ only if } (\T h)(\de) \Tleq \al.
\end{equation}
Now define
\[
\be' := (\T' h)(\de'),
\]
then we may conclude from the fact that $\nta$ respects relation lifting that
$\nta_{L}(\be') = (\T h)\nta_{M}(\de')$, and so by the assumption that
$\nta_{M}(\de') \Tleq (\T f) (\al)$, we obtain by \eqref{eq:nt:16} that
\begin{equation}
\label{eq:nt2:11}
\nta_{L}(\be') \Tleq \al.
\end{equation}

Similarly, from the fact that $d \leq fhd$, for each $d \in \Base^{\T'}(\de')$,
we may derive that $\de' \Tpleq (\T' f)(\be')$, and so by ($\nb1$) we may
conclude that
\begin{equation}
\label{eq:nt:15}
\nb\de' \leq \nb(\T' f)(\be').
\end{equation}
Finally, \eqref{eq:nt:12} is immediate by \eqref{eq:nt:11} and
\eqref{eq:nt:15}.

This finishes the proof of Claim~\ref{cl:nt4}.
\end{pfclaim}
\end{proof}

\begin{remark}\label{Remark:RhoAdjoint}
The definition of the $\widehat \rho_\bbL \colon \VT \bbL \to \VTp
\bbL$, using the assignment
\[
\al \mapsto \bv \{ \nb \al' \mid \al' \in \T' L, \nta(\al') \Tleq \al \},
\]
is very similar to that of a right adjoint. If it were the case that
$\widehat \rho_\bbL$ preserved \emph{all} meets, then the adjoint
functor theorem would allow us to define its left adjoint.  However,
we only have a proof that $\widehat \rho_\bbL \colon \VT \bbL \to \VTp
\bbL$ preserves \emph{finite} conjunctions, so it is not at all
obvious at this point if there even is a left adjoint to $\widehat
\rho_\bbL$. This is an interesting question for future work.
\end{remark}

\subsection{$\T$-powerlocales via flat sites}
\label{Subsect:VTFlatsite}

In this subsection, we will show that $\VT \bbL$, the $\T$-powerlocale
of a given frame $\bbL$, has a flat site presentation as $\VT \bbL
\simeq \Fr \struc{\T L, \Tleq, \cov_0^\bbL}$. It then follows by the
Flat site Coverage Theorem that every element of $\VT \bbL$ has a
disjunctive normal form, and that the suplattice reduct of $\VT \bbL$
has a presentation defined only in terms of the order $\Tleq$ and the
lifted join function $\T \smallvee \colon \T\P L \to \T L$.

Recall that $\struc{X,\lesssim, \cov_0}$ is a flat site if $\struc{ X,
\lesssim}$ is a pre-order and $\cov_0$ is a basic cover relation
compatible with $\lesssim$. In that case, we know that
$\struc{X,\lesssim, \cov_0}$ presents a frame $\Fr \struc{X,\lesssim,
\cov_0}$, and that if we denote the insertion of generators by $\hs
\colon X \to \Fr \struc{X,\lesssim, \cov_0}$, then
\[\begin{array}{rl}
\Fr \struc{ X, \lesssim,
\cov_0 } \simeq \Fr \struc{ X \mid  &\hs a \leq
\hs b \quad (a \lesssim b),\\
& 1 = \smallvee \{ \hs a \mid a \in X\} \\
&
\hs a\wedge \hs b = \smallvee \{ \hs c \mid c\lesssim
a, c\lesssim b \}\\
& \hs a \leq \smallvee \{ \hs b \mid b \in A\} \quad (a\cov_0 A)
}.
\end{array} \]
Observe that this is very similar to our presentation of $\VT \bbL$
from Corollary \ref{VT:Nabla2prime} using $(\nabla 1)$, $(\nabla 2')$
and $(\nabla 3)$, namely
\[\begin{array}{rl}
\VT \bbL \simeq \Fr \struc{ \T L \mid
&\nabla \alpha \leq
\nabla \beta\quad ( \alpha \Tleq \beta),\\
&
\smallwedge_\Gamma \nabla \gamma = \smallvee \{ \nabla \delta\mid
\forall \gamma \in \Gamma,\, \delta\Tleq \gamma\} \quad (\Gamma \in
\T\Pom L)\\
&\nabla \T \smallvee (\Phi) \leq \smallvee \{  \nabla \beta \mid \beta
\in \nbsem (\Phi)\} \quad (\Phi \in \T \P L)
}.
\end{array} \]
We will see below that if we define a cover relation $\cov_0^\bbL$
which is inspired by $(\nabla 3)$, then we obtain a flat site $\struc{
\T L , \Tleq, \cov_0^\bbL}$, and this flat site presents $\VT \bbL$.

So how do we go about defining a basic cover relation $\cov_0^\bbL
\subseteq \T L \times \P \T L$ so we can give a presentation of $\VT
\bbL$? Intuitively, we would like to take the $\T$-lifting of the relation
$ \{(a,A) \in L \times \P L \mid a\leq \smallvee A\} = {\leq} \cor
\cv{\left( \Graph{ {\smallvee} } \right)} $. However, the $\T$-lifting
of this relation is of type $\T L \times \T\P L$, while a basic cover
relation on $\struc{ \T L, \Tb {\leq} }$  should be of type $\T L
\times \P \T L$.
We solve this by involving the natural transformation $\nbsem\colon
\T\P \to \P\T$, given by $\nbsem (\Phi) := \{ \beta \in \T L \mid
\beta \rel{\Tb{\in}} \Phi\}$, assigning to each $\Phi \in \T\P L$ the
set of its lifted members.
That is, we define
\[
\cov_0^\bbL \;:=\;
\{ (\al,\nbsem(\Phi)) \in L \times \P\T L
\mid \alpha \rel{ \Tb {\leq}} \T \smallvee (\Phi )\}.
\]
In other words: we put $\al\cov_{0}^{\bbL}\Ga$ iff $\Ga$ is of the form
$\nbsem(\Phi)$ for some $\Phi \in \T\P L$ such that $\al
\rel{ \Tb {\leq}} (\T \bvsmall) \Phi$.
Two tasks lie ahead of us:
first, we must show that $\struc{
\T L , \Tleq, \cov_0^\bbL}$ is a flat site, meaning that $\cov_0$ is
compatible with $\Tleq$.
Second, we must show that $ \struc{ \T L , \Tleq, \cov_0^\bbL}$ presents
$\VT \bbL$.
The following technical observation
about the relation $\alpha \Tleq \T \smallvee (\Phi)$ is the main
reason why $\VT \bbL$ admits a flat site presentation.
The reason for introducing a $\wedge$-semilattice $\bbM$ below will become
apparent in \S \ref{Subsect:Compactness}.

\begin{lemma}\label{Flatsite:Stability}
Let $\T\colon \Set \to \Set$ be a standard, finitary, weak pullback-preserving functor,
let $\bbL$ be a frame and let $\bbM$ be a $\wedge$-subsemilattice of
$\bbL$. Then for all $\alpha \in \T M$ and $\Phi \in \T \P M$ such
that $\alpha \rel{\Tb {\leq}} \T \smallvee (\Phi)$, there exists
$\Phi' \in \T\P M$ such that \begin{enumerate}
\item $\alpha \rel{\Tb {\leq}} \T \smallvee (\Phi')$;
\item $\Phi' \rel{\Tb {\subseteq}} \T \downarr_L \cof \T \single
(\alpha)$;
\item $\Phi' \rel{\Tb {\subseteq}} \T \downarr_L (\Phi)$
\end{enumerate}
\end{lemma}
\begin{proof}

First, we define the following relation on $M \times \P M$:
\[
R:= \{ (a,A) \in M \times \P M \mid
a \leq \smallvee A\} = \left ({\leq} \cor \cv{(\Graph\smallvee)}
\right) \rst{M \times \P M}.
\]
Consider the span $M \xleftarrow{p_1} R \xrightarrow{p_2} \P M$.
We define the following function $f\colon R \to R$:
\[f\colon (a,A) \mapsto (a,a \wedge A),
\]
where $a\wedge A := \{ a\wedge b \mid b \in A\}$.
To see why this function is well-defined, first observe that $a \wedge A \in
\P M$ because $\bbM$ is a $\wedge$-subsemilattice of $\bbL$.
Moreover, by frame distributivity, we see that if $(a,A)
\in R$, i.e.~if $a\leq
\smallvee A$, then also $a \leq \smallvee ( a\wedge A)$, so that $(a,
a\wedge A) \in R$.
Now observe that $f\colon R \to R$  satisfies an equation and two inequations: for all $(a,A)
\in R$,
\begin{align*}
&p_1 \cof f
(a,A) = a = p_1(a,A), &&\text{by def.~of $f$,}
\\ &p_2 \cof f (a,A) = a\wedge A \subseteq_L
\downarr_L \{ a\} = \downarr_L \cof \single_L
\cof p_1 (a,A),&&\text{since $\forall b \in A$, $a \wedge b \leq a$,}\\
&p_2 \cof f (a,A) = a \wedge A \subseteq_L \downarr_L A =
\downarr_L \cof p_2 (a,A)&&\text{since $\forall b \in A$, $a
\wedge b \leq b \in A$.}
\end{align*}
Now consider the lifted diagram
\[ \xymatrix{ \T M & \T R \ar[l]_{\T p_1} \ar[r]^{\T p2} & \T
\P M. }\]
It follows from Lemma~\ref{p:rl2} and the
equation/inequations above that for each $\delta \in \T R$, we have
\begin{align}
\T p_1 \cof \T f (\delta)&= \T p_1
(\delta), \label{Flatsite:Well-defined:Eq1}\\
\T p_2 \cof \T f (\delta) & \rel{\Tb {\subseteq_L} } \T \downarr_L \cof
\T \single_L \cof \T p_1 (\delta) ,\label{Flatsite:Well-defined:Eq2}\\
\T p_2 \cof \T f (\delta) & \rel{\Tb {\subseteq_L} } \T \downarr_L
\cof \T p_2 (\delta) \label{Flatsite:Well-defined:Eq3}
\end{align}
Now recall that by Fact \ref{f:rl}, \[ \Tb {\leq} \cor
\cv{\Graph(\T \smallvee)} = \Tb ( {\leq} \cor
\cv{(\Graph\smallvee)}) = \Tb R,\]
so we see that $\alpha \rel{\Tb \leq} \T \smallvee (\Phi)$ iff
$\alpha \rel{\Tb R} \Phi$. So let $\alpha \in \T M$ and $\Phi \in \T \P M$ such that $\alpha \Tleq
\T \smallvee (\Phi)$, i.e.~such that   $\alpha \rel{\Tb R}
\Phi$; we will show that there is a $\Phi' \in \T \P M$ satisfying
properties (1)--(3).
First, observe that by definition of relation lifting, there must exist
some $\delta \in \T R$ such that
\[
\T p_1 (\delta) = \alpha \text{ and }\T p_2 (\delta) = \Phi.
\]
We claim that $\Phi' := \T p_2 \cof \T f
(\delta)$ satisfies properties (1)--(3).
We know by definition of relation lifting that $\left( \T p_1 \cof\T
f(\delta) \right) \rel{\Tb
R} \left( \T p_2 \cof\T f (\delta) \right)$. Since \begin{align*}\T p_1 \cof\T
f(\delta) &= \T p_1 (\delta) &&\text{by
\eqref{Flatsite:Well-defined:Eq1},}\\
&= \alpha &&\text{by assumption,} \end{align*} it follows that $\alpha
\rel{ \Tb R} \Phi'$, i.e.~$\alpha \rel{\Tb {\leq}} \T \smallvee
(\Phi')$; we conclude that (1) holds. Moreover, it follows immediately from \eqref{Flatsite:Well-defined:Eq2}
that (2) holds.
Similarly, it follows immediately from \eqref{Flatsite:Well-defined:Eq3} that
(3) holds.
\end{proof}

In the lemma above, we use the lifted inclusion relation $\T
{\subseteq}$ and the lifted downset function $\T \downarr$. In the
lemma below we record some elementary observations about the
interaction between $\T
{\subseteq}$, $\T \downarr$ and the natural transformation $\nbsem
\colon \T \P \to \P \T$.

\begin{lemma} \label{RelationLifting:LiftedDownset} \label{RelationLifting:LiftedInclusion}
Let $\T\colon \Set \to \Set$ be a standard, finitary, weak pullback-preserving functor, let $\struc{ X, \lesssim}$ be a pre-order, let $\alpha\in \T
X$  and let $\Phi, \Phi' \in \T\P
X$. Then
\begin{enumerate}
\item ${\downarrow}_{\T X} \nbsem (\Phi) = \nbsem
\left( \T {\downarrow}_X  (\Phi) \right)$;
\item ${\downarrow}_{\T X} \{ \alpha\} = \nbsem \left( \T
{\downarrow}_X \cof \T \single_X (\alpha) \right)$;
\item If $\Phi'
\rel{\Tb \subseteq_X } \Phi$, then also $\nbsem (\Phi')
\subseteq \nbsem (\Phi)$.
\end{enumerate}
\end{lemma}
\begin{proof}
(1). For all $a \in X$ and all $A \in \P X$, we have $a \rel{{\leq}
\cor {\in}} A $ iff $a \in \downarr_X A$. Consequently,
\[ \forall \alpha \in \T L, \forall \Phi \in \T \P L,\, \alpha \Tleq
\cor \Tin \Phi \text{ iff } \alpha \Tin \T \downarr_X (\Phi).\]
Now we see that
\begin{align*}
\alpha \in \downarr_{\T X} \nbsem (\Phi) &\Leftrightarrow \alpha \Tleq
\cor \Tin \Phi &&\text{by def.~of $\downarr$ and $\nbsem$,}\\
&\Leftrightarrow \alpha \Tin \T \downarr_X (\Phi)&&\text{by the
above,}\\
&\Leftrightarrow \alpha \in \nbsem
\left( \T {\downarrow}_X  (\Phi) \right)&&\text{by def.~of $\nbsem$.}
\end{align*}

(2). For all $a,b \in X$, we have $b \leq a$ iff $b \in \downarr_X \{
a\}$. It follows by relation lifting that
\[ \forall \alpha,\beta \in \T X,\, \beta \Tleq \alpha \text{ iff }
\beta \Tin \T \downarr_X \cof \T \eta_X (\alpha). \]
It now follows by an argument analogous to that for (1) above that (2)
holds.

(3). Observe that for all $A,A' \in \P X$ and all $a \in X$, we have
that $a \in A' \subseteq A$ implies that $a\in A$. The statement
follows by relation lifting.
\end{proof}

We are now ready to prove that $\struc{ \T L, \Tb {\leq},
{\cov_0^\bbL} }$ is indeed a flat site.

\begin{lemma}\label{Flatsite:Well-defined}
Let $\T\colon \Set \to \Set$ be a standard, finitary, weak pullback-preserving functor.
If $\bbL$ is a frame then $\struc{ \T L, \Tb {\leq},
{\cov_0^\bbL} }$ is a flat site.
\end{lemma}
\begin{proof}
We have already know from Lemma \ref{p:rl2} that $\struc{ \T L, \Tb {\leq}}$ is a
pre-order, so what remains to be shown is that the relation
${\cov_0^\bbL}$ is compatible with the pre-order. Fix $\alpha \in \T
L$ and $\Phi \in \T \P L$ such that $\alpha \Tleq \T \smallvee
(\Phi)$, so that $\alpha \cov_0^\bbL \nbsem(\Phi)$.
We need to show that
\begin{equation}\label{Flatsite:Well-defined:C1}
\forall \beta \in \T L, \text{ if } \beta \Tleq \alpha \text{ then } \exists
\Gamma \in \T \P L \text{ with } \Gamma \subseteq \downarr_{\T L} \{ \beta\} \cap \downarr_{\T L}
\nbsem(\Phi) \text{ and } \beta \cov_0^\bbL \Gamma.
\end{equation}
But this is easy to see: if $\beta \Tleq \alpha $  then since $\alpha
\Tleq \T \smallvee (\Phi)$, it follows by transitivity of $\Tleq$ that
$\beta \Tleq \T \smallvee(\Phi)$.
Now by Lemma
\ref{Flatsite:Stability} there exists $\Phi' \in \T\P L$ such that
$\alpha \rel{\Tb {\leq}} \T \smallvee (\Phi')$, $\Phi' \rel{\Tb {\subseteq}} \T \downarr_L \cof \T \single
(\beta)$ and $\Phi' \rel{\Tb {\subseteq}} \T \downarr_L
\Phi$.
Define $\Gamma := \nbsem(\Phi')$, then we have  $\beta \cov_0 \Gamma$ by
definition of $\cov_0^\bbL$ that; moreover, it now follows from Lemma
\ref{RelationLifting:LiftedDownset}  that $\Gamma \subseteq \downarr_{\T L} \{ \beta\} \cap \downarr_{\T L}
\nbsem(\Phi)$. We conclude that \eqref{Flatsite:Well-defined:C1}
holds.
Since $\alpha \in \T L$ and $\Phi \in \T \P L$ were arbitrary, we have shown
that $\cov_0^\bbL$ is compatible with the order $\Tleq$, so that
$\struc{ \T L, \Tb {\leq}, {\cov_0^\bbL} }$ is a flat site.
\end{proof}

Having established that $\struc{
\T L , \Tleq, \cov_0^\bbL}$ is a flat site, we will now prove that
it presents $\VT \bbL$, i.e.~that $\VT \bbL \simeq \Fr \struc{
\T L , \Tleq, \cov_0^\bbL}$.

\begin{theorem}\label{Flatsite:VTFlatsite}
Let $\bbL$ be a frame and let $\T$ be a standard, finitary, weak
pullback-preserving functor.
Then $\V_\T \bbL$ admits the following flat site presentation:
\[
\V_\T \bbL \simeq \Fr \struc{   \T L, \Tb {\leq}, \cov_0^\bbL},
\]
where $\cov_0^\bbL =
\{ (\al,\nbsem(\Phi)) \in L \times \P\T L
\mid \alpha \rel{ \Tb {\leq}} \T \smallvee (\Phi) \}$,
and in each direction, the isomorphism is the unique frame homomorphism
extending the identity map $\id_{\T L}$ on the set of generators of
$\V_\T \bbL$ and $\Fr \struc{   \T L, \Tb {\leq}, \cov_0^\bbL}$,
respectively.
\end{theorem}
\begin{proof}
For this proof,  we denote the insertion of generators from $\T L$ to
$\VT\bbL$ by $\nb$, and from $\T L$ to $\Fr\struc{   \T L, \Tb {\leq},
\cov_0^\bbL }$ by $\heartsuit$.
We will show that
\begin{enumerate}
\item
the function $\heartsuit \colon \T L \to \Fr \struc{\T L,
\Tb {\leq}, \cov_0^\bbL}$ is compatible with the relations $(\nabla 1)$,
$(\nabla 2')$ and $(\nabla 3)$, and
\item
that the function $\nabla \colon \T L \to \V_\T \bbL $ has the
following properties:
\begin{enumerate}
\item $\nabla$ is
order-preserving;
\item $1= \smallvee \{ \nabla \alpha \mid \alpha \in \T L\}$;
\item for all $\alpha,\beta \in \T L$, $\nabla \alpha \wedge \nabla \beta = \smallwedge \{ \nabla
\gamma \mid \delta \Tleq \alpha,\beta\}$;
\item for all $\alpha \cov_0^\bbL \Gamma$, $\nabla \alpha \leq
\smallvee \{ \nabla \beta \mid \beta \in \Gamma\}$.
\end{enumerate}
\end{enumerate}

(1). First consider $(\nabla 1)$. Suppose that $\alpha,\beta \in \T L$
such that $\alpha \rel{\Tb {\leq}} \beta$; we have to show that
$\heartsuit \alpha \leq \heartsuit \beta$. This follows immediately
from the fact that $\heartsuit\colon \T L \to \Fr
\struc{   \T L, \Tb {\leq}, \cov_0^\bbL
}$ is order-preserving. Secondly, consider $(\nabla 2')$. Let $\Gamma
\in \Pom \T L$, we then have to show that
\begin{equation}\label{Flatsite:VTFlatsite:C1}
\smallwedge_{\gamma \in \Gamma} \heartsuit \gamma \leq \smallvee \{
\heartsuit \delta \mid \forall \gamma\in\Gamma, \, \delta \Tleq \gamma
\}.\end{equation}
Recall from \S \ref{Subsect:FramePres} that since $\struc{   \T L, \Tb {\leq},
\cov_0^\bbL }$ is a flat site, we know that $1 = \smallvee \{
\heartsuit \alpha \mid \alpha \in \T L\}$ and that for all
$\alpha,\beta \in \T L$, $\heartsuit \alpha \wedge \heartsuit \beta = \smallwedge \{ \heartsuit
\gamma \mid \delta \Tleq \alpha,\beta\}$. It now follows by
induction on the size of $\Gamma$ that
\eqref{Flatsite:VTFlatsite:C1} holds.

Finally for $(\nabla 3)$, take $\Phi \in \T \P L$. We have to show
that $\heartsuit \T \smallvee (\Phi) \leq \smallvee \{ \heartsuit
\beta \mid \beta \in \nbsem (\Phi)\}$. This follows
immediately from the definition of $\cov_0^\bbL$, since $ \T \smallvee
(\Phi) \rel{\Tb {\leq}} \T \smallvee (\Phi) $. We conclude that
$\heartsuit\colon \T L \to \Fr \struc{ \T L, \Tb
{\leq}, \cov_0^\bbL }$ is compatible with the relations $(\nabla
1)$, $(\nabla 2)$ and $(\nabla 3)$ and thus there must be a unique frame
homomorphism $f\colon \V_{\T} \bbL \to \Fr \struc{ \T L, \Tb {\leq},
\cov_0^\bbL }$ which extends $\heartsuit$.

(2). We first have to show that $\nabla$ is
order-preserving, i.e.~that if $\alpha \rel{\Tb {\leq}} \beta$, then
$\nabla \alpha \leq \nabla \beta$. This follows immediately from
$(\nabla 1)$. Secondly, we must show that (2)(b) and (2)(c) are
satisfied, but this follows immediately from $(\nabla 2')$. Finally,
consider (2)(d), i.e.~suppose that $\alpha \cov_0^\bbL \Gamma$.
By definition of $\cov_0^\bbL$, there is some $\Phi \in \T\P L$ such
that $\alpha \Tleq \T \smallvee (\Phi)$ and $\nbsem(\Phi) =
\Gamma$. Now we need to show that
$\nabla \alpha \leq \smallvee \{ \nabla \beta \mid \beta \in
\nbsem (\Phi)\}$. This is easy to see, since
\begin{align*} \nabla \alpha & \leq \nabla \T \smallvee (\Phi)
&&\text{by $(\nabla 1)$,}\\
&\leq \smallvee \{ \nabla \beta \mid \beta \in \nbsem(\Phi) \}
&&\text{by $(\nabla 3)$.}
\end{align*}
It follows that (2)(d) holds; consequently, there exists  a unique frame homomorphism \[g\colon \Fr
\struc{ \T L, \Tb {\leq},\cov_0^\bbL } \to \V_{\T} \bbL,\] extending $\nabla$.

Finally, it is easy to see that \[gf = \id_{\VT \bbL} \text{ and } fg = \id_{\struc{{ \T L, \Tb {\leq},\cov_0^\bbL }}},\] so that indeed
$\V_\T \bbL \simeq \Fr \struc{   \T L, \Tb {\leq}, \cov_0^\bbL}$.
\end{proof}

In light of Theorem \ref{Flatsite:VTFlatsite} above, we denote the insertion of generators by $\nabla \colon \T
L \to \Fr \struc {\T L, \Tb {\leq}, \cov_0^\bbL}$. We now arrive at
the most important corollary of Theorem \ref{Flatsite:VTFlatsite},
which says that every element of $\VT \bbL$ has a disjunctive normal form.

\begin{corollary} \label{Flatsite:NormalForm}
Let $\T\colon \Set \to \Set$ be a standard, finitary, weak pullback-preserving functor
and let $\bbL$ be a frame. Then for all $x \in \VT \bbL$, there is a
$\Gamma \in \P\T L$ such that $x = \smallvee \{ \nabla \gamma\mid
\gamma\in \Gamma\}$.
\end{corollary}
\begin{proof}
By Theorem \ref{Flatsite:VTFlatsite} we know that $\VT \bbL \simeq \SupLat
\struc{\T L, \Tleq, \cov_0^\bbL}$. The statement now follows by
Fact~\ref{Prelim:FlatSiteCoverage}.
\end{proof}

\begin{remark}
It is not hard to show that \[\SupLat
\struc{\T L, \Tleq, \cov_0^\bbL} \simeq \SupLat \struc{ \T L \mid (\nabla 1), (\nabla 3)}.\]
Consequently, by Theorem \ref{Flatsite:VTFlatsite} and Fact
\ref{Prelim:FlatSiteCoverage}, the \emph{order} on $\VT \bbL$ is
uniquely determined by the relations $(\nabla 1)$ and $(\nabla 3)$.
\end{remark}

\section{Preservation results}
\label{Section:Preservation}
Now that we have established the $\T$-powerlocale construction, we can
set about to prove that it is well-behaved. One particular kind of
good behavior is to ask that it preserves algebraic
properties. In this section, we present several initial results in
this area. %
We start by briefly reviewing some of the preservation properties of $\V$, the usual Vietoris powerlocale, in \S \ref{sec:pres-prop-v}; in addition, we prove that $\V$ preserves compactness. In \S \ref{sec:reg-and-zer}, we show that $\VT$, the $\T$-powerlocale construction, preserves regularity and zero-dimensionality. Finally, in \S \ref{Subsect:Compactness} we show that if we assume that $\T$ maps finite sets to finite sets, then $\VT$ preserves the combination of compactness and zero-dimensionality.

\subsection{Preservation properties of $\V$}
\label{sec:pres-prop-v}

There are various relations between properties of $\mathbb{L}$ and of
$\V\mathbb{L}$.
For instance, \citep{VietLoc} shows that $\mathbb{L}$ is regular, completely
regular, zero-dimensional or compact regular iff $\V\mathbb{L}$ is, and also
that if $\mathbb{L}$ is locally compact then so is $\V\mathbb{L}$.
The same paper also mentions without proof that if $\mathbb{L}$ is compact then
so is $\V\mathbb{L}$, referring to a proof by transfinite induction similar to
that used
for the localic Tychonoff theorem in \citep{Johnstone1982}.
The paper leaves open the converse question, of whether $\V\mathbb{L}$ compact
implies so is $\mathbb{L}$.
We shall give here a
constructive (topos-valid) proof using preframe techniques that $\mathbb{L}$ is compact
iff $\V\mathbb{L}$ is.

\begin{definition}
A frame $\mathbb{L}$ (or, more properly, its locale) is \emph{compact}
if whenever $1\leq\bigvee_{i}^{\uparrow}a_{i}$ then $1\leq a_{i}$ for some $i$.
\end{definition}

The following constructive proof is a routine application of the techniques in
\citep{PrePrePre}.
\begin{theorem}
$\mathbb{L}$ is compact iff $\V\mathbb{L}$ is.
\end{theorem}

\begin{proof}
$\Rightarrow$: $\mathbb{L}$ is compact iff the function $\mathbb{L}%
\rightarrow\Omega$ that maps $a\in\mathbb{L}$ to the truth value of $a=1$ is a
preframe homomorphism, i.e. preserves finite meets and directed joins. This
function is characterized by being right adjoint to the unique frame
homomorphism $!\colon \Omega\rightarrow\mathbb{L}$ and so to prove compactness it
suffices to define a preframe homomorphism $\mathbb{L}\rightarrow\Omega$ and
show that it is right adjoint to $!$. If $\mathbb{L}$ is presented -- as a
\emph{frame} -- by generators and relations, then the ``preframe coverage
theorem'' of \citep{PrePrePre} shows how to derive a presentation as preframe,
which can then be used for defining preframe homomorphisms from $\mathbb{L}$.
The strategy is to generate a $\vee$-semilattice from the generators, and add
relations to ensure a $\vee$\emph{-stability} condition analogous to the
$\wedge$-stability used in Johnstone's coverage theorem \citep{Johnstone1982}.

Our first step is to apply the preframe coverage theorem to derive a preframe
presentation of $\V\mathbb{L}$. We show
\[
\begin{array}{rll}
\V\mathbb{L}\cong\Fr\langle\Pom\mathbb{L}\times\mathbb{L}
& \text{(qua }\vee\text{-semilattice)}\mid\\
&  1\leq(\gamma\cup\{1\},d)\\
&  (\gamma\cup\{a\},d)\wedge(\gamma\cup\{b\},d)
\leq(\gamma\cup\{a\wedge b\},d)\\
&  (\gamma\cup\{\smallvee^{\uparrow} A\},d)\leq\smallvee^{\uparrow}_{a\in A}(\gamma\cup\{a\},d)
& \text{(}A\text{ directed)}\\
&  (\gamma,\smallvee^{\uparrow} A\vee d)\leq\smallvee^{\uparrow}_{a\in A}(\gamma,a\vee d)
& \text{(}A\text{ directed)}\\
&  (\gamma\cup\{a\},d)\wedge(\gamma,b\vee d)\leq(\gamma,(a\wedge b)\vee d)\\
&  (\gamma\cup\{a\vee b\},d)\leq(\gamma\cup\{a\},b\vee d)\\
& \rangle\text{.}%
\end{array}
\]
The $\vee$-semilattice structure on $\Pom\mathbb{L}\times\mathbb{L}$ is the
product structure from $\cup$ on $\Pom\mathbb{L}$ and $\vee$ on $\mathbb{L}$.
The homomorphisms between the frame presented above and $\V\mathbb{L}$ are given by%
\begin{align*}
\Box a  & \mapsto(\{a\},0)\text{, }\Diamond a\mapsto(\emptyset,a)\\
(\gamma,d)  & \mapsto\bigvee_{c\in\gamma}\Box c\vee\Diamond d\text{.}%
\end{align*}

The relations shown are $\vee$-stable, so the preframe coverage shows that%
\[
\V\mathbb{L}\cong\operatorname*{PreFr}\langle\Pom\mathbb{L}\times
\mathbb{L}\text{ (qua poset)}\mid\text{same relations as above }%
\rangle\text{.}%
\]
We can now define a preframe homomorphism $\phi\colon \V\mathbb{L}\rightarrow\Omega$
by%
\[
\phi(\gamma,d)=\exists c\in\gamma.\,c\vee d=1\text{.}%
\]
To motivate this, we want criteria for $\bigvee_{c\in\gamma}\Box c\vee\Diamond
d=1$, and intuitively this means that for every sublocale $K$ corresponding to
a point of $\V\mathbb{L}$ either $K$ is included in some $c\in\gamma$ or $K$
meets $d$. Taking $K$ to be the closed complement of $d$, we get the given
condition. This is not a rigorous argument, since that closed complement is
not necessarily a point of $\V\mathbb{L}$. However, the rest of our argument
validates the choice. The relations in the preframe presentation of
$\V\mathbb{L}$ are largely easy to check. We shall just mention the
penultimate one. Suppose $(\gamma\cup\{a\},d)$ and $(\gamma,b\vee d)$ are both
mapped to $1$. We have either some $c\in\gamma$ with $c\vee d=1$, in which
case $c\vee(a\wedge b)\vee d=1$, or we have $a\vee d=1$ and in addition some
$c^{\prime}\in\gamma$ with $c^{\prime}\vee b\vee d=1$. In this latter case
$c^{\prime}\vee(a\wedge b)\vee d=1$.

Next we show that $\phi$ is right adjoint to $!\colon \Omega\rightarrow\V\mathbb{L}%
$, the unique frame homomorphism defined by%
\[
!(p)=\bigvee\left\{  1\mid p\right\}  =\bigvee\nolimits^{\uparrow}\left(
\{0\}\cup\left\{  1\mid p\right\}  \right)  \text{.}%
\]
We must show $\phi(!(p))\geq p$ and $!(\phi(\gamma,d))\leq(\gamma,d)$.%
\begin{align*}
\phi(!(p))  & =\phi\left(  \bigvee\nolimits^{\uparrow}\left(  \{0\}\cup
\left\{  1\mid p\right\}  \right)  \right)  \\
& =\phi(\emptyset,0)\vee\bigvee\left\{  \phi(\{1\},0)\mid p\right\}  \geq p
\end{align*}
since if $p$ holds then the disjuncts include $\phi(\{1\},0)=1$. For the other
inequality, we must show that%
\[
\bigvee\left\{  1\mid\phi(\gamma,d)\right\}  \leq(\gamma,d)\text{.}%
\]
If $\phi(\gamma,d)$ holds true then $c\vee d=1$ for some $c\in\gamma$, so%
\[
1=(\{1\},0)=(\{c\vee d\},0)\leq(\{c\},d)\leq(\gamma,d)\text{.}%
\]

$\Leftarrow$: Suppose in $\mathbb{L}$ we have $1=\bigvee_{i}^{\uparrow}a_{i}$.
Then in $\V\mathbb{L}$ we have $1=\Box1=\bigvee_{i}^{\uparrow}\Box a_{i}$ and
so $1=\Box a_{i}$ for some $i$. Applying $\counit$ to both sides gives $1=a_{i}$.
\end{proof}

\newcommand{\downR}{\mathop{\Downarrow}}
\newcommand{\RL}{C_{\bbL}}

\subsection{Regularity and zero-dimensionality}
\label{sec:reg-and-zer}
The purpose of this subsection is to prove that the operation $\VT$ preserves
regularity and zero-dimensionality of frames. Both of these notions
are defined in terms of the well-inside relation $\wi$; accordingly, the main
technical result of this subsection states that if $\alpha \Twi
\beta$, then also $\nabla\alpha \wi_{\VT \bbL} \nabla \beta$.
We first recall some notions leading up to the definition of regularity.

\begin{definition}
Given two elements $a,b$ of a distributive lattice $\bbL$, we say that $a$
is \emph{well inside} $b$, notation: $a \wi b$, if there is some $c$ in $\bbL$
such that $a \land c = 0$ and $b \lor c = 1$.
If $a \wi a$ we say $a$ is \emph{clopen}. We denote the clopen
elements of $\bbL$ by $\RL$.
\end{definition}

In case $\bbL$ is a frame, in the definition of $\wi$, for the element $c$
witnessing that $a \wi b$ we may always take the Heyting complementation
$\neg a$ of $a$.
In other words, $a\wi b$ iff $b \lor \neg a = 1$. Consequently, if $a$
is clopen then $a \vee \neg a =1$.
In the sequel we will use not only this fact, but also the following properties
of $\wi$ without warning. For proofs, see \citep[\S
III-1.1]{Johnstone1982}.

\begin{fact}
\label{f:wi}
Let $\bbL$ be a frame.
\begin{enumerate}
\item ${\wi} \sse {\leq}$;
\item ${\leq}\cor{\wi}\cor{\leq} \sse {\wi}$;
\item for $X \in \Pom L$, if $\forall x\in X. x \wi y$ then $\bv X \wi y$;
\item for $X \in \Pom L$, if $\forall x\in X. y \wi x$ then $y \wi \bw
X$;
\item $a \wi a$ iff $a$ has a complement.
\end{enumerate}
\end{fact}

\begin{definition}
A frame $\bbL$ is \emph{regular} if every $a \in \bbL$ satisfies
\begin{equation*}
\label{eq:dreg}
a = \smallvee \{ b \in L \mid b \wi a\}.
\end{equation*}
We say $\bbL$ is \emph{zero-dimensional} if for all $a \in \bbL$,
\[ a = \smallvee \{b \in \RL \mid b \leq a\}.\]
\end{definition}
We record the following useful property of $\RL$ \citep[see][\S
III-1.1]{Johnstone1982}:
\begin{fact} \label{Reg:RegOpenFacts}
Let $\bbL$ be a frame. Then
$\langle \RL, \wedge, \vee, 0, 1\rangle$ is a
sublattice of $\bbL$.
\end{fact}

We define a function $\downR \colon \P L \to \P \RL$ which maps $A\in
\P L$ to $\downarr A \cap \RL$.

\begin{lemma} \label{Compactness:DownRProp}
Let $\T\colon \Set \to \Set$ be a standard, finitary, weak
pullback-preserving functor. If  $\bbL$ is a zero-dimensional frame, then
\begin{enumerate}
\item $\forall \alpha \in \T L$, $\nabla \alpha = \smallvee \{ \nabla
\beta \mid \beta\in \T \RL, \, \beta \Tleq \alpha\}$;
\item $\forall \Phi \in \T\P L$, $\T \smallvee (\Phi) = \T \smallvee
\cof \T \downR (\Phi)$;
\item $\forall \Phi \in \T \P L$, $\forall \alpha \in \T L$, $[ \alpha
\in \T \RL$ and $\alpha \Tleq \cor \rel{\Tb {\in}} \Phi
]$ iff $\alpha \in \nbsem \left(\T \downR (\Phi) \right)$.
\end{enumerate}
Similarly to (1), if $\bbL$ is regular then
$\forall \alpha \in \T L,\, \nabla \alpha = \smallvee \{ \nabla \beta
\mid \beta \in \T L,\, \beta \Twi \alpha\}$.
\end{lemma}

\begin{proof}
(1). First, observe that for all $a \in L$, we have
that
\begin{align*} a& =
\smallvee \{ b \in \RL \mid b \leq a\} &&\text{by zero-dimensionality,}\\
&= \smallvee \downR \{ a\} &&\text{by definition of $\downR$,}\\
&= \smallvee \downR \cof \single (a) &&\text{by def.~of $\eta \colon
\operatorname{Id}_{\Set} \to \P$.}
\end{align*}
By relation lifting, it follows that
\begin{equation} \label{Compactness:DownRProp:C1}
\forall \alpha \in \T L,
\, \alpha = \T \smallvee \cof \T \downR \cof \T \single (\alpha).
\end{equation}
Now observe that for all $a,b \in L$, we have $b \in \downR \eta(a)$ iff
$b \in \RL$ and $b \leq a$. By relation lifting, it follows that
\begin{equation}\label{Compactness:DownRProp:C2}
\forall \alpha,\beta \in \T L, \, \left[\beta \Tin \T \downR \cof \T
\single (\alpha) \text{ iff } \beta \in \T \RL \text{ and } \beta
\Tleq \alpha \right].
\end{equation}
Combining these two observations, we see that
\begin{align*}
\nabla \alpha &= \nabla \left( \T \smallvee \cof \T \downR \cof \T
\single (\alpha) \right) &&\text{by \eqref{Compactness:DownRProp:C1},}\\
&= \smallvee \{ \nabla \beta \mid \beta \Tin \T \downR \cof \T
\single (\alpha) \}&&\text{by $(\nabla 3)$,}\\
&= \smallvee \{ \nabla \beta \mid \beta \in \T \RL ,\, \beta \Tleq
\alpha \} &&\text{by \eqref{Compactness:DownRProp:C2}.}
\end{align*}

(2).
It follows by zero-dimensionality of $\bbL$ that for all $A \in \P L$, we
have $\smallvee A = \smallvee \downR A$.
Consequently, by relation lifting, (2) holds.

(3). Take $a \in L $ and $A \in \P L$. Then
\begin{align*}
a \in \downR A &\Leftrightarrow a \in \RL \text{ and } \exists b \in
A,\,a \leq b&&\text{by definition of $\downR$,}\\
&\Leftrightarrow a \in \RL \text{ and } a \rel{{\leq} \cor {\in }} A
&&\text{by def.~of relation composition.}
\end{align*}
It follows by relation lifting that
\begin{equation*}\label{Compactness:DownRProp:C3}
\forall \Phi \in \T \P L, \forall \alpha \in \T L,\, \alpha \Tin \T
\downR (\Phi) \text{ iff } \alpha \in \T \RL \text{ and } \alpha
\rel{\T {\leq} \cor \T {\in }} \Phi.
\end{equation*}
Now it follows by definition of $\nbsem (\Phi)$ that (3) holds.

For the last part of the proof, first observe that if $\bbL$ is
regular, then for all $a \in L$, $a = \smallvee w(a)$, where we
temporarily define $w  \colon L \to \P L$ as
\[w \colon a \mapsto \{ b\in L \mid b \wi a\}.
\]
By relation
lifting, it follows that
\begin{equation}\label{Compactness:DownRProp:C4}
\T \smallvee \cof \T w = \id_L.
\end{equation}
Moreover, it follows by definition of $w \colon L \to \P L$ that for
all $a,b \in L$, $ b \in w(a)$ iff $b \wi a$. Consequently,
\begin{equation}\label{Compactness:DownRProp:C5}
\forall \alpha,\beta \in \T L,\, \beta \Tin \T w (\alpha) \text{ iff }
\beta \Twi \alpha.
\end{equation}
Now we see that for any $\alpha \in \T L$,
\begin{align*}
\nabla\alpha &= \nabla \left( \T \smallvee \cof \T w (\alpha) \right) &&\text{by
\eqref{Compactness:DownRProp:C4},}\\
&= \smallvee \{ \nabla \beta \mid \beta \Tin \T w(\alpha)\} &&\text{by
$(\nabla 3)$,}\\
&= \smallvee \{ \nabla \beta \mid \beta \Twi \alpha \}&&\text{by \eqref{Compactness:DownRProp:C5}.}
\end{align*}
\end{proof}

The key technical lemma of this subsection states that relation
lifting preserves the $\wi$-relation.

\begin{lemma}
\label{p:rlwi}
Let $\T$ be a standard, finitary, weak pullback-preserving functor and let
$\bbL$ be a frame.
Then
\begin{equation}
\label{eq:rg1}
\text{for all } \al,\be \in \T L:\;  \al \Twi \be
\text{ implies } \nb\al \wi_{\VT\bbL} \nb\be.
\end{equation}
\end{lemma}

\begin{proof}
Let $\al,\be \in \T L$ be such that $\al \Twi \be$.
Our aim will be to show that $\nb\al \wi_{\VT\bbL} \nb\be$.

We may assume without loss of generality that
\begin{multline}
\label{eq:rg2}
\be = (\T f) \al \text{ for some } f\colon  \Base^\T (\al) \to
\Base^\T (\be) \\
\text{ such that } a \wi fa \text{ for all } a \in \Base^\T (\al).
\end{multline}

To justify this assumption, assume that we have a proof of \eqref{eq:rg1}
for all $\be$ satisfying \eqref{eq:rg2}.
To derive \eqref{eq:rg1} in the general case, consider arbitrary elements
$\al,\be'\in\T L$ such that $\al \Twi \be'$.
In order to show that $\nb\al \Twi \nb\be'$, consider the map $f\colon
\Base^\T (\al) \to L$ given by $f(a) := \bw \{ b \in \Base^\T (\be') \mid a \wi b \}$.
On the basis of Fact~\ref{f:wi} it is not difficult to see that
$\Graph{(f)} \sse {\wi}$ and so by the
properties of relation lifting we obtain $\Graph{(\T f)} \sse {\Twi}$.
In particular, we find that $\al \Twi (\T f)\al$; thus by our assumption
we may conclude that $\nb\al \wi \nb(\T f)\al$.
Also, observe that $a \wi b$ implies $fa \leq b$, for all $a\in\Base^\T
(\al)$ and $b \in \Base^\T (\be')$.
Hence by Lemma~\ref{p:rl2} we may conclude from $\al \Twi \be'$ that
$(\T f)\al \Tleq \be'$, which gives $\nb(\T f)\al \leq \nb\be'$.
Combining our observations thus far, by Fact~\ref{f:wi} it follows
from $\nb\al \wi \nb(\T f)\al$ and $\nb(\T f)\al \leq \nb\be'$ that
$\nb\al \wi \nb\be'$ indeed.
Thus our assumption \eqref{eq:rg2} is justified indeed.

Turning to the proof itself, consider the map $h\colon  \P\Base^\T (\al) \to L$ given by
\[
h(A)  := \bw \left(
\{ \neg a \mid a \in A \} \cup \{ fa \mid a \not\in A \}
\right).
\]
Our first observation is that, since by assumption $\neg a \lor f a =
1_{\bbL}$ for each $a \in \Base^\T (\al)$, we may infer that
\[
1_{\bbL} = \bw \{ \neg a \lor fa \mid a \in \Base^\T (\al) \},
\]
a straightforward application of the (finitary) distributive law yields
that
\begin{equation}
\label{eq:reg1}
1_{\bbL} = \bv \{ h(A) \mid A \in \P\Base^\T (\al) \}.
\end{equation}
Define $X \sse L$ to be the range of $h$, so that we may think of $h$ as a
surjection $h\colon  \P\Base^\T (\al) \to X$, and read \eqref{eq:reg1} as saying that
$1 = \bv X$.
Using Lemma~\ref{p:basicV}(5), from the latter observation we may infer
that
\begin{equation}
\label{eq:reg1b}
1_{\VT\bbL} = \bv \{ \nb\xi \mid \xi \in \T X \}.
\end{equation}
However, from $h\colon  \P\Base^\T (\al) \to X$ being surjective we may infer that
$\T h\colon  \T\P\Base^\T (\al) \to \T X$ is also surjective, so that we may read
\eqref{eq:reg1b} as
\begin{equation}
\label{eq:reg2a}
1_{\VT\bbL} = \bv \{ \nb\T h (\Phi) \mid
\Phi \in \T\P\Base^\T (\al) \}.
\end{equation}
This leads us to the key observation in our proof:
We may partition the set $\{ \T h (\Phi) \mid \Phi \in \T\P\Base^\T (\al) \}$
into elements $\ga$ such that $\nb\ga \leq \nb\be$, and elements $\ga$
satisfying $\nb\al\land\nb\ga = 0_{\VT\bbL}$.

\begin{claimfirst}
Let $\Phi \in  \T\P\Base^\T (\al)$.
\\(a) If $(\al,\Phi) \in \Tb{\not\in}$, then
$\T h (\Phi) \Tleq \be$;
\\(b) if $(\al,\Phi) \not\in \Tb{\not\in}$, then
$\nb\al \land \nb\T h(\Phi) = 0_{\VT\bbL}$.
\end{claimfirst}

\begin{pfclaim}
For part~(a), it is not hard to see that
\[
a \not\in A \Rightarrow h(A) \leq f(a),
\text{ for all } a \in \Base^\T (\al), A \in \P\Base^\T (\al).
\]
From this it follows by Lemma~\ref{p:rl2} that
\[
\al \rel{\Tb \not\in} \Phi \Rightarrow \T h(\Phi)
\Tleq (\T f)(\al) = \be.
\]

For part (b), assume that $\nb\al\land\nb Th (\Phi) > 0_{\VT\bbL}$.
It suffices to derive from this that $\al \rel{\Tb{\not\in}} \Phi$.

Let $\leq'$ be the restriction of $\leq$ to the non-zero part of $\bbL$,
that is, ${\leq'} := {\leq}\rst{L'\times L'}$, where $L' = L \setminus
\{0_{\bbL}\}$.
We claim that for all $\ga,\de \in \T L$:
\begin{equation}
\label{eq:reg01}
\nb\ga\land\nb\de > 0_{\VT\bbL}
\;\Rightarrow\; (\ga,\de) \in {\Tb{\geq'}}\cor{\Tb{\leq'}}.
\end{equation}
To see this, assume that $\nb\ga\land\nb\de > 0_{\VT\bbL}$, and observe that
Lemma~\ref{Flatsite:LowerBounds} yields the existence of a $\theta \in \T L$
such that $\nb\theta > 0_{\VT\bbL}$ and $\theta \Tleq \ga,\de$.
It follows from Lemma~\ref{p:basicV}(1) that $\ga,\de$ and $\theta$
all belong to $\T L'$, and so $\theta$ is witnesses to the fact that
$(\ga,\de) \in {\Tb{\geq'}}\cor{\Tb{\leq'}}$.

By \eqref{eq:reg01} and the assumption on $\al$ and $\Phi$ it follows that
$(\al,\Phi) \in {\Tb{\geq'}}\cor{\Tb{\leq'}}\cor\cv{(\Graph{Th})}$, and so by
Fact~\ref{f:rl} we obtain
\begin{equation}
\label{eq:reg02}
(\al,\Phi) \in {\Tb({\geq'}\cor{\leq'}\cor\cv{(\Graph{h})})}
\end{equation}

The crucial observation now is that
\begin{equation}
\label{eq:reg03}
{\geq'}\cor{\leq'}\cor\cv{(\Graph{h})} \;\sse\; {\not\in}.
\end{equation}
For a proof, take a pair $(a,A) \in L \times \P L$ in the LHS of
\eqref{eq:reg03}, and suppose for contradiction that $a \in A$.
Then by definition of $h$ we obtain $h(A) \leq \neg a$, so that $a \land
h(A) = 0_{\bbL}$.
But if $a \rel{{\geq'}\cor{\leq'}\cor\cv{(\Graph{h})}} A$, then there must
be some $b$ such that $b \leq' a, h(A)$, and by definition of $\leq'$ this
can only be the case if $b > 0_{\bbL}$.
This gives the desired contradiction.

Finally, by monotonicity of relation lifting, it is an immediate consequence
of \eqref{eq:reg02} and \eqref{eq:reg03} that $\al \rel{\Tb{\not\in}} \Phi$.
This finishes the proof of the Claim.
\end{pfclaim}

On the basis of the Claim it is straightforward to finish the proof.
Define
\[
c := \bv \left\{ \T h (\Phi) \mid
\Phi \in \T\P\Base^\T (\al) \text{ such that } (\al,\Phi) \not\in \Tb{\not\in}
\right\},
\]
then we may calculate that
\begin{align*}
&c \lor \nb\be\\
&\geq c \lor \bv \left\{
\T h (\Phi) \mid
\Phi \in \T\P\Base^\T (\al) \text{ such that } (\al,\Phi) \in \Tb{\not\in}
\right\}
&& \text{(Claim~1(a))}
\\&= \bv \left\{
\T h (\Phi) \mid
\Phi \in \T\P\Base^\T (\al)
\right\}
&& \text{(definition of $c$)}
\\&= 1_{\VT\bbL}
&& \text{(equation~\eqref{eq:reg2a})}
\end{align*}
and
\begin{align*}
&\nb\al \land c\\
&= \bv \left\{
\nb\al \land \T h(\Phi) \mid
\Phi \in \T\P\Base^\T (\al) \text{ such that } (\al,\Phi) \not\in \Tb{\not\in}
\right\}
&& \text{(distributivity)}
\\&= \bv \left\{ 0_{\VT\bbL} \mid
\Phi \in \T\P\Base^\T (\al) \text{ such that } (\al,\Phi) \not\in \Tb{\not\in}
\right\}
&& \text{(Claim~1(b))}
\\&= 0_{\VT\bbL}
\end{align*}
In other words, $c$ witnesses that $\nb\al \wi_{\VT\bbL} \nb\be$.
\end{proof}

We now arrive at the main result of this subsection, namely, that the
$\T$-powerlocale construction preserves regularity and zero-dimensionality.

\begin{theorem}
\label{t:reg}
Let $\bbL$ be a frame and let $\T$ be a standard, finitary, weak
pullback-preserving functor.
\begin{enumerate}
\item If $\bbL$ is regular then so is $\VT\bbL$.
\item If $\bbL$ is zero-dimensional then so is $\VT \bbL$.
\end{enumerate}
\end{theorem}

\begin{proof}%
(1).
By Corollary \ref{Flatsite:NormalForm}, it suffices to show that for
all $\alpha \in \T L$,
\begin{equation}\label{t:reg:C1}
\nabla \alpha = \smallvee \{ \nabla \beta \in \VT \bbL \mid  \nabla
\beta \wi \nabla \alpha\}.
\end{equation}
Take
$\alpha \in \T L$; we
see that
\begin{align*}
\nabla \alpha &= \smallvee \{ \nabla
\beta \mid \beta \Twi \alpha\} &&\text{by Lemma
\ref{Compactness:DownRProp},}\\
&\leq \smallvee \{ \nabla \beta \mid \nabla\beta \wi_{\VT \bbL} \nabla
\alpha\} &&\text{by Lemma \ref{p:rlwi},}\\
&\leq \nabla \alpha &&\text{since ${\wi} \subseteq {\leq}$.}
\end{align*}
It follows that \eqref{t:reg:C1} holds, concluding the proof of part (1).

(2). Again by Corollary \ref{Flatsite:NormalForm}, it suffices to show that for
all $\alpha \in \T L$,
\begin{equation}\label{t:reg:C2}
\nabla \alpha = \smallvee \{ \nabla \beta \mid  \nabla
\beta \in C_{\VT \bbL},\, \nabla \beta \leq \nabla \alpha\}.
\end{equation}
The main observation here is that
\begin{equation}\label{t:reg:C3}
\forall \beta \in \T \RL,\, \nabla\beta \in C_{\VT \bbL} .
\end{equation}
To see why, recall that $\RL := \{ b\in L \mid b \wi b\}$,  so that
for all $b \in \RL$, $b=b$ implies $b \wi b$. Consequently, by
relation lifting,
\[ \forall \beta \in \T\RL,\, \beta \Twi \beta.\]
It follows by Lemma \ref{p:rlwi} that \eqref{t:reg:C3} holds. Now
\begin{align*}
\nabla \alpha &= \smallvee \{ \nabla \beta \mid \beta \in \T \RL,\,
\beta \Tleq \alpha\}&&\text{by Lemma Lemma
\ref{Compactness:DownRProp}(1),}\\
&\leq \smallvee \{ \nabla \beta \in C_{\VT \bbL} \mid \beta \Tleq
\alpha \} &&\text{by \eqref{t:reg:C3},}\\
&\leq \smallvee \{ \nabla \beta \in C_{\VT \bbL} \mid \nabla\beta \leq
\nabla \alpha \} &&\text{by $(\nabla 1)$,}\\
&=\nabla \alpha &&\text{by order theory.}
\end{align*}
It now follows that \eqref{t:reg:C2} holds; consequently we see that
(2) holds.
\end{proof}

\newcommand{\covR}{\cov^{C}}
\newcommand{\VTR}{\V_{\T}^C}

\subsection{Compactness + zero-dimensionality} \label{Subsect:Compactness}

In this subsection, we will show that if $\bbL$ is compact and
zero-dimensional, then so is $\VT \bbL$. Our proof strategy is as follows.
Given a compact zero-dimensional frame $\bbL$, we will define a new
construction $\VTR \bbL$ which is guaranteed to be compact, and then
we show that $\VT \bbL \simeq \VTR \bbL$.

We define a flat site presentation $\struc{ \T \RL, \Tb
{\leq}, \covR_0}$, where
\[ \covR_0 := \{(\alpha,\nbsem(\Phi)) \in \T\RL \times \P \T L
\mid \alpha \rel{\Tb {\leq}} \T \smallvee (\Phi),\, \Phi \in \T \Pom \RL \}.\]
Observe that we view $\T \RL$ as a substructure of $\T L$, which is
justified by the fact that $\RL$ is a sublattice of $\bbL$ (Fact
\ref{Reg:RegOpenFacts}): this fact tells us that $\smallvee \colon \P L \to
L$ restricts to a function from $\Pom \RL$ to $\RL$; consequently,
$\T \smallvee$ maps $\T\Pom \RL$ to $\T \RL$, by standardness of $\T$.
Below, we will need the following property of relation lifting with
respect to ordered sets.

\begin{lemma}\label{Properties:Top}
Let $\T\colon \Set \to \Set$ be a standard, finitary, weak
pullback-preserving functor and
let $\mathbb{P}$ be a poset with a top element $1$.
Then for every $\beta \in \T P$ there is some $\alpha \in \T \{ 1\}$ such that
$\beta \Tleq \alpha$;
\end{lemma}
\begin{proof}
Consider the following function at the ground level: $f
\colon P \to \{1\}$, where $f$ is the constant function $f \colon b
\mapsto 1$. Then for all $b\in P$, we have $b \leq f(b)$ and $f(b)
\in \{ 1\}$. By relation lifting, we see that for all $\beta \in \T
P$, $\beta \Tleq \T f(\beta)$ and $\T f (\beta) \in \T \{ 1
\}$. The statement follows.
\end{proof}

\begin{lemma}\label{Compactness:VTRCompact}Let $\T\colon \Set \to
\Set$ be a standard, finitary, weak pullback-preserving functor and
let $\bbL$ be a frame. Then $\struc{ \T \RL, \Tb {\leq},
\covR_0}$ is a flat site. Moreover, if $\T$ maps finite sets to
finite sets then $\Fr \struc{ \T \RL, \Tb {\leq},
\covR_0}$ is a compact frame.
\end{lemma}
\begin{proof}
Because $\RL$ is a meet-subsemilattice of $\bbL$, we can apply Lemma
\ref{Flatsite:Stability} to $\T \RL$. Now the proof that $\struc{ \T \RL, \Tb {\leq},
\covR_0}$ is a flat site is analogous to that of Lemma
\ref{Flatsite:Well-defined}.

Now suppose that $\T$ maps finite sets to finite sets. Then for all
$\Phi \in \T\Pom \RL$, it follows
by Fact
\ref{f:rlnt}(3) that $\nbsem(\Phi)$ is finite. Consequently,
\begin{equation*}\label{Compactness:VTRCompact:C1}
\text{$\forall \alpha\covR_0 \nbsem(\Phi)$, $\nbsem(\Phi)$ is
finite.}\end{equation*}
Moreover, by Lemma \ref{Properties:Top},
\begin{equation*}\label{Compactness:VTRCompact:C2}
\T\RL = \downarr_{\T\RL} \T\{ 1_\bbL\},
\end{equation*}
since $1_\bbL \in \RL$ as $\RL$ is a sublattice of $\bbL$.
Since we assumed that $\T$ maps finite sets to finite sets, the set
$\T\{ 1_\bbL\}$ must be finite.
It now follows from a straightforward generalization of
\citep[Proposition 11]{Vickers2006} that $\Fr \struc{ \T \RL, \Tb {\leq},
\covR_0}$ is a compact frame. (The only change we need to make to
\citep[Proposition 11]{Vickers2006} is to generalize from using single
finite trees to using disjoint unions of $|\T \{ 1_\bbL\} |$-many trees, so that one can
cover each element of $\T \{ 1_\bbL\}$.)
\end{proof}

We define $\VTR \bbL := \Fr \struc{ \T \RL, \Tb {\leq},
\covR_0}$, and for the time being we denote the insertion of
generators by $\heartsuit \colon \T \RL \to \VTR \bbL$.
Our goal is now to show that $\VT \bbL \simeq \VTR \bbL$. We will use
a shortcut, exploiting the fact that both $\VT \bbL$ and $\VTR \bbL$
have flat site presentations: we will define \emph{suplattice}
homomorphisms $f' \colon \VT \bbL \to \VTR \bbL$ and $g' \colon \VTR
\bbL \to \VT \bbL$. We then show that $g' \cof f' = \id$ and $f' \cof
g' = \id$, so that $\VT \bbL$ and $\VTR \bbL$ are isomorphic as
suplattices. It then follows from order theory that they are also
isomorphic as frames.
We start by defining a function $g\colon \T \RL \to \VT \bbL$, defined as
\[ g \colon \alpha \mapsto \nabla \alpha.\]
\begin{lemma} \label{Compactness:VTRtoVT}
Let $\T\colon \Set \to \Set$ be a standard, finitary, weak
pullback-preserving functor and let $\bbL$ be a frame. Then the function $g$ defined above extends to a \emph{suplattice}
homomorphism $g' \colon \VTR \bbL \to \VT \bbL$ such that $g' \circ
\heartsuit = g$.
\[\xymatrix{ \VTR \bbL \ar@{..>}[r]^{g'} & \VT \bbL \\ \T \RL
\ar[u]^{\heartsuit} \ar[ru]_g }\]
\end{lemma}
\begin{proof}
We need to show that $g\colon \T \RL \to \VT \bbL$ preserves the order
on $\T \RL$ and preserves covers in to joins: if $\alpha \covR_0
\nbsem (\Phi)$, where $\alpha \in \T\RL$, $\Phi \in \T \P \RL$ and
$\alpha \Tleq \smallvee (\Phi)$, then $g(\alpha) \leq \smallvee \{
g(\beta) \mid \beta \in \nbsem(\Phi) \}$. Both of these properties
follow straightforwardly from the fact that $\struc{ \T \RL, \Tb {\leq},
\covR_0}$ is a substructure of $\struc{ \T L, \Tb {\leq},
\cov_0^\bbL}$.
\end{proof}

The next step is to define the suplattice homomorphism $f' \colon \VT
\bbL \to \VTR \bbL$. This requires a little more work than the
definition of $g'\colon \VTR \bbL \to \VT \bbL$, beginning with the
following lemma.

\begin{lemma} \label{Compactness:JoinsFinite}
Let $\T\colon \Set \to \Set$ be a standard, finitary, weak
pullback-preserving functor and
let $\bbL$ be a compact frame. If $\alpha \in \T \RL$ and $\Phi \in
\T \P \RL$ such that $\alpha \rel{\Tb {\leq}} \T \smallvee (\Phi)$,
then there exists $\Phi_\alpha \in \T \Pom \RL$ such that
$\Phi_\alpha \rel{\Tb {\subseteq_L}} \Phi$ and $\alpha \rel{\Tb
{\leq}} \T \smallvee (\Phi_\alpha)$.
\end{lemma}

\begin{proof}
Since $\bbL$ is compact, we can show that
\begin{equation}\label{Compactness:JoinsFinite:C1}
\text{for all $a \in \RL$, $a$ is compact.}
\end{equation}
After all, if $a \in \RL$ and $A \in \P L$ such that $a \leq
\smallvee A$, then also $1 \leq a \vee \neg a \leq \smallvee A \cup \{\neg a\}$, so by
compactness of $\bbL$, there exists a finite $A' \subseteq A$ such
that $a\vee \neg a \leq \smallvee A' \cup \{ \neg
a\}$. Consequently, $a \leq \smallvee A'$. Since $A$ was arbitrary,
it follows that $a$ is compact.

We define \[S:= \left( \leq \cor \cv{\Graph(\smallvee)}
\right)\rst{\RL \times \P \RL};\] so that $(a,A) \in S$ iff $a\in
\RL$, $A \in \P \RL$ and $a \leq \smallvee A$. By
\eqref{Compactness:JoinsFinite:C1}, we can define a function
$h\colon S \to S$ where $h\colon (a,A) \mapsto (a',A')$ such that
$a=a'$, $A' \subseteq A$, $a' \leq \smallvee A'$ (otherwise $h$
would not be well-defined) and such that $A'$ is finite, i.e.~$A'
\in \Pom \RL$. In other words, $h\colon S \to S$ is a function which
assigns a finite subcover $A'$ to a set of zero-dimensional opens $A$
covering a zero-dimensional open element $a$ . If we denote the projection
functions of $S$ as \[ \xymatrix{ \RL & S \ar[l]_{p_1} \ar[r]^{p_2}
& \P \RL}\] then we can encode the above-mentioned properties of
$h$ as follows:
\begin{align*}
\forall x \in S,\, & p_1 \cof h(x) = p_1(x);\\
\forall x \in S,\, & p_2 \cof h (x ) \subseteq p_2(x);\\
\forall x \in S,\, & p_2 \cof h(x) \in \Pom \RL.
\end{align*}
By relation lifting, it follows that
\begin{align}
\forall x \in \T S,\, & \T p_1 \cof \T h(x) = \T
p_1(x); \label{Compactness:JoinsFinite:C2}\\
\forall x \in \T S,\, & \T p_2 \cof \T h (x ) \rel{\Tb{\subseteq}}
\T p_2(x); \label{Compactness:JoinsFinite:C3}\\
\forall x \in \T S,\, & \T p_2 \cof \T h(x) \in \T \Pom
\RL. \label{Compactness:JoinsFinite:C4}
\end{align}
Finally, observe that it follows by relation lifting that
\[ \forall \alpha \in \T \RL, \forall \Phi \in \T \P \RL,\, \alpha
\Tleq \smallvee (\Phi) \text{ iff } \alpha \rel{\Tb S} \Phi.\] Now
take $\alpha \in \T\RL$ and $\Phi \in \T \P \RL$ such that $\alpha
\Tleq \smallvee (\Phi)$. Then by the above, we have $\alpha \rel{\Tb
S} \Phi$, so by definition of $\Tb$ there must exist some $x
\in \T S$ such that $\T p_1 (x) = \alpha$ and $\T p_2(x) = \Phi$. We
define $\Phi_\alpha := \T p_2 \cof \T h(x)$; observe that $\T p_1
\cof \T h (x) = \T p_1 (x) = \alpha$ by
\eqref{Compactness:JoinsFinite:C2}.  Since $\T h$ is a function from
$\T S$ to $\T S$, we see that $\alpha \rel{\Tb S} \Phi_\alpha$, so
that $\alpha \Tleq \T \smallvee (\Phi_\alpha)$. Moreover by
\eqref{Compactness:JoinsFinite:C3} $\Phi_\alpha \rel{\Tb {\subseteq}}
\Phi$ and by \eqref{Compactness:JoinsFinite:C4}, $\Phi_\alpha \in \T
\Pom \RL$. This concludes the proof.
\end{proof}
We now define a
map $f\colon \T L \to \VTR \bbL$ by sending \[ f \colon \alpha \mapsto
\smallvee \{ \heartsuit \beta \mid \beta \in \T \RL,\, \beta \rel{\Tb
{\leq}} \alpha \}.\] This will give us our suplattice homomorphism
$f' \colon \VT \bbL \to \VTR \bbL$.

\begin{lemma} \label{Compactness:VTtoVTR}
Let $\T\colon \Set \to \Set$ be a standard, finitary, weak pullback-preserving functor.
If $\bbL$ is a compact zero-dimensional frame then $f\colon \T L \to \VTR
\bbL$ defined above extends to a \emph{suplattice} homomorphism $f'
\colon\V_\T \bbL \to \VTR \bbL$, where $f' \circ \nabla =
f$. \[\xymatrix{ \VT \bbL \ar@{..>}[r]^{f'} & \VTR \bbL \\ \T L
\ar[u]^{\nabla} \ar[ru]_f }\]
\end{lemma}

\begin{proof}

In order to show that $f\colon \T L \to \VTR \bbL$ extends to a
suplattice homomorphism, we need to show that $f$ preserves the
order on $\T L$ and $f$ transforms covers into joins, i.e.~that for all $( \alpha, \nbsem (\Phi) ) \in
\cov_0$, where $\alpha \Tleq \T \smallvee (\Phi)$, we have
$f(\alpha) \leq \smallvee \{ f(\gamma ) \mid \gamma \in \nbsem
(\Phi)\}$. To see why $f$ is order-preserving, suppose that
$\alpha_0,\alpha_1 \in \T L$ and that $\alpha_0 \Tleq
\alpha_1$. Then
\begin{align*}
f(\alpha_0) &= \smallvee \{ \heartsuit \beta \mid \beta \in \T
\RL,\, \beta \rel{\Tb
{\leq}} \alpha_0 \} &&\text{by definition of $f$,}\\
&\leq \smallvee \{ \heartsuit \beta \mid \beta \in \T \RL,\, \beta
\rel{\Tb {\leq}} \alpha_1 \}&&\text{since $\beta \Tleq \alpha_0
\Tleq
\alpha_1 \Rightarrow \beta \Tleq \alpha_1$,}\\
&= f(\alpha_1) &&\text{by definition of $f$.}
\end{align*}
Before we go ahead and show that $f$ transforms covers $\alpha \cov_0
\nbsem(\Phi)$ into joins, we show that the expression $\smallvee
\{ f(\gamma ) \mid \gamma \in \nbsem (\Phi)\}$ can be simplified:
\begin{equation}\label{Compactness:VTtoVTR:C2}
\forall \Phi \in \T \P L, \, \smallvee \{ f(\gamma) \mid \gamma
\in \nbsem (\Phi) \} = \smallvee \{ \heartsuit \beta \mid
\beta\in \nbsem \left(
\T \downR (\Phi) \right)\}.
\end{equation}
To see why, observe that
\begin{align*}
&\smallvee \{ f(\gamma) \mid \gamma
\in \nbsem (\Phi) \}\\
&= \smallvee \left\{ \smallvee \{ \heartsuit \beta \mid \beta \in
\T \RL,\, \beta \leq \gamma \} \mid \gamma \in \nbsem (\Phi)
\right\}&&\text{by definition of $f$,}\\
&= \smallvee \left\{ \smallvee \{ \heartsuit \beta \mid \beta \in
\T \RL,\, \beta \leq \gamma \} \mid \gamma \rel{\Tb {\in}} \Phi
\right\}&&\text{by definition of $\nbsem$,}\\
&= \smallvee \{ \heartsuit \beta \mid \beta \in \T \RL,\, \exists
\gamma \rel{\Tb {\in}} \Phi, \, \beta \leq \gamma
\}&&\text{by associativity of $\smallvee$,}\\
&= \smallvee \{ \heartsuit \beta \mid \beta \in \T
\RL,\, \beta \Tleq \cor \rel{\Tb {\in}} \Phi \}&&\text{by def.~of
relation composition,}\\
&= \smallvee \{ \heartsuit \beta \mid \beta \in \nbsem
\left(\T \downR (\Phi) \right) \}&&\text{by Lemma
\ref{Compactness:DownRProp}(3).}
\end{align*}

Let $\alpha \in \T L$ and $\Phi \in \T \P L$ such that $\alpha
\rel{\Tb {\leq}} \T \smallvee (\Phi)$; we need to show that
$f(\alpha) \leq \smallvee \{ f(\gamma) \mid \gamma
\in \nbsem (\Phi) \} $. By
\eqref{Compactness:VTtoVTR:C2} it suffices to show that
\begin{equation} \label{Compactness:VTtoVTR:C3}
f(\alpha) \leq \smallvee  \{ \heartsuit \gamma \mid
\gamma\in \nbsem \left(
\T \downR (\Phi) \right)\}.
\end{equation}
Recall that $f(\alpha) = \smallvee \{ \heartsuit \beta \mid \beta
\in \T \RL,\, \beta \leq \alpha \} $. We will show that
\begin{equation}
\label{Compactness:VTtoVTR:C4}\forall \beta \in \T
\RL, \, \beta \Tleq \alpha \Rightarrow \heartsuit \beta \leq
\smallvee \{ \heartsuit \gamma \mid \gamma\in \nbsem \left( \T
\downR (\Phi) \right)\} .
\end{equation}
Suppose that $\beta \in \T \RL$ and that $\beta \Tleq \alpha$. Then
since we assumed that $\alpha \Tleq \T \smallvee (\Phi)$, it follows
that $\beta \Tleq \T \smallvee (\Phi)$. By Lemma
\ref{Compactness:DownRProp}(2), we know that $T \smallvee (\Phi) =
\T \smallvee \cof \T \downR (\Phi)$, so we see that \[\beta \Tleq \T
\smallvee \cof \T \downR (\Phi).\] Now since $\T \downR (\Phi) \in
\T \P \RL$, we can now apply Lemma \ref{Compactness:JoinsFinite} to
conclude that there must be some $\Phi' \in \T \Pom \RL$ such that
$\Phi' \rel{\Tb {\subseteq}} \T\downR (\Phi)$ and $\beta \Tleq
\smallvee \Phi'$. Now it follows by definition of $\covR_0$ that $
\beta \covR_0 \nbsem(\Phi')$.  Now
\begin{align*}
\heartsuit \beta &\leq \smallvee \{ \heartsuit \gamma \mid \gamma
\in \nbsem (\Phi')\}&&\text{since $ \beta \covR_0 \nbsem(\Phi')$,}\\
&\leq \smallvee \{ \heartsuit \gamma \mid \gamma
\in \nbsem (\T \downR (\Phi)) \}&&\text{by
L.~\ref{RelationLifting:LiftedInclusion} since $\Phi' \rel{\Tb {\subseteq}} \T\downR (\Phi)$.}
\end{align*}
Since $\beta \in \T \RL$ was arbitrary it follows that
\eqref{Compactness:VTtoVTR:C4} holds; consequently,
\eqref{Compactness:VTtoVTR:C3} holds so that we may indeed conclude
that $f$ transforms covers into joins. We conclude that $f\colon \T L
\to \VTR \bbL$ extends to a suplattice homomorphism $f' \colon \VT
\bbL \to \VTR \bbL$.
\end{proof}

Now that we have established the existence of suplattice homomorphisms
$f' \colon \VT \bbL \to \VTR\bbL$ and $g' \colon \VTR \bbL \to \VT
\bbL$, we are ready to prove the theorem of this subsection.

\begin{theorem} \label{Compactness:CompRegPreserved}
Let $\T \colon \Set \to \Set$ be a standard, finitary, weak
pullback-preserving set functor which maps finite sets to finite sets
and let $\bbL$ be a frame. If $\bbL$ is compact and zero-dimensional then so
is $\VT \bbL$.
\end{theorem}

\begin{proof}
It follows by Theorem \ref{t:reg} that $\VT \bbL$ is zero-dimensional. To
show that $\VT \bbL$ is compact, it suffices to show that $\VT \bbL
\simeq \VTR \bbL$ by Lemma \ref{Compactness:VTRCompact}. We will
establish that $\VT \bbL \simeq \VTR \bbL$ by showing that $g'
\colon \VTR \bbL \to \VT \bbL$ and $f' \colon \VT \bbL \to \VTR
\bbL$ are \emph{suplattice} isomorphisms, because $g' \cof f' =
\id_{\VT \bbL}$ and $f' \cof g' = \id_{\VTR \bbL}$. This is
sufficient since by order theory, any suplattice isomorphism is also
a frame isomorphism. We begin by making the following claim:
\begin{equation} \label{Compactness:CompRegPreserved:C1}
\forall \alpha \in \T L,\, g' \cof f(\alpha) = \nabla \alpha.
\end{equation}
After all, if $\alpha \in \T L$ then
\begin{align*}
g' \cof f(\alpha)
&= g' \left( \smallvee \{ \heartsuit \beta \mid  \beta \in \T \RL,\, \beta \rel{\Tb
{\leq}} \alpha \} \right) &&\text{by definition of $f$,}\\
&= \smallvee \{g'(\heartsuit \beta) \mid  \beta \in \T \RL,\, \beta \rel{\Tb
{\leq}} \alpha \}&&\text{since $g'$ preserves $\smallvee$,}\\
&=\smallvee \{ g(\beta) \mid \beta \in \T \RL,\, \beta \Tleq
\alpha\}&&\text{by Lemma \ref{Compactness:VTRtoVT},}\\
&= \smallvee \{ \nabla \beta \mid \beta \in \T \RL,\, \beta \Tleq
\alpha\}&&\text{by definition of $g$,}\\
&= \nabla \alpha &&\text{by Lemma \ref{Compactness:DownRProp}(1).}
\end{align*}
It follows that \eqref{Compactness:CompRegPreserved:C1} holds.
Conversely, we claim that
\begin{equation} \label{Compactness:CompRegPreserved:C2}
\forall \alpha \in \T \RL,\, f' \cof g(\alpha) = \heartsuit \alpha.
\end{equation}
This is also not hard to see. Take $\alpha \in \T \RL$, then
\begin{align*}
f' \cof g(\alpha)
&=f' \left( \nabla \alpha \right) &&\text{by definition of $g$,}\\
&= f(\alpha) &&\text{by Lemma \ref{Compactness:VTtoVTR},}\\
&=\smallvee \{ \heartsuit \beta \mid \beta \in \T \RL,\, \beta \rel{\Tb
{\leq}} \alpha \} &&\text{by definition of $f$,}\\
&= \heartsuit \alpha &&\text{since $\alpha \in \T \RL$ and
$\heartsuit$ is order-preserving.}
\end{align*}
It follows that \eqref{Compactness:CompRegPreserved:C2} holds.
Now we see that for all $\alpha \in \T L$,
\begin{align*}
g' \cof f' (\nabla \alpha) &= g' \cof f(\alpha) &&\text{since $f' \cof
\nabla = f$,}\\
&= \nabla \alpha &&\text{by
\eqref{Compactness:CompRegPreserved:C1},}\\
&= \id_{\VT \bbL} \left( \nabla \alpha \right).
\end{align*}
In other words, we see that $g'\cof f'$ and $\id_{\VT \bbL}$ agree on
the generators of $\VT \bbL$; it follows that $g' \cof f' = \id_{\VT
\bbL}$. An analogous argument shows that $f' \cof g' = \id_{\VTR
\bbL}$. We conclude that $\VT \bbL$ and $\VTR \bbL$ are isomorphic
as suplattices and consequently also as frames; it follows that $\VT
\bbL$ is compact.
\end{proof}

\section{Future work}
\label{s:fw}

To finish off the paper, we list some open problems and directions for
future work.

\subsection*{Preservation properties}
The main technical problems that we would like to solve concern further
preservation properties of our construction.
In particular, we are very eager to find out for which functors $\T$ the
$\T$-power construction preserves compactness, or the combination of
compactness and regularity.
Observe that any functor satisfying this property must map finite sets to
finite sets; if $\T A$ would be infinite for some finite $A$ subset of $\bbL$,
then we may have $1_{\VT\bbL} = \bv \{ \nb \al \mid \al \in A \}$, without
there being a finite subcover.
We conjecture that this condition (that is, of $\T$ restricting to finite
sets) is in fact not only necessary, but also sufficient to prove the
preservation of compactness.

\subsection*{Functorial properties}
In section~\ref{Section:VT:functorial} we saw that certain natural
transformations $\rho\colon  \T' \to \T$ induce natural transformations $\wh{\rho}\colon
\VT \to \V_{\T'}$, with the unit of the Vietoris comonad $\VP$ providing an
instance of this phenomenon.
There are some natural open questions related to this.
In particular, we are interested whether, in the case that $\T$ is actually
a \emph{monad}, it holds that $\VT$ is a co-monad.

Another question related to the natural transformation $\wh{\rho}$ is whether
$\wh{\rho}_{\bbL}\colon  \VT\bbL \to \V_{\T'}\bbL$ always has a right adjoint, see
Remark~\ref{Remark:RhoAdjoint}.

\subsection*{Spatiality and compact Hausdorff spaces}
Palmigiano \& Venema \citep{PV2007} introduce a lifting construction on Chu
spaces to prove that for Stone spaces, the Vietoris construction can be
generalized from the power set case to an arbitrary set functor $\T$
(meeting the same constraints as in the current paper).
Can we generalize to arbitrary topological spaces, or at least to compact
Hausdorff spaces?

Assume that, for any functor $\T$ mapping finite sets to finite sets, we
can prove that our $\T$-powerlocale construction $\VT$ preserves the
combination of compactness and regularity.
Then, using the well-known duality between compact regular locales and compact
Hausdorff spaces, we obtain a Vietoris-like functor on compact Hausdorff
spaces for free.
The question is then whether we can give a more direct, insightful
description of this functor.

\subsection*{Locales and constructivity}
In this paper, we have mostly adopted a frame- rather than a
locale-oriented perspective. Theorem \ref{th:nat-trans} suggests
however, that if one wants to understand the relationship between
coalgebra functors $\T\colon \Set \to \Set$ and the $\VT$
construction, one should think of $\VT$ as a functor on
\emph{locales}, since natural transformations $\T' \to \T$ satisfying
the conditions of Th.\  \ref{th:nat-trans}  correspond
to frame natural transformations $\VT \to \V_{\T'}$. It would be interesting to pursue this idea further,
especially in conjunction with the use of \emph{constructive
mathematics}. We have seen that certain constructive techniques, such as frame, flat site and preframe presentations, can
be brought over to the framework of coalgebraic logic. Making the
entire approach of this paper constructive would be a lot of work; we
believe however that this would be a promising line of further research.

\end{document}